\documentclass[
 10pt, aps, prx, twocolumn,
 amsfonts,amsmath,amssymb, 
 superscriptaddress,
 floatfix,nofootinbib,
 longbibliography
]{revtex4-2}

\usepackage{xcolor}
\usepackage{amsthm}
\usepackage{thmtools}
\usepackage{braket}
\usepackage{graphicx}
\usepackage{booktabs}
\usepackage{mathtools}
\usepackage[ruled,vlined,linesnumbered]{algorithm2e}
\usepackage[hypertexnames=false]{hyperref}
\usepackage[capitalize,nameinlink]{cleveref}
\usepackage{rotating}
\usepackage{orcidlink}
\usepackage{tikz}
\usepackage{adjustbox}
\usepackage{microtype}
\usepackage{makecell}
\usetikzlibrary{fit, matrix, arrows.meta, positioning, calc}

\IncMargin{1.5em}

\graphicspath{{./figures/}}

\declaretheorem{theorem}
\declaretheorem[sibling=theorem]{proposition}
\declaretheorem[sibling=theorem]{lemma}
\declaretheorem[sibling=theorem]{assumption}
\declaretheorem[parent=theorem]{corollary}
\declaretheorem[style=remark, parent=theorem]{remark}
\declaretheorem[style=definition]{definition}

\definecolor{red4}{HTML}{B82727}
\definecolor{azure4}{HTML}{2263A3}
\definecolor{blue4}{HTML}{4C4CD9}
\definecolor{purple4}{HTML}{B3256C}

\crefname{section}{Sec.}{Secs.}

\DeclareMathOperator{\tr}{Tr}
\DeclareMathOperator{\poly}{poly}

\newcommand{\ketbra}[2]{\vert{#1}\rangle\langle{#2}\vert}

\renewcommand{\vec}[1]{\boldsymbol{#1}}
\newcommand{\tel}{\mathrm{el}}
\newcommand{\Del}{D^{\tel}_{n,j}}

\newcommand{\dent}[1]{%
  \hspace{#1em}%
}

\makeatletter
\newcommand{\apptocfile}{atoc}
\let\apptoc@orig@appendix\appendix
\renewcommand{\appendix}{%
  \apptoc@orig@appendix
  \let\apptoc@orig@addtocontents\addtocontents
  \long\def\addtocontents##1##2{%
    \def\apptoc@ext{##1}%
    \def\apptoc@toc{toc}%
    \ifx\apptoc@ext\apptoc@toc
      \apptoc@orig@addtocontents{\apptocfile}{##2}%
    \else
      \apptoc@orig@addtocontents{##1}{##2}%
    \fi
  }%
}
\newcommand{\appendixtableofcontents}{%
  \begingroup
    \setcounter{tocdepth}{3}%
    \phantomsection
    \let\addcontentsline\@gobblethree
    \section*{Contents}%
    \pdfbookmark[1]{Appendices}{apxcontents}%
    \@starttoc{\apptocfile}%
  \endgroup
}
\makeatother

\newcommand{\QMT}{\affiliation{Quantum Motion, 9 Sterling Way, London N7 9HJ, United Kingdom}}
\newcommand{\OxMaterials}{\affiliation{Department of Materials, University of Oxford, Parks Road, Oxford OX1 3PH, United Kingdom}}
\newcommand{\OxMaths}{\affiliation{Mathematical Institute, University of Oxford, Woodstock Road, Oxford OX2 6GG, United Kingdom}}
\newcommand{\OxPhysChem}{\affiliation{Physical and Theoretical Chemistry Laboratory, University of Oxford, South Parks Road, Oxford, OX1 3QZ, United Kingdom}}
\newcommand{\UofT}{\affiliation{Department of Computer Science, University of Toronto, Canada}}
\newcommand{\UofTP}{\affiliation{Department of Physics, University of Toronto, Canada}}
\newcommand{\PNNL}{\affiliation{Pacific Northwest National Laboratory, Richland WA, USA}}
\newcommand{\CIFAR}{\affiliation{Canadian Institute for Advanced Studies, Toronto, Canada}}
\newcommand{\BI}{\affiliation{Quantum Lab, Boehringer Ingelheim, 55218 Ingelheim am Rhein, Germany}}
\newcommand{\BIBibMedChem}{\affiliation{Medicinal Chemistry, Boehringer Ingelheim Pharma GmbH \& Co KG, 88397 Biberach, Germany}}

\mathchardef\mhyphen="2D
\allowdisplaybreaks
\hypersetup{
    colorlinks,
    citecolor=blue4,
    linkcolor=red4,
    urlcolor=purple4,
    breaklinks=true,
}

\begin{document}
\title{Fullqubit alchemist: Quantum algorithm for alchemical free energy calculations}

\author{Po-Wei Huang\orcidlink{0009-0009-6973-5009}}
\email{po-wei.huang@maths.ox.ac.uk}
\QMT
\OxMaths

\author{Gregory Boyd\orcidlink{0000-0001-7822-1688}}
\QMT

\author{Gian-Luca R.~Anselmetti\orcidlink{0000-0002-8073-3567}}
\BI

\author{Matthias Degroote\orcidlink{0000-0002-8850-7708}}
\BI

\author{Nikolaj Moll\orcidlink{0000-0001-5645-4667}}
\BI

\author{Raffaele Santagati\orcidlink{0000-0001-9645-0580}}
\BI

\author{Michael Streif\orcidlink{0000-0002-7509-4748}}
\BI

\author{Benjamin Ries\orcidlink{0000-0002-0945-8304}}
\BIBibMedChem

\author{Daniel Marti-Dafcik\orcidlink{0000-0002-2871-238X}}
\QMT
\OxPhysChem

\author{Hamza Jnane\orcidlink{0000-0002-0713-3354}}
\QMT
\OxMaterials

\author{Sophia Simon\orcidlink{0000-0003-3261-2338}}
\UofTP

\author{Nathan Wiebe\orcidlink{0000-0001-6047-0547}}
\UofT
\PNNL
\CIFAR

\author{Thomas R. Bromley\orcidlink{0000-0002-7480-7478}}
\email{thomas@quantummotion.tech}
\QMT

\author{B\'alint Koczor\orcidlink{0000-0002-4319-6870}}
\email{balint.koczor@maths.ox.ac.uk}
\QMT
\OxMaths

\date{\today}
\begin{abstract}
    Accurately computing the free energies of biological processes is a cornerstone of computer-aided drug design, but it is a daunting task. The need to sample vast conformational spaces and account for entropic contributions makes the estimation of binding free energies very expensive. While classical methods, such as thermodynamic integration and alchemical free energy calculations, have significantly contributed to reducing computational costs, they still face limitations in terms of efficiency and scalability. We tackle this through a quantum algorithm for the estimation of free energy differences by adapting the existing Liouvillian approach and introducing several key algorithmic improvements. We directly implement the Liouvillian operator and provide an efficient description of electronic forces acting on both nuclear and electronic particles on the quantum ground state potential energy surface. This leads to super-polynomial runtime scaling improvements in the precision of our Liouvillian simulation approach and quadratic improvements in the scaling with the number of particles relative to prior quantum algorithms. Second, our algorithm calculates free energy differences via a fully quantum implementation of thermodynamic integration and alchemy, thereby foregoing expensive entropy estimation subroutines used in prior works. Our results open new avenues towards the application of quantum computers in drug discovery.
\end{abstract}

\maketitle
\begingroup
\renewcommand\thefootnote{}
\footnotetext{\vspace*{0.25ex}

Published in \emph{npj Quantum Information} (2026), DOI: \href{https://doi.org/10.1038/s41534-026-01275-2}{10.1038/s41534-026-01275-2}. The present arXiv version reformats and reorganizes sections for readability; the technical content is otherwise unchanged.}
\endgroup

\section{Introduction}

Biological processes occur in warm, aqueous environments where molecules are in constant thermal motion. Rather than adopting a single geometry, molecules in the human body continuously explore a large ensemble of configurations. As a result, thermodynamic properties such as entropy or free energy depend on the statistical distribution of a vast number of accessible conformations. Accurately estimating these properties \textit{in silico} requires sampling this high-dimensional configuration space, which is very computationally intensive.

Modern computer-aided drug design pipelines rely heavily on the determination of thermodynamic quantities. In particular, the binding free energy between a ligand and its biological target is a key metric for ranking drug candidates~\citep{Chodera2011_AlchemicalDrug}. Significant advancements in this field have been driven by two key factors: the exponential growth in computational power---consistent with Moore's law and the rise of parallel computing architectures---and the development of sophisticated algorithms and various protocols for free energy estimation~\citep{DeRuiter2021_ThermodynamicIntegration}. 

Free energy calculations focus primarily on calculating relative binding free energies, which quantify the change in free energy between two different molecular ensembles~\citep{frenkel_understanding_1996, kastner_qmmm_2006, Cournia2017DeltaDeltaG}. Thus, in such calculations, instead of evaluating two absolute binding free energies independently, one computes their difference directly. Alchemical free energy methods accomplish this by introducing a continuous, nonphysical transformation from one chemical species into another through a pathway of fictitious intermediary states~\citep{SHIRTS2007AlchemFreeEnergyCalc}, drawing parallels with transforming elements in the art of alchemy. The main benefit is that because free energy is a state function, the actual physical binding process does not need to be simulated. This flexibility avoids the need to model complex kinetic mechanisms and enables more efficient estimation of free energy differences. Furthermore, for condensed-phase systems such as liquids and solids---typical in drug design---the Gibbs free energy difference can be accurately approximated by the Helmholtz free energy, as the pressure-volume work is negligible under the nearly incompressible conditions of these phases~\citep{mey2020best,dejong2011determining}. \begin{figure*}
    \begin{adjustbox}{width=\linewidth}
    \begin{tikzpicture}[
    every node/.style={
        draw, rounded corners, align=flush center,
        text width=6.5em, minimum height=3em,
        anchor=center
    },
    invisible/.style={draw=none, fill=none},
    thickarrow/.style={thick, -{Stealth[scale=1.1]}},
    ]
    \matrix[matrix of nodes, row sep=1em, column sep=0, draw=none] (m) {
    \node[fill=azure4!15, name=ehb] {Electronic Hamiltonian \citep{su2021faulttolerant}}; &
    \node[invisible, anchor=north] {Ground State Preparation \citep{lin2020nearoptimal}}; &
    \node[fill=azure4!15, name=gs] {Ground State}; &&
    \node[fill=azure4!15, name=pdo] {Momentum Derivative Operator \citep{simon2024improved}}; &&
    \node[fill=azure4!15, name=clb] {Classical Liouvillian \citep{simon2024improved}}; &
    \node[invisible, name=3w]{};&&\node[invisible, name=3e]{};
    \\
    &
    \node[fill=azure4!15, name=fobe] {Force Operator \citep{obrien2022efficient}}; &
    \node[invisible, anchor=south] {Phase Kickback \citep{cleve1998quantum}};&
    \node[fill=red4!15, name=gsedb] {Ground State Energy Derivative \cref{eqDelInfml}}; &
    \node[invisible, anchor=south] {Tensor Prod.\\ + LCU \citep{childs2012hamiltonian}}; &
    \node[fill=red4!15, name=elb] {Electronic Liouvillian \cref{eqElecLiouInfml}}; &
    \node[invisible, anchor=south] {LCU \citep{childs2012hamiltonian}}; &
    \node[fill=red4!15, name=lb] {Liouvillian \cref{propLiouInfml}}; &
    \node[invisible, anchor=south] {Hamiltonian Simulation \citep{low2017optimal,low2019hamiltonian,gilyen2019quantum}}; &
    \node[fill=red4!15, name=leb] {Liouvillian Evolution \cref{thmLiouSimInfml}};
    \\
    \node[invisible, name=nw]{};&&\node[invisible, anchor=north] {Ground State Preparation \citep{lin2020nearoptimal}};&
    \node[fill=azure4!15, name=gsr] {Ground State}; &
    \node[invisible, align=left, anchor=north] {Phase Kick-back \citep{cleve1998quantum}};&
    \node[fill=red4!15, name=gseb] {Ground State Energy Hamiltonian \cref{eqGSEInfml}}; &
    \node[invisible, anchor=north] {LCU \citep{childs2012hamiltonian}}; &
    \node[fill=red4!15, name=ieb] {Nuclear Hamiltonian \cref{propInHamInfml}}; &
    \node[invisible, anchor=south] {Hadamard Test \citep{kitaev1995quantum,cleve1998quantum}}; &
    \node[fill=red4!15, name=febe] {Free Energy \cref{thmThermIntInfml}};
    \\
    \node[invisible, name=sw]{};&\node[invisible, name=se]{}; &
    \node[fill=azure4!15, name=ehbr] {Electronic Hamiltonian \citep{su2021faulttolerant}}; &&&&
    \node[fill=azure4!15, name=kpeb] {Kinetic/ Potential Energy Hamiltonian \citep{simon2024improved}}; &
    \node[invisible, name=4w]{};&&\node[invisible, name=4e]{};
    \\
    };
    
    \draw[thickarrow] (ehb) -- (gs);
    \draw[thickarrow] (gs) -- (gsedb);
    \draw[thickarrow] (fobe) -- (gsedb);
    \draw[thickarrow] (gsedb) -- (elb);
    \draw[thickarrow] (pdo) -- (elb);
    \draw[thickarrow] (clb) -- (lb);
    \draw[thickarrow] (elb) -- (lb);
    \draw[thickarrow] (lb) -- (leb);
    \draw[thickarrow] (ehbr) -- (gsr);
    \draw[thickarrow] (gsr) -- (gseb);
    \draw[thickarrow] (ehbr) -- (gseb);
    \draw[thickarrow] (kpeb) -- (ieb);
    \draw[thickarrow] (gseb) -- (ieb);
    \draw[thickarrow] (ieb) -- (febe);
    \draw[thickarrow] (leb) -- (febe);
    \draw[rounded corners] ($(ehb.north west)+(-1em,0.5em)$) rectangle ($(leb.south east)+(1em,-1.3em)$);
    \draw[rounded corners] ($(sw.south west)+(-1em,-2em)$) rectangle ($(febe.north east)+(1em,1.5em)$);
    \node[invisible, fit=(3w)(3e), align=center] {\\[1ex]\large\textbf{Liouvillian Simulation\\\vspace{1ex} \cref{secLS}}};
    \node[invisible, fit=(4w)(4e), align=center,anchor=north] {\large\textbf{Free Energy Calculation\\\vspace{1ex} \cref{secTI}\vspace{3ex}}};
    \end{tikzpicture}
\end{adjustbox}
\caption[Overview of our algorithm construction.]{\emph{Overview of our algorithm construction.} The results in the top panel of the figure regarding the efficient construction of a Liouvillian simulation quantum algorithm are featured in \cref{secLS}. In addition, the results in the bottom panel of the figure for the design of a quantum algorithm for alchemical free energy calculation are featured in \cref{secTI}. In this figure, we color our technical contributions in red, while those from existing work are colored in blue. For the Liouvillian simulation algorithm, we use the Hellmann--Feynman theorem~\citep{hellmann1937einfuhrung,feynman1939forces} to efficiently construct a diagonal block encoding of the ground state energy derivative for each point in phase-space, which is then used to construct the electronic and full Liouvillians. Subsequently, using Hamiltonian simulation~\citep{low2017optimal,low2019hamiltonian,gilyen2019quantum}, we can implement a Liouvillian evolution operator with logarithmic dependency on precision. Integrating these results with an efficient block encoding of the nuclear Hamiltonian, we can find the free energy difference between two systems through implementing the thermodynamic integration~\citep{frenkel_understanding_1996} via the Hadamard test~\citep{kitaev1995quantum,cleve1998quantum}.}
\label{figMain}
\end{figure*}
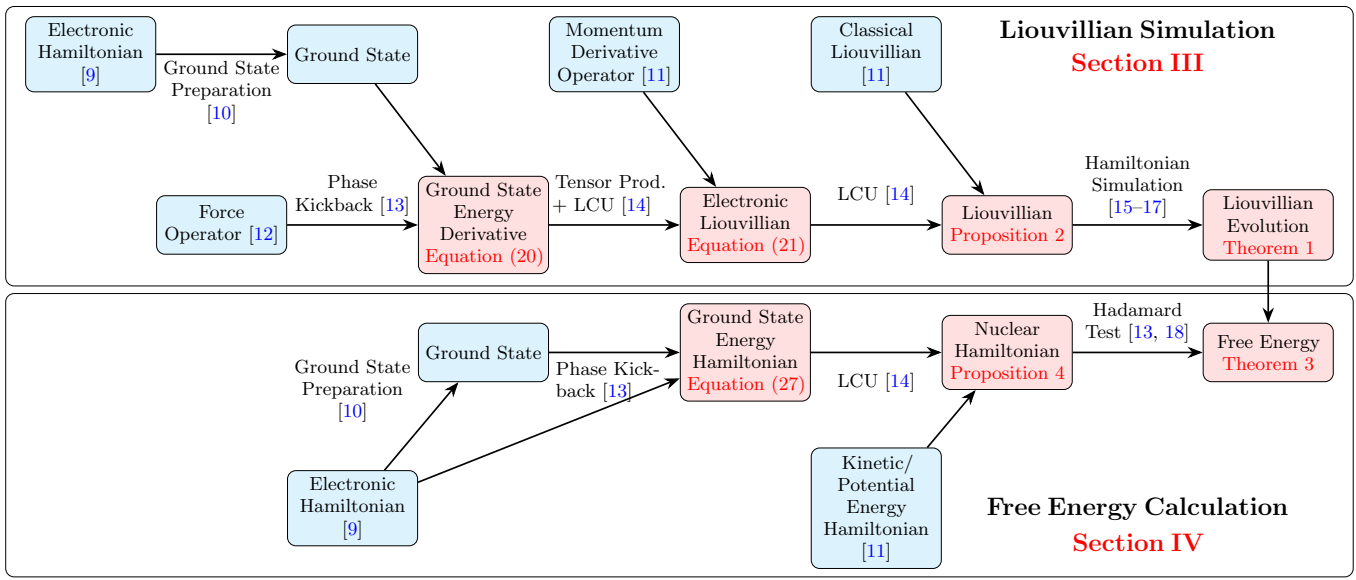
Among the most widely used techniques to estimate the Gibbs free energy differences via Helmholtz free energy differences are thermodynamic integration~\citep{Kirkwood1935_StatistiscalMechanics} and free energy perturbation~\citep{Zwanzig1955FreeEnPerturbation, BENNETTFreeEnDifference}, both of which rely on sampling equilibrium ensembles on classical computers.

Among the available techniques for sampling the configuration space, molecular dynamics is one of the most widely used~\citep{allen1987computer,frenkel_understanding_1996}. It is routinely employed in academic and industrial settings to compute thermodynamic properties from first principles~\citep{Williams2018_FreeEnergyInDrugDesign}. However, the efficiency of free energy calculations remains a considerable challenge due to the inherent complexity of accurately capturing entropic contributions to the free energy. Reliable sampling of conformational spaces, particularly in systems with high entropic barriers, requires substantial computational resources and robust statistical mechanical frameworks~\citep{Chodera2011_AlchemicalDrug, Hansen2014PracticalFreeEnergy}.

To address some of these computational bottlenecks, the efficient simulation of molecular and quantum systems on quantum computers has emerged as a promising avenue~\citep{feynman1982simulating,lloyd1996universal, Caesura2025Faster, Low2025SpectrumAmplif}. Quantum simulation techniques have already shown significant improvements in runtime for calculating ground state energies of systems of correlated electrons, with potential applications spanning catalysis, materials science, and drug design~\citep{reiher_elucidating_2017, McArdle2020QChem, von_burg_quantum_2021, Outeiral2021ProspectQCMolBio,rubin_fault-tolerant_2023, Santagati2024DrugDesignOnQC}.

Molecular dynamics requires simulating both electronic and nuclear degrees of freedom. In many practical cases, including most drug-like molecules, it is appropriate to treat the nuclei classically while modeling the electrons quantum mechanically, with the nuclei evolving under forces generated by the electrons~\citep{Born1927Zur}. Recognizing that the main bottleneck in classical computing lies in the electronic structure calculations, earlier work proposed a hybrid computational paradigm: use a quantum computer to compute electronic forces and then propagate the nuclei according to Newton’s equations of motion on a classical computer~\citep{obrien2022efficient, steudtner2023faulttolerant, gunther2025how}. While this is a natural way of enhancing existing classical simulation pipelines with quantum computers, the approach is fundamentally limited by measurement overhead as each electronic force must be extracted from the quantum computer via repeated measurements, a process that can take hours to days and must be repeated at each time step~\citep{Santagati2024DrugDesignOnQC}.

\begin{table*}
    \centering
    \caption[Overview of results on algorithmic complexity in terms of Toffoli gate count.]{\emph{Overview of results on algorithmic complexity in terms of Toffoli gate count.} Here we consider the Toffoli gate count in terms of $N$ (the number of nuclei), $\widetilde N$ (the number of electrons), $N_{\mathrm{tot}} = N + \widetilde N$ (the total number of particles), $t$ (the evolution time), $\varepsilon$ (the precision of the state preparation/free energy estimation), $\delta$ (a lower bound for the overlap between the approximate ground state prepared by an oracle $U_I$ and the true ground state), $\gamma$ (a lower bound on the spectral gap of the electronic Hamiltonian), $\eta$ (the number of grid points in phase-space, which is exponential in $N$), and $1-\xi$ (the success probability for free energy estimation). Throughout this paper, we use the notation $\widetilde{\mathcal{O}}(\cdot)$ to hide polylogarithmic factors. For simplicity in this table, we suppress the dependence on all other variables for simplicity and use the notation $\mathcal O^*(\cdot)$ to hide both polylogarithmic dependencies and $o(1)$ terms in the exponent. More details can be found in \cref{appLS}. Note that in this work we compute the difference of free energies between two systems while \citet{simon2024improved} computes the absolute free energy of a single system.}
    \label{tabMain}
    \begin{tabular}{l@{\hskip 3em}c@{\hskip 3em}c}
        \toprule
        & Molecular dynamics & Free energy calculation\\
        \midrule
        Liouvillian + Trotterization~\citep{simon2024improved} & $\displaystyle\mathcal{O}^*\left(\frac{N^2 \widetilde N^2 N_{\mathrm{tot}}^3  t}{\delta\gamma\varepsilon^{o(1)}}\right)$ & $\displaystyle\mathcal{O}^*\left( \left( \frac{\eta^{o(1)} N^2 \widetilde N^2 N_{\mathrm{tot}}^3 t}{\delta\gamma \, \varepsilon} \left( N_{\mathrm{tot}}^2 + \frac{\eta}{\sqrt{\varepsilon}} \right) + \frac{\widetilde N N_{\text{tot}}^4 }{\varepsilon^2}\right) \log \frac{1}{\xi}\right)$\\[3ex]
        Our work & $\displaystyle\widetilde{\mathcal{O}}\left(\frac{N \widetilde N N_{\mathrm{tot}}^3 t}{\delta\gamma}\log^3\frac{1}{\varepsilon}\right)$ & $\displaystyle\widetilde{\mathcal{O}}\left(\frac{N \widetilde N N_{\mathrm{tot}}^5 t}{\delta\gamma\varepsilon}\log\frac{1}{\xi}\right)$\\
        \bottomrule
    \end{tabular}
\end{table*}

Recent advances extend the runtime advantages of quantum algorithms to hybrid classical–quantum dynamics, enabling a fundamental shift from the paradigm described above~\citep{simon2024improved,gonzalez-conde2023mixed}. Notably, \citet{simon2024improved} introduced an approach based on the Liouvillian formalism~\citep{liouville1838theorie} that allows the entire molecular dynamics simulation to be carried out on a quantum computer. By directly applying electronic forces to the classical nuclei via the phase kickback technique~\citep{cleve1998quantum}, it eliminates the need for repeated force measurements. When combined with quantum algorithms for estimating thermodynamic quantities~\citep{simon2024improved}, this framework enables the computation of free energies over exponentially large configuration spaces without relying on any classical-quantum feedback loop at every time step.

Nonetheless, these quantum algorithms still face major bottlenecks in computational efficiency. In particular, the quantum simulation method presented in Ref.~\citep{simon2024improved} relies on Trotterization~\citep{suzuki1976generalized,childs2021theory}, which has a less favorable scaling with respect to error tolerance and evolution time than modern techniques based on block encodings~\citep{berry2015simulating,chakraborty2019power, gilyen2019quantum, low2018hamiltonian}. Furthermore, the technique lacks an efficient method for calculating the differences in free energies when comparing different compounds, which is critical for end-to-end applications such as drug discovery.

In this work, we introduce a new quantum algorithm for molecular dynamics that overcomes these limitations. The algorithm relies on an improved version of the original implementation of molecular dynamics in the Liouvillian formalism presented in Ref.~\citep{simon2024improved}. The latter relied on the Suzuki-Trotter product formula~\citep{trotter1959product,suzuki1976generalized,suzuki1991general} to implement the calculation of the electronic structure and to combine the electronic and classical Liouvillians. Here, instead, we introduce a block encoding of the electronic Liouvillian, preparing ground state energy forces on a separate electronic register, and then ``kicking-back'' the computed values onto the nuclear registers via a form of phase kickback~\citep{cleve1998quantum}. This enables us to implement the quantum simulation with the full (electronic and classical) Liouvillian by Hamiltonian simulation algorithms~\citep{low2017optimal,low2019hamiltonian,gilyen2019quantum}. In doing so, we achieve a super-polynomial improvement in runtime scaling with respect to precision and substantially decrease the dependency on the number of particles.

We then develop a new quantum algorithm for thermodynamic integration that runs solely on the quantum computer, building on these improvements. This algorithm avoids the need for a direct entropy estimation, which can scale exponentially with the number of atoms in the simulated system~\citep{gilyen2020distributional,simon2024improved}. It replaces it with internal energy estimation across interpolated molecular dynamic simulations, which can be implemented efficiently on a quantum computer and be further integrated with amplitude estimation~\citep{brassard2002quantum} to provide additional speedups over classical methods.

\begin{figure}
    \includegraphics[width=\linewidth]{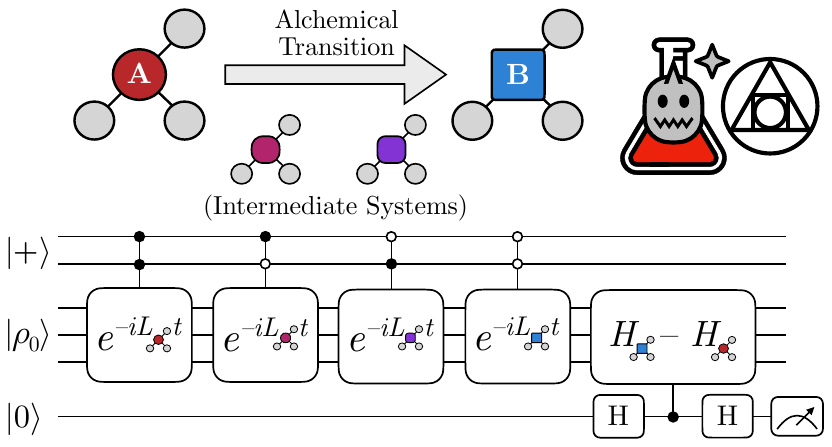}
    \caption[Alchemical free energy calculation on a quantum computer.]{\emph{Alchemical free energy calculation on a quantum computer.} Given systems $A$ and $B$, we perform Liouvillian simulations of the two systems and their interpolated intermediate systems in superposition. Note that while the operators are shown separately in this illustration, we implement them as one single operator via controlled weighting of LCU components without the multiplicative cost of executing separate simulations for each intermediate system. We then compute the internal energy difference and sum over the different results via the Hadamard test to provide the free energy difference.}
    \label{figConcept}
\end{figure}

\cref{figMain} provides an overview of our quantum algorithms: the top two rows detail the Liouvillian simulation. In addition, the bottom two rows present the computation of free energy via alchemical methods. Using the Hellmann–Feynman theorem~\citep{hellmann1937einfuhrung,feynman1939forces}, we construct diagonal block encodings of energy derivatives to build the Liouvillian operator, which is then used to implement its corresponding evolution operator. Combined with a block encoding of the nuclear Hamiltonian, this enables the computation of free energy differences by implementing thermodynamic integration via the Hadamard test~\citep{kitaev1995quantum,cleve1998quantum}. A sketch of the concept behind the implementation of thermodynamic integration on a quantum computer is provided in \cref{figConcept}. Note that the controlled Liouvillian is implemented as a single operator, rather than as separate operators. A comparison of algorithmic complexities in terms of Toffoli gate counts between our work and that of \citet{simon2024improved} is presented in \cref{tabMain}. For the Liouvillian simulation algorithm, we provide a super-polynomial speedup in terms of the precision $\varepsilon$ as well as a reduction in the particle numbers $N$ and $\widetilde N$ by a polynomial factor. On the other hand, for the free energy calculation, we remove the dependency on the number of grid points $\eta$, which can scale exponentially with the number of qubits, as well as reduce the dependency on precision $\varepsilon$ by a polynomial factor.

The paper is structured as follows. After the Preliminaries in \cref{secPreliminaries}, the improvement of the Liouvillian algorithm by implementing more efficient block encodings is presented in \cref{secLS}. In \cref{secTI}, we devise a quantum algorithm for thermodynamic integration, evolving molecular dynamics through Liouvillian simulation in superposition. Finally, we conclude in \cref{secconclusion}.

\section{Preliminaries}
\label{secPreliminaries}
\subsection{Molecular dynamics within the Liouvillian formalism}

Simulating the time evolution of a physical system is a crucial task in physics, with numerous applications, such as molecular dynamics (MD)~\citep{frenkel_understanding_1996}. MD simulations are often used to evaluate ensemble properties of a physical system, which can be used to estimate macroscopic quantities such as free energies. Different ensembles can be defined. For example, the microcanonical ensemble~\citep{gibbs1902elementary}, also known as NVE, is a statistical ensemble where the total number of particles (N), total volume (V), and the total energy (E) are held constant. However, for the calculation of chemically relevant properties, we require the canonical ensemble~\citep{gibbs1902elementary,nose1984molecular}, NVT, where the system is coupled to a heat bath to ensure a constant temperature (T). When coupled to the Nos\'e thermostat~\citep{nose1984molecular,nose1984unified}, an extra degree of freedom $s$ for the heat bath is introduced into the Hamiltonian. 

To construct the corresponding nuclear Hamiltonian that governs the dynamics, we write out the kinetic, Coulombic, and ground state energy terms as follows:
\begin{equation}
\label{eqNucHam}
    H_{\rm nuc} = \sum_{n=1}^N\sum_{j=1}^3 \frac{p_{n,j}^2}{2m_n} + \sum_{k=1}^{N-1}\sum_{n=k+1}^{N} \frac{Z_n Z_k}{\lVert x_n - x_k\rVert} + E_{\tel}(\vec x),
\end{equation}
Under NVT canonical simulations with the Nos\'e thermostat, the nuclear Hamiltonian has additional terms concerning an extra degree of freedom $s$ that represents the heat bath that equilibrates the system to a fixed temperature $T$. The extended Nos\'e nuclear Hamiltonian governing the dynamics can thus be formulated as follows~\citep{nose1984molecular, simon2024improved}:
\begin{multline}
\label{eqNoseHam}
    H_{\rm ext}^{\rm (NVT)} = \sum_{n=1}^N\sum_{j=1}^3 \frac{p_{n,j}'^2}{2m_ns^2} + \sum_{k=1}^{N-1}\sum_{n=k+1}^{N} \frac{Z_n Z_k}{\lVert x_n - x_k\rVert}\\+\frac{p_s^2}{2Q} + N_fk_BT\ln(s) + E_{\tel}(\vec x),
\end{multline}
where $m_n$ is the mass of the $n$-th nucleus, $Z_n$ is the atomic number of the $n$-th nucleus, $p_s$ is the momentum variable conjugate to $s$, $Q$ is an effective mass of $s$ that controls the coupling of the system to the heat bath, $k_B$ is the Boltzmann constant, and $N_f$ is the number of degrees of freedom of the system. In this extended system, $p_{n,j}'$ is the conjugate momentum variable to $x_{n,j}$ that is known as a ``virtual'' momentum variable where the real momentum $p_{n,j}$ can be obtained by $p_{n,j} = \frac{p_{n,j}'}{s}$. The last term $E_{\tel}$ is the adjustment of the Hamiltonian that adds the ground state electronic energy at a nuclear configuration $\boldsymbol x$, which is responsible for the electronic forces. Given this additional degree of freedom $s$ and the use of ``virtual momentum'', one needs to integrate over the heat bath variable to obtain the actual particle distribution after simulations. 

To calculate the ground state electronic energy $E_{\tel}$, we require assumptions on the ground state preparation and the separation of ground and excited state energies. For many chemical systems and most biologically relevant systems, where nuclear dynamics is much slower than electronic dynamics, the Born-Oppenheimer approximation can achieve reliable predictions of the system's properties~\cite{Born1927Zur, Tully2000Perspective, Bransden2003Physics}. Our work primarily operates under this assumption, which precludes the existence of conical intersections between the electronic ground state potential energy surface and the electronic excited state potential energy surface over different nuclear configurations. We further note that the particles in the simulated system in question would be correlated weakly enough such that the Hamiltonians remain sufficiently gapped, leading to a fair assumption that the spectral gap $\gamma$ is polynomially bounded. Under this assumption, to solve the electronic structure problem at ground state configurations, we assume that we can make use of an approximate ground state oracle $U_I$ that prepares an initial electronic state $\ket{\phi_{\rm init}(\vec x)}$ that has nontrivial overlap with the ground state $\ket{\psi_0(\vec x)}$ of the electronic Hamiltonian $H_{\tel}$ describing the dynamics of the electronic system. Specifically, given $0 < \delta \leq 1$, we demand that for any $\vec{\bar x}$,
\begin{equation}
    U_I \ket{\vec{\bar x}} \ket{0} = \ket{\vec{\bar x}} \ket{\phi_{\rm init} (\vec x)},
\end{equation}
where $\left| \braket{\psi_0(\vec x)|\phi_{\rm init}(\vec x)} \right| \ge \delta$ for all nuclear configurations. $\delta$ is thus a lower bound on the overlap of the prepared approximate state with the true electronic ground state. 

While we utilize the initial approximate ground state as generated from a black-box oracle $U_I$, to implement this in practice, we can borrow methods in classical computational chemistry that provide faster convergence in self-consistent field (SCF) procedures~\citep{lehtola2019assessment} that are nuclear-position independent. These methods include superposition of atomic densities (SAD)~\citep{vanlenthe2006starting} and its variants, as well as further refinements with the extended H\"{u}ckel method~\citep{hoffmann1963extended, lehtola2019assessment}. Such methods have appeared as wavefunction initialization techniques in computational chemistry packages such as QuantumESPRESSO~\citep{giannozzi2009quantum,giannozzi2017advanced} and VASP~\citep{kresse1993ab,kresse1996efficiency,kresse1996efficient}, and encode these results into the plane wave basis quantum state to provide an initial guess. Future work on quantizing the implementation of obtaining the Hartree-Fock state, such that it can be prepared in superposition, may further improve ground state preparation.

As in Ref.~\cite{simon2024improved}, our quantum algorithm operates within the Liouvillian formalism for simulating Koopman–von Neumann mechanics~\citep{koopman1931hamiltonian,neumann1932zur,neumann1932zusatze}. The Liouvillian formalism allows for the description of the time evolution of a physical system through the Liouville equation of motion:
\begin{equation}
    \frac{\partial \rho}{\partial t} = -i L \rho \, ,
\end{equation}
where $\rho(\vec x, \vec p,t)$ can be thought of as the probability density of a system of $N$ classical particles with position $\{\vec x_n\}_{n\in[N]}$ and momentum $\{\vec p_n\}_{n\in[N]}$, and where $L$ is the Liouvillian operator, defined as:
\begin{equation}\label{eq:Liouvillian}
    L:=-i\left((\nabla_{\vec p}H_{\rm ext})\cdot\nabla_{\vec x}-(\nabla_{\vec x}H_{\rm ext})\cdot\nabla_{\vec p}\right) \, ,
\end{equation}
where $H_{\rm ext}$ is the extended Nos\'e nuclear Hamiltonian underlying the dynamics of the nuclear system shown above, and is can be separated into the kinetic and potential Hamiltonians  $H_{\rm kin}$ and $H_{\rm pot}$ with additional ground state correction terms from the electronic system $H_{\rm gse}$ and extra terms regarding the heat bath $H_{\rm bath}$. The terms in the sum defining the Liouvillian operator depend directly on the derivatives of the nuclear Hamiltonian $H_{\rm ext}$ with respect to the $j$-th components of position and momentum of the $n$-th particle. To map the classical probability density on a quantum computer, we exploit the Koopman-von Neumann description of classical mechanics, which directly connects the probability density to a quantum state~\cite{Mauro2002KvN, simon2024improved}, by defining $\rho:= \psi^*_\mathrm{KvN} \psi_\mathrm{KvN}$.
Representing the KvN wavefunction $\psi_\mathrm{KvN}$ as a quantum state $\ket{\rho}$, the time evolution of the system can be computed as the evolution of the density under the Liouvillian:
\begin{equation}
    \ket{\rho_t} = e^{-iLt} \ket{\rho_0} \, .
\end{equation}
This can then be implemented on a quantum computer using Hamiltonian simulation~\cite{childs2012hamiltonian, berry2014exponential,low2017optimal,low2019hamiltonian, gilyen2019quantum}. 

Here, we are interested in MD simulations within the Born-Oppenheimer approximation. The force on the nuclei consists of nuclear and electronic contributions, which we can obtain by considering the two components of the Liouvillian
\begin{equation}
    L= L_\mathrm{cl}+L_\mathrm{el} \, ,
\end{equation}
where the classical part of the Liouvillian $L_\mathrm{cl}$ is given by the kinetic term and nucleus-nucleus repulsion, and $L_\mathrm{el}$ adds the quantum interaction between the nuclei and electrons and amongst the electrons, which takes the form of the ground state energy. The electronic Liouvillian can thus be formulated as
\begin{equation}
    L_{\tel}:=i\sum_{n=1}^{N}\sum_{j=1}^{3}\left(\frac{\partial E_\tel(\vec x)}{\partial x_{n,j}} \partial_{p_{n,j}}\right) \, ,
\end{equation}
where $E_\tel$ is the ground state energy. For simulations under the NVT canonical ensemble with the Nos\'e thermostat, an additional degree of freedom $s$ and its momentum conjugate $p_s$ are additionally considered in the Liouvillian simulation process, but integrated out to retrieve the final probability density $\rho(\vec x, \vec p, t)$. For implementation on a quantum computer, we prepare a set of registers $\ket{\bar x_{n,j}}$, $\ket{\bar p_{n,j}'}$, $\ket{\bar s}$, and $\ket{\bar p_s}$ such that operations are only carried out on the relevant registers, where the overbar stands for the discretized phase-space. The ``virtual'' momentum $\vec p'$ is a reparameterized version of the momentum dependent on both the true momentum $\vec p$ and the heat bath variable $s$. Details of the spatial discretization and NVT simulation are deferred to \cref{ap:preliminaries}.

Prior work by \citet{simon2024improved} has shown that a Liouvillian simulation algorithm can be achieved. They show how to efficiently block-encode the classical Liouvillian and implement the evolution with Hamiltonian simulation via quantum singular value transform (QSVT)~\citep{low2017optimal,low2019hamiltonian,gilyen2019quantum}. On the other hand, the direct implementation of electronic forces, and by extension, the electronic Liouvillian $L_{\tel}$, is circumvented by utilizing controlled Hamiltonian simulation~\citep{chakraborty2019power}, which allows for the implementation $e^{-iL_{\tel}t}$ in logarithmic dependency on error. However, since the evolution operator over the electronic Liouvillian $L_{\tel}$ is implemented directly instead of the full Liouvillian, Trotterization is required to approximate the evolution over the full Liouvillian $e^{-iLt}$. Their results provide a Liouvillian simulation algorithm with a runtime of
\begin{equation}
    \widetilde{\mathcal{O}}\left(\frac{N^{2+o(1)} \widetilde N^{2+o(1)} N_{\mathrm{tot}}^{3+o(1)} t^{1+o(1)}}{\delta\gamma\varepsilon^{o(1)}}\right),
\end{equation}
where we use $\widetilde{\mathcal{O}}(\cdot)$ to hide polylogarithmic factors.

In \cref{ap:preliminaries}, we formalize the construction of the electronic Liouvillian $L_{\tel}$ and further provide technical details required from implementation introduce in prior works, including the formulation and cost of implementing the electronic Hamiltonian $H_\tel$ from Refs.~\citep{babbush2019quantum,su2021faulttolerant} and the force operators from Ref.~\citep{obrien2022efficient} on quantum computers.

\subsection{Alchemical free energy calculation via thermodynamic integration}
\label{prelim:thermodynamic_integration}

Thermodynamic quantities such as free energy and entropy can be determined through ensemble averaging over configurations sampled in molecular dynamics or Monte Carlo simulations~\citep{BENNETTFreeEnDifference} using appropriate Boltzmann weighting~\citep{Zwanzig1955FreeEnPerturbation}. It is important to distinguish between the Helmholtz free energy, defined at constant volume $V$ as $F = U - TS$, where $U$ is the internal energy and $S$ is entropy, and the Gibbs free energy, defined at constant pressure $P$ as $G = U - TS + PV$ \cite{frenkel_understanding_1996}. Since experiments are typically carried out at constant pressure, the Gibbs ensemble provides a more realistic representation of physical conditions. Nevertheless, in the context of molecular simulations relevant to drug discovery in aqueous systems, the difference between Gibbs and Helmholtz free energies is negligible, i.e., $\Delta G \approx \Delta F$, because $P \Delta V$ is nearly zero for drug-like molecules in water, as they are effectively incompressible~\citep{mey2020best,dejong2011determining}. As molecular dynamics often employ constant-volume ensembles, the Helmholtz free energy is often used in computational chemistry. In the following, we will therefore refer to the free energy as $F$ throughout.

We now consider two closely related molecular systems, $A$ and $B$, with a large phase space overlap and with microstate energies $E_{\mathcal S}(A)$ and $E_{\mathcal S}(B)$ where ${\mathcal S}$ labels the microstates, corresponding to different molecular geometries, in the NVT canonical ensemble. These molecular systems differ by only one atom or functional group, resulting in significant overlap in nuclear phase-space. Computing the absolute free energy of both systems requires extensive sampling for reliable convergence, making the process computationally demanding. As an alternative, we can employ a method that focuses solely on calculating the free energy difference $\Delta F$ of systems $A$ and $B$, which reduces the required sampling and computational cost~\citep{DeRuiter2021_ThermodynamicIntegration, Kirkwood1935_StatistiscalMechanics, frenkel_understanding_1996, kastner_qmmm_2006, Cournia2017DeltaDeltaG}.

Consider a new system with microstate energy functions $E_{\mathcal S}(\Lambda)$ which are differentiable functions of a new coupling parameter $\Lambda \in [0,1]$ satisfying
\begin{align}
    E_{\mathcal S}(0) & = E_{\mathcal S}(A) \\
    E_{\mathcal S}(1) & = E_{\mathcal S}(B)  \, .
\end{align}
To obtain the free energy in an NVT canonical ensemble, we need to consider the partition function of the system:
\begin{equation}
    Z(N, V, T, \Lambda) = \sum_{{\mathcal S}} \exp (-E_{\mathcal S}(\Lambda)/k_{B}T) \, .
\end{equation}
The microstates $\mathcal S$ accessible to the system depend on the number of particles $N$, the volume $V$, and the temperature $T$. A sufficient number of states $\mathcal S$ must be chosen so that the partition function converges. Under Liouvillian simulation in the NVT canonical ensemble, we can directly obtain a phase-space density over such states. The Helmholtz free energy of this system is given by
\begin{equation}
    F(N,V,T,\Lambda)=-k_{B}T \ln Z(N,V,T,\Lambda) \, .
\end{equation}
Helmholtz free energy differences are important because they approximate Gibbs free energy differences in aqueous solutions.

However, when solely calculating the free energy difference $\Delta F$, the explicit calculation of the partition function can be avoided. The integral over the derivative of the free energy gives the free energy difference from $A$ to $B$,
\begin{align}
    \Delta F_{A \rightarrow B}
    &= \int_0^1 \frac{\partial F(\Lambda)}{\partial\Lambda} \mathrm d\Lambda \nonumber\\
    &= -k_BT \int_0^1 \frac{1}{Z}  \frac{\partial Z}{\partial \Lambda} \mathrm d\Lambda\nonumber\\
    &= \int_0^1 \frac{1}{Z}  \sum_{\mathcal S} \frac{\partial E_{\mathcal S}(\Lambda)}{\partial \Lambda} e^{-E_{\mathcal S}(\Lambda)/k_B T} \mathrm d\Lambda \nonumber\\
    &= \int_0^1 \left\langle\frac{\partial E(\Lambda)}{\partial\Lambda}\right\rangle_{\Lambda} \mathrm d\Lambda \, ,
\end{align}
where the angular brackets denote an ensemble-averaged derivative of the microstate energies with respect to the coupling parameter $\Lambda$ over the thermal distribution of the states $\mathcal S$ for a particular coupling parameter $\Lambda$. As noted, the distribution is provided by equilibration via the NVT canonical ensemble, so the partition function need not be further computed. Thus, under this scheme, rather than computing the derivative of the free energy directly, only the expectation value of the derivative of the microstate energy needs to be evaluated~\citep{frenkel_understanding_1996}.

When choosing $E_{\mathcal S}(\Lambda)$ as a linear function between the microstate energies of system $A$ and system $B$,
\begin{equation}
    E_{\mathcal S}(\Lambda) = E_{\mathcal S}(A) + \Lambda(E_{\mathcal S}(B) - E_{\mathcal S}(A)),
\end{equation}
we obtain
\begin{align}
    \Delta F_{A\to B} 
    &= \int_0^1 \braket{E(B) - E(A)}_{\Lambda} \mathrm d\Lambda.
\end{align}
In practice, the integral is discretized by calculating the expectation values for a discrete number of values of $\Lambda$. Then, to determine the difference in free energy from $A$ to $B$, the expectation values are summed over all values of $\Lambda$, yielding the following approximate free energy difference:
\begin{equation} \label{eq:free_energy_difference_integral}
    \Delta \widetilde F := \frac{1}{N_\Lambda} \sum_\Lambda \braket{E(B)-E(A)}_\Lambda \, .
\end{equation}
Note that while we choose a linear path here for simplicity, selecting a non-linear path may be advantageous because the instantaneous difference in the block encoding constants for nuclear Hamiltonians $H_A$ and $H_B$ that are used to compute the internal energies $\braket{E_A}$ and $\braket{E_B}$ may be chosen to be far smaller for such a thermodynamic integration rather than the differences, which can lead to smaller costs for the overall calculation as discussed in Ref.~\cite{simon2024amplified}.

Under the Liouvillian formalism, the thermal distribution is obtained by evolving the phase-space density with the NVT canonical Liouvillian for some sufficient equilibration time $t_{\rm eq}$, while thermal averages are obtained by taking the expectation value of the relevant operator over the thermal distribution for each point in the phase-space. In our case, where a quantum state on registers represents the thermal distribution, the computational basis corresponds to a point in the discretized phase-space, and the square of the corresponding amplitude would be its probability. The microstate energies over the phase-space can then be encoded into a diagonal matrix to form the nuclear Hamiltonian, with each entry corresponding to the microstate energy at the relevant point in phase-space.

\subsection{Block encodings and quantum linear algebra}

We assume access to operators such as Hamiltonians and forces in the form of block encodings~\citep{chakraborty2019power, gilyen2019quantum} in this paper. A block encoding is an encoding scheme that embeds a general matrix into the top left corner of a unitary matrix, such that
\begin{equation}
    \lvert A-\alpha (\bra{0^a}\otimes \mathbb{I}_n)U_A(\ket{0^a}\otimes \mathbb{I}_n)\rvert\leq \varepsilon \, .
\end{equation}
This embedding of $A$ into a unitary $U_A$ is then known as a $(\alpha, a, \varepsilon)$-block-encoding of $A$. In the main text, we refer to $\alpha$ as the scaling factor of the block encoding when mentioned in isolation.

To manipulate such block encodings, we can use a framework known as quantum singular value transformation (QSVT) \citep{gilyen2019quantum}, which can be used to apply a polynomial transformation to the singular values of the matrix $A$ such that
\begin{equation}
    \text{QSVT}(U_A) = \begin{bmatrix}
        \mathrm{Poly}^{\rm (SV)}(A) & * \\
        *                       & *
    \end{bmatrix} \, .
\end{equation}
In the case of Hermitian matrices, the transformations are applied to the eigenvalues of the matrix. We refer the reader to \cref{appQLA} for a detailed account of QSVT as a quantum linear algebra toolbox to perform ground state preparation~\citep{lin2020nearoptimal} and Hamiltonian simulation~\citep{low2017optimal, low2019hamiltonian, gilyen2019quantum}, which are subroutines used in our main results. The reader will also find a brief review on quantum signal processing (QSP) without angle finding~\citep{alase2025quantum}. Furthermore, we deferred several new results to the Appendix that were derived as part of our proofs, such as algorithms for robust ground state preparation and robust Hamiltonian simulation via QSP without angle finding.

\section{Improved quantum algorithm for Liouvillian simulation}
\label{secLS}

\begin{algorithm*}
\caption{Quantum algorithm for Liouvillian simulation}
\label{algoLS}
\Indm
\KwIn{Number of nuclei $N$, Evolution time $t$, Precision $\varepsilon$, Block encoding of the electronic Hamiltonian $H_{\tel}^{\rm ctrl}$~\citep{su2021faulttolerant}, Block encodings of the nuclear-position-controlled electronic force operator $\partial_{x_{n,j}} H_{\tel}^{\rm ctrl}$ for $n \in [N]$ and $j \in \{1, 2, 3\}$~\citep{obrien2022efficient}, Block encoding of the discretized momentum derivative operator $D_{p_{n, j}'}$ for $n \in [N]$ and $j \in \{1, 2, 3\}$~\citep{simon2024improved}, Block encoding of the classical Liouvillian $L_{\rm class}$~\citep{simon2024improved}, Approximate ground state preparation oracle $U_I$, Quantum state $\ket{\rho_0}$ encoding initial phase-space density of nuclei and heat bath.}
\KwOut{$\varepsilon$-close approximation of $\ket{\rho_t} = e^{-iLt}\ket{\rho_0}$ in $\ell_2$ distance}
\Indp 
\DontPrintSemicolon
\For{every query to the block encoding $L$ in the QSVT-based Hamiltonian simulation algorithm~\citep{gilyen2019quantum}}{
    Prepare state $\frac{1}{\sqrt{3N}}\sum_{n=1}^N\ket{n} \otimes \sum_{j=1}^3\ket{j}$ and additional electronic register.\;
    Apply approximate ground state preparation oracle $U_I$ to nuclear and electronic registers $U_I\ket{\bar{\vec x}}\ket{0} \to \ket{\bar{\vec x}}\ket{\phi_{\rm init}(\vec x)}$.\;
    Apply ground state preparation algorithm with nuclear-position-controlled electronic Hamiltonian $H_{\tel}^{\rm ctrl} = \sum_{\vec x} \ketbra{\bar{\vec x}}{\bar{\vec x}}\otimes H_{\tel}(\vec x)$.\;
    Controlled on $\ket{n}$ and $\ket{j}$, apply nuclear-position-controlled electronic force operator $\partial_{x_{n,j}} H_{\tel}^{\rm ctrl} = \sum_{\vec x} \ketbra{\bar{\vec x}}{\bar{\vec x}}\otimes \frac{\partial H_{\tel}(\vec x)}{\partial \vec x_{n,j}}$.\;
    Uncompute the (approximate) electronic ground state to obtain block-encoding of $\Del$ controlled on $\ket{n}$ and $\ket{j}$.\;
    Controlled on $\ket{n}$ and $\ket{j}$, tensor product $\Del$ with discretized momentum derivate operator $D_{p_{n,j}'}$.\;
    Use LCU to sum over registers $\ket{n}$ and $\ket{j}$ to obtain block encoding of electronic Liouvillian $L_{\tel}$.\;
    Use LCU to sum over classical Liouvillian $L_{\rm class}$ and electronic Liouvillian $L_{\tel}$ to obtain block encoding of Liouvillian $L$.\;
}
Apply $e^{-iLt}$ operator obtained from QSVT-based Hamiltonian simulation algorithm to $\ket{\rho_0}$.\;
\Return $\varepsilon$-close approximation of $\ket{\rho_t} = e^{-iLt}\ket{\rho_0}$.\;
\end{algorithm*}

We provide an algorithm that utilizes evolution under the Liouvillian directly by implementing the block encoding of the electronic Liouvillian from Ref.~\citep{simon2024improved}. We then encode the ground state energy derivative on an electronic register and exploit a form of phase kickback~\citep{cleve1998quantum} to report the results to the nuclear registers. This allows us to add together the classical and electronic Liouvillians, which we then exponentiate via QSVT (in the same way one would perform Hamiltonian simulation). The benefits of our approach are the following:
\begin{enumerate}
    \item We achieve super-polynomial asymptotic reductions of the runtime in terms of precision $\varepsilon$ from $\mathcal{O} (\varepsilon^{-o(1)})$ in the simulation relative to the quantum algorithm in Ref.~\citep{simon2024improved} to $\mathcal{O} (\poly\log \frac{1}{\varepsilon})$.
    \item We decrease the dependency on the number of particles by order of $\mathcal{\widetilde O}(N^{1+o(1)}\widetilde N^{1+o(1)}N_{\mathrm{tot}}^{o(1)})$ and time by an order of $\mathcal{\widetilde O}(t^{o(1)})$ from the removal of Trotterization and from the direct implementation of energy derivatives when compared to Ref.~\citep{simon2024improved}.
\end{enumerate}
These advantages are informally summarized in the following theorem, while the formal statement (\cref{thmLiouSim}) and proof of this theorem are detailed in \cref{appLS}.

\begin{theorem}[Born-Oppenheimer Liouvillian simulation of coupled quantum-classical dynamics; informal]
    Given an initial state $\ket{\rho_0}$ that encodes the initial discretized phase-space density and a discretized Liouvillian $L$ under the NVT canonical ensemble, we can construct an explicit quantum algorithm that outputs a quantum state that is $\varepsilon$-close in $\ell_2$ distance to $\ket{\rho_t} = e^{-iLt}\ket{\rho_0}$ using
    \begin{equation*}
        \widetilde{\mathcal{O}}\left(\frac{N \widetilde N N_{\mathrm{tot}}^3 t}{\delta\gamma}\log^3\frac{1}{\varepsilon}\right)
    \end{equation*}
    Toffoli gates,
    \begin{equation*}
        \widetilde{\mathcal{O}}\left(\frac{1}{\delta} \left(N N_{\mathrm{tot}} t + \log\frac{1}{\gamma\varepsilon}\right)\right)
    \end{equation*}
    queries to the approximate ground state preparation oracle $U_I$, and $\widetilde{\mathcal{O}}(N_{\mathrm{tot}} + \log\frac{t}{\delta\gamma\varepsilon})$ qubits, where $N$ is the number of nuclei, $\widetilde N$ is the number of electrons, $N_{\mathrm{tot}} = N + \widetilde N$ is the total number of particles, $\delta$ is the lower bound for the overlap between the approximate and true ground state prepared by $U_I$ and $\gamma$ is a lower bound on the spectral gap of the electronic Hamiltonian over all nuclear phase-space grid points.
    \label{thmLiouSimInfml}
\end{theorem}

Let us now detail the steps in the construction of the algorithm, which is given above as \cref{algoLS}. What is required is to prepare the block encoding of the Liouvillian, which can be obtained by computing partial derivatives of the extended nuclear Hamiltonian $H_{\rm ext}$ (inclusive of nuclear kinetic, nuclear potential, ground state energy terms, and bath terms) of the system on both position and momentum as shown in \cref{eq:Liouvillian}. Once we have the block encoding of the Liouvillian, given that it remains self-adjoint, we can utilize a Hamiltonian simulation algorithm~\citep{low2017optimal,low2019hamiltonian} to obtain its corresponding evolution operator and apply it to evolve a prepared initial state $\ket{\rho_0}$ encoding the system's initial phase-space density.

We start with the block encoding of the electronic forces, which is found to be the derivative of the ground state energy function over nuclear positions:
\begin{equation}
    \Del := \frac{\partial E_{\tel}}{\partial x_{n,j}}( \vec x)\ketbra{\bar {\vec x}}{\bar {\vec x}}.
\end{equation}
We make use of the Hellmann--Feynman theorem~\citep{hellmann1937einfuhrung,feynman1939forces} to implement this block encoding.  This theorem states that the derivative of the total energy with respect to a parameter can be found by computing an expected value of the derivative of the Hamiltonian with respect to that same parameter, such that
\begin{equation}
    \frac{\partial E(x)}{\partial x} = \Braket{\psi(x)|\frac{\partial H(x)}{\partial x}|\psi(x)}\, ,
\end{equation}
where $H(x)$ is a Hamiltonian dependent on the continuous parameter $x$, $\ket{\psi(x)}$ is an eigenstate of $H(x)$ implicitly dependent on $x$, and $E(x)$ is the corresponding eigenenergy/eigenvalue.

When combined with the ground state preparation algorithm by \citet{lin2020nearoptimal}, one can utilize the Hellmann--Feynman theorem to encode electronic forces directly into the amplitude, eliminating the need for central finite difference schemes to approximate these values. Usage of the Hellmann--Feynman theorem was previously explored in Ref.~\citep{obrien2022efficient} under the context of estimating the molecular forces themselves, using a quantum phase estimation~\citep{kitaev1995quantum} of the Szegedy walk unitary~\citep{szegedy2004quantum} to implement the overlap estimation algorithm~\citep{knill2007optimal}. These forces can then be integrated into classical MD methods by replacing subsystems that compute quantum mechanical effects. In this paper, we develop a diagonal block encoding method that encodes the molecular forces for each distinct nuclear position configuration. Rather than leveraging phase estimation to extract the forces, we utilize the Hellmann--Feynman theorem directly as a series of matrix multiplication operations, block-encoding the forces by first computing the ground state, applying the force operator, and then directly uncomputing the ground state. This block-encodes the scalar value of the forces in the top left element.

Note that in our case, even though we do not use the exact eigenstate for the Hellmann--Feynman theorem, the errors that stem from the inexact ground state, even after the ground state preparation algorithm, can be accounted for in the error analysis of the block encoding, while discretization errors can be suppressed by adding more qubits. Further, as our choice of bases from the electronic system, we also need not consider the effects of Pulay stress~\citep{pulay1969ab,francis1990finite}, given that they are absent under nuclear-position-independent bases such as plane waves. A further discussion of basis sets involving Pulay forces under the Hellmann--Feynman theorem can be found in Ref.~\citep{obrien2022efficient}.

When implemented in superposition, by applying the controlled operators above, and subsequently projecting the electronic registers into the zero state, we can encode the force values onto the nuclear registers for all different nuclear positions by phase kickback~\citep{cleve1998quantum}. Therefore, we can implement
\begin{equation}
    \Del = \Braket{\psi_0(\vec x)|\frac{\partial H_{\tel}(\vec x)}{\partial x_{n,j}}|\psi_0(\vec x)} \ketbra{\bar {\vec x}}{\bar {\vec x}},
    \label{eqDelInfml}
\end{equation}
for the discretized grid space of position on the quantum computer. This approach results in the construction of the energy derivative operator with a logarithmic dependency on precision as shown in \cref{lemDel}. Combining the results with an implementation of the discretized momentum derivative operator $D_{p_{n,j}'}$ shown in Ref.~\citep{simon2024improved}, this provides a block encoding of the electronic Liouvillian via linear combination of unitaries~\citep{childs2012hamiltonian}:
\begin{equation}
    \widetilde L_{\tel} = i \sum_{n=1}^N\sum_{j=1}^3  \Del
    \otimes D_{p_{n,j}'}.
    \label{eqElecLiouInfml}
\end{equation}
The details of this block encoding are shown in \cref{propElecLiou}. This further provides a logarithmic dependency on the error of the total Liouvillian simulation algorithm, as the Liouvillian is directly and analytically block-encoded.

In addition to decreasing the runtime dependency on precision, by using the Hellmann--Feynman theorem to block-encode the force operator instead of a central finite difference approximation~\citep{simon2024improved}, we reduce the scaling factor of the block encoding of the electronic Liouvillian by a factor of $\mathcal{O}(N_{\mathrm{tot}})$~\citep{obrien2022efficient}. This is due to the reduced number of terms to implement in the Hellmann--Feynman approach. When combined with the classical Liouvillian, this produces a reduction in the total scaling of $\mathcal{O}(\widetilde N)$ in the full Liouvillian, with the implementation cost and details shown in the proposition below:
\begin{proposition}[Block encoding of the Liouvillian; informal]
    We can construct an explicit $(\alpha_L, a_L, \varepsilon_L)$-block-encoding of the Liouvillian $L$ where
    \begin{align*}
        \alpha_L \in \mathcal O & \left(N N_{\mathrm{tot}}\right) \\
        a_L \in \mathcal{O}     & \left(\widetilde{N} + \log\frac{N_{\mathrm{tot}}}{\varepsilon_L}\right).
    \end{align*}
    This block encoding can be prepared with $\mathcal{O}(\frac{1}{\delta})$ queries to $U_I$, and
    \begin{align*}
        \widetilde{\mathcal O} & \left(\frac{\widetilde N N_{\mathrm{tot}}^2}{\delta\gamma}\log^2 \frac{1}{\varepsilon_L}\right)
    \end{align*}
    Toffoli gates.
    \label{propLiouInfml}
\end{proposition}
The formal statement (\cref{propLiou}) and proof of this proposition can be found as part of the proof of \cref{thmLiouSimInfml} in \cref{appLS}. In the context of the Liouvillian simulation algorithm, this results in a reduction in runtime by $\mathcal{O}(\widetilde N)$ in the Hamiltonian simulation step, due to the reduced scaling factor of central finite difference methods.

In addition to the runtime improvements resulting from the direct implementation of the Liouvillian and the use of the Hellmann--Feynman theorem, further speedups can be achieved. Namely, all sublinear dependencies are removed solely by using direct evolution as opposed to Trotterization. Furthermore, by directly implementing the electronic Liouvillian in superposition via LCU~\citep{childs2012hamiltonian} instead of taking the product of multiple commutating controlled Hamiltonian simulations, we can effectively reduce the number of queries to the ground state preparation algorithm. This results in a further reduction of the runtime by a factor of $\mathcal{O}(N)$.

Lastly, while our results are primarily focused on the NVT canonical ensemble, these improvements can be easily translated to the NVE microcanonical ensemble by replacing the block encoding of the classical Liouvillian with its NVE counterpart in Ref.~\citep{simon2024improved}. These results are shown in \cref{appNVE}.

\section{Quantum algorithm for alchemical free energy calculation}
\label{secTI}

\begin{algorithm*}
\caption{Quantum algorithm for thermodynamic integration}
\label{algoTI}
\Indm
\KwIn{Number of interpolation points $N_\Lambda$, Equilibration time $t_{\rm eq}$, Precision $\varepsilon$, Failure rate $\xi$, Block encoding of Liouvillian of systems $A/B$ over nuclear phase-space $U_{L_A}/U_{L_B}$ and scaling factors $\alpha_{L_A}/\alpha_{L_B}$ from \cref{propLiouInfml}, Block encoding of the difference of the nuclear Hamiltonian between systems $A$ and $B$ over nuclear phase-space $U_{\Delta H}$ and scaling factor $\alpha_{\Delta}$, Initial nuclear phase-space density $\ket{\rho_0} = \sum_{\vec x, \vec p', s, p_s} c_{\vec x, \vec p', s, p_s} \ket{\bar{\vec x}}\ket{\bar{\vec p}'}\ket{\bar s}\ket{\bar p_s}$}
\KwOut{Helmholtz free energy difference $\Delta F_{A\to B}$ within $\varepsilon$ additive error with success rate $1-\xi$}
\Indp
\DontPrintSemicolon
Prepare uniform superposition over $N_\Lambda$ such that $\frac{1}{\sqrt{N_\Lambda}}\sum_{\bar \Lambda=0}^{N_\Lambda-1} \ket{\Lambda}$ where $\bar \Lambda = N_{\Lambda}\cdot \Lambda$.\;
Use quantum arithmetic circuits~\citep{vedral1996quantum} to compute a fixed point representation of $\alpha_{\Lambda} = \alpha_{L_A} + \Lambda (\alpha_{L_B}-\alpha_{L_A})$ such that we have $\frac{1}{\sqrt{N_\Lambda}}\sum_{\bar \Lambda=0}^{N_\Lambda-1} \ket{\Lambda}\ket{\alpha_\Lambda} $.\;
\uIf{QSVT-based Hamiltonian simulation~\citep{gilyen2019quantum}}{
    Controlled on $\ket{\alpha_\Lambda}$, compute, on a separate register, the fixed-point representation of phase angles $\{\theta_j^{\Lambda}\}_{j\in [\mathcal{D}]}$ for polynomial approximation of function $f_{\Lambda}(x) =\exp(-i\alpha_\Lambda x t_{\rm eq})$ for degree $\mathcal{D} \in\mathcal O (\max(\alpha_{L_A}, \alpha_{L_B})t_{\mathrm{eq}} + \log\frac{1}{\varepsilon})$.\;
}
\uElseIf{QSP-without-angle-finding-based Hamiltonian simulation~\citep{alase2025quantum}}
{
    Controlled on $\ket{\alpha_\Lambda}$, using quantum arithmetic circuits to compute $\exp(-it\alpha_\Lambda \cos \frac{2\pi k}{4\mathcal{D}})$ for $k \in [4\mathcal{D}]$ for degree $\mathcal D \in\mathcal O (\max(\alpha_{L_A}, \alpha_{L_B})t_{\mathrm{eq}} + \log\frac{1}{\varepsilon})$.\;
    Use controlled rotations to create a diagonal block encoding of the values $\exp(-it\alpha_\Lambda \cos \frac{2\pi k}{4\mathcal{D}})$ for $\Lambda \in [N_\Lambda]$, $k \in [4 \mathcal D]$.\;
}
\For{every query to the block encoding $L$ in the Hamiltonian simulation algorithm}{
    Use quantum arithmetic circuits to compute a fixed point representation of $\frac{(1-\Lambda)\alpha_{L_A}}{\alpha_{L_\Lambda}}$ such that we have $\frac{1}{\sqrt{N_\Lambda}}\sum_{\bar \Lambda=0}^{N_\Lambda-1} \ket{\Lambda}\ket{\alpha_\Lambda}\Ket{\frac{(1-\Lambda)\alpha_{L_A}}{\alpha_{L_\Lambda}}}$ \; 
    Controlled on $\Ket{\frac{(1-\Lambda)\alpha_{L_A}}{\alpha_{L_\Lambda}}}$, take weighted LCU to prepare block encoding of $L_{\Lambda} = L_A + \Lambda(L_B-L_A)$.\;
    Uncompute $\Ket{\frac{(1-\Lambda)\alpha_{L_A}}{\alpha_{L_\Lambda}}}$.
}
Controlled on $\ket{\Lambda}$, obtain $e^{-iL_\Lambda t_{\rm eq}}$ from Hamiltonian simulation algorithm such that we have $\sum_{\bar \Lambda=0}^{N_\Lambda-1}\ketbra{\Lambda}{\Lambda}\otimes e^{-iL_{\Lambda}t_{\rm eq}}$.\;
Apply $\sum_{\bar \Lambda=0}^{N_\Lambda-1}\ketbra{\Lambda}{\Lambda}\otimes e^{-iL_{\Lambda}t_{\rm eq}}$ to $\frac{1}{\sqrt{N_\Lambda}}\sum_{\bar \Lambda=0}^{N_\Lambda-1} \ket{\Lambda} \otimes \ket{\rho_0}$.\;
Duplicate register $\ket{\bar s}$ using CNOTs to retain information of the bath register.\;
Apply the Hadamard test~\citep{cleve1998quantum} to the block encoding of $\Delta H$ on the registers $\ket{\bar{\vec x}}\ket{\bar{\vec p}'}\ket{\bar s}$ to obtain amplitude $\sqrt{\frac{1}{2} (\frac{\Delta \widetilde F}{\alpha_\Delta}+1)}$ on $\ket{0}$ of the measurement qubit.\;
Use amplitude estimation~\citep{brassard2002quantum} to obtain $\frac{1}{2} (\frac{\Delta \widetilde F}{\alpha_\Delta}+1)$ to $\frac{\varepsilon}{2\alpha_{\Delta}}$ precision with $1-\xi$ success rate.\;
\Return $\varepsilon$-close estimate of Helmholtz free energy difference $\Delta F$
\end{algorithm*}

The results we developed in the previous section give us access to a block encoding of the full Liouvillian as per \cref{propLiouInfml}. This allows us to implement algorithms that require access to weighted sums of Liouvillian terms. With simulations under the NVT ensemble, this enables the efficient preparation of thermal states that can be further used to compute thermal averages such as the free energy.

Prior work in Ref.~\citep{simon2024improved} computes the free energy by first estimating both the internal energy and Gibbs entropy separately, which are then subtracted. While the internal energy can be efficiently estimated by the Hadamard test with the nuclear Hamiltonian, the estimation of Gibbs entropy scales with the dimension $\eta$ of the density matrix $\ketbra{\rho_t}{\rho_t}$~\citep{gilyen2020distributional}. Given that $\ket{\rho_t}$ is prepared on a collection of registers $\ket{\bar x_{n,j}}$, $\ket{\bar p_{n,j}'}$, $\ket{\bar s}$, and $\ket{\bar p_s}$, the dimension $\eta$ corresponds to the dimension of the classical nuclear phase-space. If an individual register has dimension $g$, then $\eta \in \mathcal{O}(g^{6N+2})$ and is hence exponential to the number of atoms. However, this exponential scaling of entropy calculation may not be due to the algorithm itself, as the problem of computing absolute free energies is found to be NP-hard by reducing it to the problem of partition function computation, which is NP-complete~\citep{istrail2000statistical}.

Contrary to the computation of absolute free energies, relative free energies are known to be easier to compute, with various classical algorithms dedicated to this purpose~\citep{Kirkwood1935_StatistiscalMechanics, Zwanzig1955FreeEnPerturbation}. Aside from the reduced computational complexity, relative free energies are more relevant to applications because they directly determine observable quantities like binding affinities and reaction equilibria~\citep{mey2020best}. Hence, in our work, we compute relative free energies instead of absolute free energies. Borrowing techniques from classical computational chemistry literature and leveraging alchemical methods described in \cref{prelim:thermodynamic_integration} enables us to calculate the free energy differences between two closely related systems. With this strategy, we provide the following improvements to calculating free energy differences.
\begin{enumerate}
    \item Compared to Ref.~\citep{simon2024improved}, by implementing thermodynamic integration, we bypass the use of entropy estimation algorithms~\citep{gilyen2020distributional}, which scale with the total size of the nuclear phase-space $\eta$ used in the simulation (unless entropy estimation is evaluated over a coarsened grid) and therefore exponentially with the number of qubits $\Theta(N)$ that are used to represent the nuclear degrees of freedom. Our algorithm is only dependent on calculating the internal energy, which can be achieved efficiently via the Hadamard test and is not dependent on $\eta$.
    \item We achieve $\widetilde{\mathcal O} \left(\frac{1}{\varepsilon}\right)$ scaling for Helmholtz free energy estimation by using the Liouvillian simulation algorithm in the previous section, as well as a more efficient block-encoding of the nuclear Hamiltonian that depend only logarithmically on the overall precision---we achieve this by replacing quantum phase estimation~\citep{kitaev1995quantum} for ground state energy computation with phase kickback~\citep{cleve1998quantum} techniques.
\end{enumerate}
Below, we state our main result:
\begin{theorem}[Quantum complexity for thermodynamic integration; informal]
    Given an initial state $\ket{\rho_0}$ that encodes the initial discretized phase-space density and two discretized Liouvillians, $L_A$, and $L_B$, under the NVT canonical ensemble encoding the dynamics for systems $A$ and $B$, as well as an equilibration time $t_{\mathrm{eq}}$ for Liouvillian simulation, there exists a quantum algorithm that can compute the free energy difference between the two systems up to additive $\varepsilon$-precision with $1-\xi$ success probability using
    \begin{equation*}
        \widetilde{\mathcal{O}}\left(\frac{N \widetilde N N_{\mathrm{tot}}^5 t_{\mathrm{eq}}}{\delta\gamma\varepsilon}\log\frac{1}{\xi}\right)
    \end{equation*}
    Toffoli gates,
    \begin{equation*}
        \widetilde{\mathcal{O}}\left(\frac{N N_{\mathrm{tot}}^3 t_{\mathrm{eq}}}{\delta\varepsilon} \log^2\left(\frac{1}{\gamma}\right)\log\left(\frac{1}{\xi}\right)\right)
    \end{equation*}
    queries to the approximate ground state preparation oracle $U_I$, and $\widetilde{\mathcal{O}}(N_{\mathrm{tot}} + \log\frac{t_{\mathrm{eq}}}{\delta\gamma\varepsilon})$ qubits, where $N$ is the number of nuclei, $\widetilde N$ is the number of electrons, $N_{\mathrm{tot}} = N + \widetilde N$ is total number of particles, $\delta$ is the lower bound for the overlap on the approximate and true ground state prepared by $U_I$ and $\gamma$ is a lower bound on the spectral gap of the electronic Hamiltonian over all nuclear phase-space grid points.
    \label{thmThermIntInfml}
\end{theorem}
The formal statement (\cref{thmThermInt}) and proof of this proposition have been deferred to \cref{appThermoInt}.

\begin{figure*}
    \includegraphics[width=0.9\linewidth]{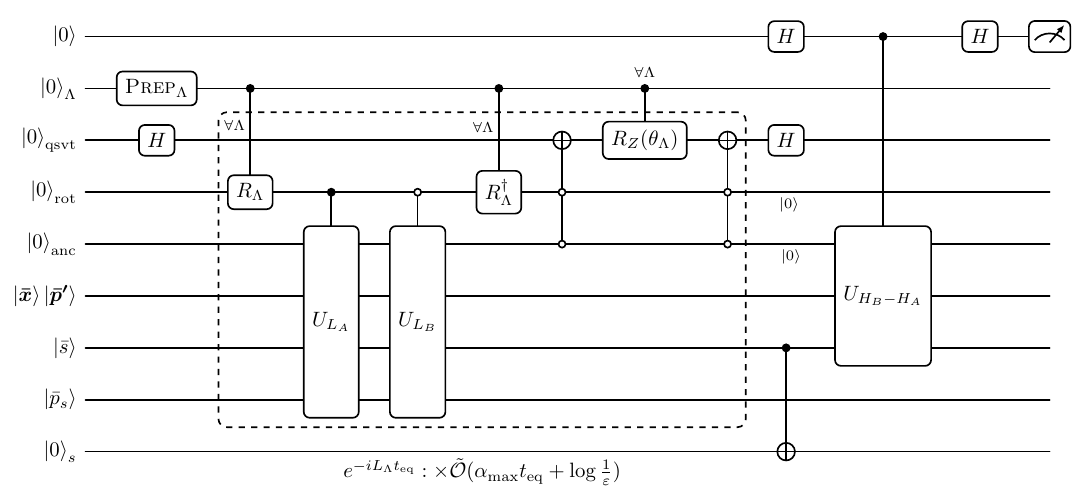}
    \caption[Circuit for obtaining the free energy difference.]{\emph{Circuit for obtaining the free energy difference.} The dashed block indicates the repeated linear combination of block encodings to produce a block encoding of $L_\Lambda$, which is then transformed to $e^{-iL_\Lambda t_{\mathrm{eq}}}$ via QSVT using $\Lambda$-dependent phase angles. Methods for this are discussed in the main text and circuits given in \cref{appLambdaLiouvillian}. Here we illustrate the use of the standard QSVT scheme in Ref.~\citep{gilyen2019quantum}, where we would load the phase angles from a separate quantum register and where an angle finding algorithm would be implemented. We ignore the LCU required to add $\cos$ and $\sin$ together, as well as amplitude amplification for simplicity. Note the CNOT gate that copied the contents of the bath register $\ket{\bar s}$ before it is traced out. This is due to the extended Hamiltonian with the Nos\'e thermostat being dependent on $s$ and is thus required for a description of the thermal state, while still being needed to be traced out in an analog of being integrated over in classical molecular dynamics simulations. The small $\ket{0}$ on the wires indicates that the ancilla qubits are returned to the zero state after robust oblivious amplitude amplification~\citep{berry2015simulating} such that they can be reused as the ancilla qubits for the Hadamard test.} \label{fig:thermodynamic_integration_overview}
\end{figure*}

We are now ready to describe the main algorithm (\cref{algoTI}). Suppose we want to compute the free energy differences between two systems $A$ and $B$. Then, to perform thermodynamic integration over $\Lambda$ in superposition, we need to apply an interpolated NVT Liouvillian $L_\Lambda$, which can be easily constructed from the $\Lambda$-dependent nuclear Hamiltonian $H_\Lambda = (1-\Lambda)H_A + \Lambda H_B$ as the Liouvillian is linearly dependent on the Hamiltonian as seen in \cref{eq:Liouvillian}. Given this linearity, the corresponding Liouvillian is $L_\Lambda = (1-\Lambda)L_A + \Lambda L_B$, where $L_A$ and $L_B$ may act on different sets of nuclear registers. Registers describing nuclear degrees of freedom that are shared between $A$ and $B$ can be acted on by both $L_A$ and $L_B$, or only one of them (if there is an increase or decrease of nuclei). Given this construction of the intermediate Liouvillian via LCU, we do not actually construct individual components of the intermediate Liouvillians and hence do not need to formulate the nuclear Hamiltonian $H_\Lambda$ and ground state correction terms $E_{\tel, \Lambda}$ for the intermediate systems. This implies that the assumptions on ground state overlap $\delta$ and spectral gap $\gamma$ only apply to the physical initial and final systems $A$ and $B$.

To compute the free energy, starting with a common initial state $\ket{\rho_0}$, for each $\Lambda$, we evolve the state by the interpolated NVT Liouvillian $L_\Lambda$ over a predetermined and sufficient equilibration time $t_{\mathrm{eq}}$ such that we obtain the thermal state $\ket{\rho_\Lambda} = e^{-iL_\Lambda t_{\mathrm{eq}}}\ket{\rho_0}$. Given that the state includes registers corresponding to the heat bath, we then trace out the bath registers in a quantum analog of integrating over them to obtain $\rho_{\mathrm{sys}, \Lambda} = \tr_{\mathrm{bath}}(e^{-iL_\Lambda t_{\mathrm{eq}}}\ketbra{\rho_0}{\rho_0}e^{iL_\Lambda t_{\mathrm{eq}}})$. By the thermodynamic integration algorithm, we can then compute the free energy difference between the systems $A$ and $B$ by computing
\begin{align}
    \Delta \widetilde F & = \frac{1}{N_\Lambda}\sum_\Lambda \tr(\rho_{\mathrm{sys}, \Lambda} (H_B-H_A))\nonumber \\
    & = \frac{1}{N_\Lambda}\sum_\Lambda \bra{\rho_\Lambda}(H_B-H_A)\otimes \mathbb{I}\ket{\rho_\Lambda}
\end{align}
This quantity can be computed on a quantum computer via the Hadamard test~\citep{kitaev1995quantum,cleve1998quantum}, where we first create a uniform superposition over $N_\Lambda$ points with the unitary $\textsc{Prep}_\Lambda$, apply time evolution with $L_\Lambda$, and a final application of a controlled block encoding of the nuclear Hamiltonian difference $H_B-H_A$. We show an illustration of the quantum circuit in \cref{fig:thermodynamic_integration_overview}.

Delving further into the implementation details, we first discuss the implementation of the evolution of $\Lambda$-controlled Liouvillians. From the results of \cref{propLiouInfml}, we can produce block encodings of the Liouvillians $L_A$ and $L_B$, but with different scaling factors $\alpha_{L_A}$ and $\alpha_{L_B}$. We block-encode the interpolated Liouvillian $L_\Lambda$ by a weighted LCU whose weights are dependent on $\alpha_{L_A}$, $\alpha_{L_B}$, and $\Lambda$, which has a scaling factor of
\begin{equation}
    \alpha_\Lambda = (1-\Lambda)\alpha_{L_A} + \Lambda \alpha_{L_B}
\end{equation}
To perform the Hamiltonian simulation by QSVT, the polynomial approximations of $e^{-iL_\Lambda t_{\mathrm{eq}}}$ will also depend on $\Lambda$ as the number of queries to $U_{L_\Lambda}$ is dependent on $\alpha_\Lambda$. We provide two strategies to perform this transformation of $L_\Lambda$, first using standard QSVT, and then an approach using QSVT without angle finding~\citep{alase2025quantum} that provides the $\mathcal{O}(\frac{1}{\varepsilon})$ scaling of our algorithm.

Firstly, under conventional QSVT schemes, for each $\Lambda$ in superposition, we require calculation of the phase angles $\{\theta_\Lambda\}$ with either analytical methods, or with optimization methods~\citep{chao2020finding,dong2021efficient,wang2022energy}, where the scaling of the number of $\Lambda$-s required, or $N_\Lambda$, can be shown to be
\begin{equation}
    N_\Lambda \in \mathcal{O} \left(\frac{\|L_A-L_B\|\|H_A-H_B\|t_\mathrm{eq}}{\varepsilon}\right),
\end{equation}
by bounding left Riemann sum errors and where a proof is presented in \cref{appThmIntDiscErr}.
The QSVT sequence must have a length sufficient for the largest polynomial degree required for any $\Lambda$. For other values of $\Lambda$, the final angles can be set to $0$ such that the remaining applications of $U_{L_\Lambda}$ and its inverse create the identity. In the case of optimization methods, one can set the degree of optimization $\mathcal{D}$ directly to the maximum degree. If the angles are classically loaded, then we need an extra $\mathcal O(N_\Lambda)$ arbitrary rotation gates to load all required phase angles per query of $U_{L_\Lambda}$, as we need a different phase gate for each $\Lambda$.

However, one can bypass this by implementing the classical computations on a separate quantum register in superposition and then loading the results directly into the phase angles, which would then cost at most $\mathcal {O} (\poly \mathcal{D} \poly \log N_\Lambda)$ additional gates to load the phase angles. Furthermore, suppose we replace the QSVT procedure with generalized quantum signal processing (GQSP)~\citep{motlagh2024generalized}. In that case, it is known that the classical algorithm for computing the phase angles scales as $\mathcal{O}(\mathcal{D}\log \mathcal{D})$. Thus, when implemented on a quantum computer, we expect the number of gates to be at most $\widetilde{\mathcal{O}}(\mathcal{D}\poly \log N_\Lambda)$. However, the precise scaling of such algorithms when implemented on quantum computers, as well as the quantum circuit for such implementations, is not obvious.

An alternative method sidesteps the requirement of angle finding by using recent results on QSVT without classical angle finding~\citep{alase2025quantum} by instead encoding a functional representation in a diagonal block encoding, in this case with entries $\exp(-it\alpha_\Lambda \cos \frac{2\pi k}{4\mathcal{D}})$ for $k \in [4\mathcal{D}]$, where $\mathcal{D} \in\mathcal O (\max(\alpha_{L_A}, \alpha_{L_B})t_{\mathrm{eq}} + \log\frac{1}{\varepsilon})$ is the polynomial approximation degree. Such function entries are computed in superposition via quantum arithmetic circuits~\citep{vedral1996quantum} while controlled on $\Lambda$. This can then be encoded into the amplitude of the diagonal block encoding via controlled rotations or the alternating-sign trick~\citep{berry2014exponential}. The reason that the additional cosine in the function is needed is due to the use of qubitized operators, whose eigenvalues are $\arccos$ of the original block encoded matrix. This method allows for a gate count scaling $\mathcal O(\log^2 N_\Lambda)$ to implement the function representation. Adding the additional $\mathcal O(\log^2 N_\Lambda)$ cost to produce the weighted Liouvillian per query to $U_{L_\Lambda}$, this requires an additional $\mathcal O(\mathcal{D}\log^2 N_\Lambda)$ gates for implementation. On the other hand, this method requires an additional constant factor calls to the block encoding of the interpolated Liouvillians $U_{L_\Lambda}$ as well as more ancilla qubits when compared to standard QSVT/GQSP. However, it bypasses the requirement of implementing phase angle finding algorithms on a quantum computer, as needed for standard QSVT/GQSP implementations. A more detailed discussion of an exact construction is provided in \cref{appLambdaLiouvillian}.

We note that while the dependency on $N_\Lambda$ is logarithmic in the runtime, this does not indicate an exponential speedup when compared to classical methods. While conventional thermodynamic integration implementations would evaluate the different $N_\Lambda$ interpolated systems individually and obtain the free energy via numerical integration methods such as Simpson's rule or specifically designed physics-based fitting functions~\citep{jorge2010effect} for higher accuracy, Monte Carlo integration~\citep{metropolis1949monte} can be applied to obtain faster runtime asymptotics. The number of samples in Monte Carlo integration would then depend on the square of the spectral norm of the nuclear Hamiltonian difference $\Delta H$ as opposed to $N_\Lambda$, for which the runtime is at most logarithmically dependent on the Monte Carlo integration. This logarithmic dependency stems from the precision of floating point representation on classical computers, it being $\mathcal O(N_\Lambda^{-1})$.

Moving on from the implementation of interpolated Liouvillians, one then time evolves the system for a duration $t_{\mathrm{eq}}$ that is sufficient to equilibrate all $N_\Lambda$ Liouvillians $\{L_\Lambda\}$ such that their corresponding phase-space densities $\rho_\Lambda$ serve as classical thermal distributions.
We then obtain the state
\begin{equation}
    \frac{1}{\sqrt{N_\Lambda}}\sum_\Lambda \ket \Lambda \ket{\rho_\Lambda} \, ,
\end{equation}
where $\ket{\rho_\Lambda}$ is the state of the nuclear and environmental degrees of freedom produced by the evolution of the Liouvillian $L_\Lambda$. Under NVT Liouvillian equilibration, the state obtained by tracing out the bath registers $\ket{\bar s}$ and $\ket{\bar p_s}$ corresponds to a thermal state of the nuclear registers. Note that the heat bath variable $s$ should still be part of the description of the thermal state as the extended Hamiltonian with the Nos\'e thermostat is dependent on $s$. Therefore, $s$ is required for a description of the thermal state, and will have corresponding terms in Hamiltonians $H_A$ and $H_B$ to be interacted upon in the Hadamard test. Thus, we duplicate the contents of $\ket{\bar s}$ to a separate register by CNOT gates, which can be traced out while the original is part of the registers that the Hadamard test is applied to~\citep{simon2024improved}. We refer the reader to \cref{ap:preliminaries} for a more in-depth discussion as to why $\ket{s}$ is needed. Note that this is shown in \cref{fig:thermodynamic_integration_overview} as the additional CNOT gate from the $\ket{\bar s}$ to the $\ket{0}_s$ register.

To obtain the nuclear Hamiltonian differences $\Delta H = H_B-H_A$, we take block encodings of the nuclear Hamiltonian for both systems $A$ and $B$, where for each system, we have
\begin{equation}
    H_{\rm nuc} = H_{\mathrm{kin}} + H_{\mathrm{pot}} + H_{\mathrm{gse}} \, ,
\end{equation}
where the three terms correspond to diagonal block encodings of the values of the kinetic, potential, and ground state Hamiltonian of the relevant points in phase-space, respectively. While the kinetic and potential energy terms are previously efficiently implemented in Ref.~\citep{simon2024improved}, we draw attention to the implementation of $H_{\mathrm{gse}}$, which we achieve by preparing the electronic ground state for each point in the discretized grid space for position and applying a block encoding of the electronic Hamiltonian $H_{\tel}$ before uncomputing the ground states such that
\begin{align}
    H_{\rm gse} & = \sum_{\vec {\bar x}}E_{\tel}( \vec x)\ketbra{\bar {\vec x}}{\bar {\vec x}} \nonumber \\
    & = \sum_{\vec {\bar x}}\braket{\psi_0(\vec x)|H_{\tel}(\vec x)|\psi_0(\vec x)} \ketbra{\bar {\vec x}}{\bar {\vec x}} \, .
    \label{eqGSEInfml}
\end{align}
This encodes the ground state energy in the main nuclear registers by phase kickback in a similar fashion to our implementation of encoding electronic forces in \cref{secLS}. This approach to implementing $H_{\mathrm{gse}}$ eliminates the use of phase estimation, as shown in previous approaches~\citep{simon2024improved}, and reduces the runtime dependency on precision to a logarithmic scaling.

The difference in energies is then implemented by taking the difference between the block encoding of $H_A$ and $H_B$ via the LCU approach. Here, thermodynamic integration requires the assumption that systems $A$ and $B$ have a relatively high overlap in their nuclear phase spaces, and thus, most kinetic and potential Hamiltonian terms can be expected to cancel each other out. However, this is not the case for the ground state energy of the electronic Hamiltonians due to the separate ground state preparations for systems $A$ and $B$, and would still have a scaling factor of $\mathcal{O}(\widetilde N N_{\text{tot}})$.

In the following result, we construct the difference via LCU for simplicity, while noting that in practice, the block encoding implementation can be made more efficient and with a lower scaling factor.

\begin{proposition}[Block encoding of the nuclear Hamiltonian; informal]
    We can explicit construct a $(\alpha_{H}, a_{H}, \varepsilon_{H})$-block-encoding of the nuclear Hamiltonian $H_{\rm nuc}$ where
    \begin{align*}
        \alpha_{H}            & \in \mathcal O \left(N_{\mathrm{tot}}^2\right) \\
        a_{H} \in \mathcal{O} & \left(\widetilde{N} + \log\frac{ N_{\mathrm{tot}}}{\varepsilon_{\rm nuc}}\right)
    \end{align*}
    can be prepared  using $\mathcal{O}(\frac{1}{\delta})$ queries to $U_I$, and
    \begin{align*}
        \widetilde{\mathcal O} & \left(\frac{\widetilde N N_{\mathrm{tot}}^2}{\delta\gamma}\log^2 \left(\frac{1}{\varepsilon_{\rm nuc}}\right)\right)
    \end{align*}
    Toffoli gates.
    \label{propInHamInfml}
\end{proposition}
The formal statement (\cref{propNuclearHam}) and proof of this proposition can be found in \cref{appInHam}. Using our results above, we can block-encode $\Delta H = H_B-H_A$ via LCU with a worst-case scaling factor of $\alpha_{H_A} + \alpha_{H_B}$.  However, in practice, it may often be beneficial to block encode the difference between the two operators separately.

Given this block encoding, we extract the free energy value using the Hadamard test. The probability of measuring the $\ket{0}$ state on the Hadamard test ancilla register is
\begin{align}
    \mathbb{P}_{\Delta \widetilde F} &=\frac{1}{2N_\Lambda}\sum_\Lambda \left(\text{Re} \bra{\rho_\Lambda} U_{\Delta}  \ket{\rho_\Lambda} + 1\right)\nonumber\\
    & = \frac{1}{2N_\Lambda}\sum_\Lambda \left(\frac{\braket{H_B-H_A}_\Lambda}{\alpha_{\Delta}} + 1\right) \nonumber \\
    & =\frac{1}{2}\left(\frac{\Delta \widetilde F}{\alpha_{\Delta}} + 1\right) \, ,
\end{align}
giving us the value of the free energy difference. Then, instead of directly measuring the $\ket{0}$ state, we can use amplitude estimation~\citep{brassard2002quantum} on the result of the Hadamard test. We can obtain the estimation of $\Delta\widetilde F$ up to $\varepsilon$-precision by $\mathcal{O}(\frac{\alpha_{H_A} + \alpha_{H_B}}{\varepsilon})$ queries to the full Hadamard test circuit. We can treat the superposed Liouvillian simulations as quantum sampling access over the $\Lambda$-dependent classical thermal states. Thus, our approach achieves a scaling in terms of the precision $\varepsilon$ that is quadratically better than a classical thermodynamic integration modified to utilize Monte Carlo integration \citep{montanaro2015quantum} without considering potential improvements from implementing MD simulations on quantum computers as mentioned in \cref{secLS}.

\section{Discussion}
\label{secconclusion}

We present a quantum algorithm for simulating molecular dynamics based on the Liouvillian formalism of classical mechanics that achieves a super-polynomial improvement in the precision scaling over prior quantum algorithms. We then apply this improved algorithm to develop an efficient quantum algorithm for computing free energy differences via alchemical methods. 

While previous work~\citep{simon2024improved} treated the electronic and classical Liouvillians separately, here we introduce an efficient block encoding of the full (electronic + classical) Liouvillian. Our block encoding relies on phase kickback techniques to directly transfer forces and energies computed on electronic registers to the nuclear registers and therefore removes the need for phase estimation~\citep{kitaev1995quantum} to obtain electronic forces~\citep{obrien2022efficient}, or the need for Trotterization to interleave separate classical and electronic Liouvillian evolutions~\citep{simon2024improved}. This block encoding can then be used in a Hamiltonian simulation algorithm to provide a super-polynomial improvement in the scaling of the runtime of hybrid quantum-classical molecular dynamics with respect to the precision. Furthermore, by utilizing the Hellmann--Feynman theorem~\citep{hellmann1937einfuhrung, feynman1939forces} and direct block encodings of electronic force operators~\citep{obrien2022efficient}, we significantly decrease the dependency of the runtime on the number of particles compared to central finite difference methods. For a more detailed comparison that includes a breakdown of subroutines, please refer to \cref{tabOurRuntime,tabPRXRuntime} in \cref{appThermoInt}.

Our fundamentally improved Liouvillian simulation then enables us to develop a quantum algorithm for alchemical free energy calculation that achieves a $\widetilde{\mathcal O}(\frac{1}{\varepsilon})$ scaling with respect to precision. Our alchemical approach is a key enabler as it forgoes fundamentally more expensive entropy estimation, whose runtime on a quantum computer scales exponentially with the number of qubits representing the nuclear registers---and hence with the number of atoms---due to the exponential growth of the phase-space with the number of particles~\citep{gilyen2020distributional, simon2024improved}. Specifically, we use thermodynamic integration~\citep{Kirkwood1935_StatistiscalMechanics, frenkel_understanding_1996} to compute the free energy difference between two molecular ensembles directly. This is achieved by defining a thermodynamic path connecting two ensembles and integrating over the change of internal energy along the path.

Our approach performs this numerical integration coherently on a quantum computer: we simulate molecular dynamics at each interpolation point along the path by coherently evolving a superposition of phase-space densities, each under a distinct Liouvillian. A key algorithmic development here is that we construct methods to apply different polynomial transformations in superposition using both standard QSVT/GQSP as well as QSP without angle finding~\citep{alase2025quantum}, enabling efficient application of different transformations tailored to each block-encoded Liouvillian on the thermodynamic path. Finally, we extract the free energy difference via the Hadamard test followed by amplitude estimation. The gate overhead for this superposition of Liouvillian dynamics is only polylogarithmic in the number of interpolation points that discretize the thermodynamic path. 

Our approach can, in theory, provide exponential advantage relative to known classical methods with provable guarantees, such as full configuration interaction (FCI)~\citep{szabo1996modern}, for computing free-energy differences under the assumption that an oracle can be constructed that prepares a surrogate ground state with polynomial overlap with the ground state at every point in the thermodynamic integration and that the mixing time is polynomial. Examples of families of fermionic Hamiltonians where the ground state can be efficiently prepared can be found in Ref.~\cite{breuckmann2014space}, where it is shown that their ground state can encode the output of an arbitrary quantum circuit.

We briefly compare our methods with hybrid quantum–classical approaches that replace electronic-structure calculations with quantum primitives such as phase estimation~\citep{obrien2022efficient,gunther2025how}. First, we discuss the number of classical iterations and thus, the number of times we need to query the quantum computer for electronic structure calculations. We note that the classical molecular dynamics, at best, scales as $\mathcal{O}(\varepsilon^{-o(1)})$ in the number of time steps required. MD solves Newton’s equations iteratively, updating particle positions and momenta each timestep~\citep{frenkel_understanding_1996,allen1987computer}, and estimates phase-space densities via long-time averages assuming ergodic exploration. A common method used for MD is Verlet integration~\citep{verlet1967computer}, a second-order symplectic method with global trajectory error step count $\mathcal{O}(1/\sqrt{\varepsilon})$ to achieve final error $\varepsilon$. Under the Liouvillian formalism and Suzuki–Trotter viewpoint, this induces the same error order in the phase-space density~\citep{tuckerman1992reversible}, paralleling Ref.\citep{simon2024improved}. Higher-order symplectic integrators\citep{yoshida1990construction} can reduce the error such that the runtime required is $\mathcal{O}(\varepsilon^{-1/q})$ for some $q$, though such methods are rarely used in practice. 

Second, hybrid algorithms typically require fewer qubits and shorter circuit depths, given only the computation of the electronic structure problem solved with the quantum computer. However, reading out classical information at each timestep is necessary, which incurred a cost of  $\mathcal{O}(\frac{1}{\varepsilon})$ per timestep via phase estimation. Combining the above runtime estimates for readout and for classical MD integration, we conclude that hybrid algorithms are likely to require a higher asymptotic runtime than our approach. Specifically, we see that we can compute free energy differences using $\widetilde{\mathcal{O}}\left(\frac{N \widetilde N N_{\mathrm{tot}}^5 t_{\rm eq}}{\delta\gamma\varepsilon}\log\frac{1}{\xi}\right)$ Toffoli gates.

To provide a sense of the amount of required quantum resources, we sketch a back-of-the-envelope estimate for the number of qubits. We consider a $20$-atom molecule with $200$ electrons and calculate the number of qubits required in the main nuclear and electronic registers. We assume coordinates of all nuclear degrees of freedom are discretized into $g=2^{12}$ grid points (i.e., $12$ qubits per nuclear degree of freedom). Therefore, representing all nuclear degrees of freedom requires approximately $1500$ logical qubits, given each atom has 6 degrees of freedom (three position and three momentum), and that we also account for bath registers.
             
We further assume the wave function of each electronic degree of freedom is represented using  $B=2^{15}$ plane waves (i.e., $15$ qubits per electron and coordinate). Therefore, a total of $15\times 3 \times 200 = 9000$ logical qubits are required to represent all electrons. Including the negligible number of ancilla qubits (which are logarithmically bounded), we estimate that the number of logical qubits is on the order of ten thousand even for a small system -- this is likely beyond the capabilities of early fault-tolerant machines, however, reasonably within reach for realistic and mature quantum computers. On the other hand, the significant advantage of our approach is that scaling up the number of particles only requires a proportionally increased number of logical qubits.
             
Let us note, however, that estimating the number of required Toffoli gates is significantly more involved, given resource estimates typically report upper bounds, and we expect our bounds are loose --  however, it is beyond the scope of the present work to optimize our approach for explicit implementations.

Though our algorithms may be able to provide asymptotic speedups against hybrid quantum-classical methods and classical algorithms with provable guarantees, comparisons to well-established computational chemistry methods that have heuristic elements are much more unclear. In search for potential practical quantum advantages, the quantum algorithms presented here contain three main potential sources of asymptotic speedups over classical methods: (1) the replacement of iterative updates for MD simulation with a single evolution operator, improving on runtime dependencies of precision, (2) efficient coherent simulation of a large number of molecular dynamics trajectories for thermodynamic integration, and (3) efficient computation of the quantum dynamics of the electronic ground state and associated forces that can potentially be hard to simulate on a classical computer to guaranteed precision, where the last point is of particular value.

So far, we have not addressed the asymptotic behavior of the equilibration time $t_{\rm eq}$ with the error tolerance of the simulation. With completely integrable systems, the equilibration time can scale with $\mathcal{O}(\log\frac{1}{\varepsilon})$ with additional weighting over time averages~\citep{cances2005long}. However, while the Coulomb energy forces give rise to a multidimensional harmonic oscillator, an integrable system, the potential terms provided by the Born-Oppenheimer potential energy surface are generally not believed to be completely integrable (although empirical results show that this bound can be observed for some non-integrable systems~\citep{cances2004high}). It is therefore generally a challenge to bound this parameter for general systems, and such bounds may depend on the eigenstate thermalization hypothesis (ETH) when applied to open systems~\citep{deutsch2018eigenstate, shirai2020thermalization, chen2023fast,almeida2025universality}. We therefore treat $t_{\rm eq}$ as a potentially unbounded parameter in our scaling results.

We now briefly discuss factors that may affect $t_{\rm eq}$. First, properties of the energy landscape may substantially affect the equilibration time. In particularly challenging cases, the energy landscape may have deep basins of meta-stable regions where transitions between the meta-stable regions are separated by high energy barriers, which may lead to an increase in the thermalization time of classical MD, potentially by several orders of magnitude. This issue is potentially mitigated via the Liouvillian formalism that we use in the present manuscript, as we do not propagate a single configuration by Hamiltonian dynamics, but rather the phase space density instead. This alleviates some of the need for the simulation to actually ``cross'' the rare event threshold, but low-quality initialization of the phase space density may still lead to a significantly increased thermalization time. These challenges in our quantum algorithm may be tackled by further exploration of quantum counterparts for classical rare event processing of molecular dynamics~\citep{frenkel_understanding_1996}, such as the Bennett-Chandler approach~\citep{bennett1977molecular,chandler1978statistical}. In our paper, we mainly do not consider the occurrence of rare events with meta stability, given that classical literature with methods beyond na\"ive molecular dynamics as well. Under such circumstances, a polynomial overall algorithmic runtime holds in regimes where the thermostat dynamics mix in polynomial equilibration time $t_{\rm eq}$.

On the other hand, given our coupling of the system with a N\'ose thermostat, parameters intrinsic to the thermostat can be adjusted to affect the equilibration time. In particular, the effective mass $Q$ of the thermostat and the fixed temperature $T$ affect terms in the extended nuclear Hamiltonian $H_{\rm ext}$, and this in turn affects the eigenstates of said Hamiltonian. Per the dependency of the equilibration time on the eigenstate thermalization hypothesis, as mentioned before, we can then see that when the underlying Hamiltonian and eigenstates are altered, the amount of time to thermalize and equilibrate the system would be affected as well.

We now move on to our discussion of cases where our assumptions on the Born-Oppenheimer approximation fail, and the approach of simulating nuclei classically and electrons quantumly may not be applicable. An alternative quantum computational approach would involve representing both the nuclei and electrons as quantum particles and leveraging algorithms for the preparation of quantum thermal states~\citep{chen2023quantum, rall2023thermal}, and it may be more similar in nature to Monte Carlo methods. Apart from this additional ability to operate beyond Born-Oppenheimer approximations, it would be an interesting future research direction to compare our quantum algorithm, which uses a hybrid quantum-classical representation of molecules, with such fully-quantum approaches. Currently, it is not obvious which of the algorithms would have the most favorable scaling with respect to the number of particles and whether the progress on bounds on the mixing times of Lindbladians for specific cases~\citep{ramkumar2024mixing, rouze2024optimal, tong2025fast, smid2025polynomial} would yield an asymptotic advantage over the classical equilibration times in the Liouvillian framework we use here.

Looking forward, it will be important to identify concrete scenarios where the alchemical quantum algorithm offers an absolute runtime advantage over classical methods for free energy estimation. This requires obtaining fine-grained resource estimates on systems representative of practical applications to determine constant-level overheads. Performing such an analysis will require the construction of suitable block encodings for the molecular systems as well as a quantitative examination of potential implementations of an efficient oracle for approximate ground state preparation at all points in phase-space, including SAD, extended H\"uckel methods, and beyond.

Another important question left open by this work is that of deciding the ultimate limits on Helmholtz free-energy estimation for quantum systems.  At present, our work provides a worst-case asymptotic advantage over the results of Ref.~\cite{simon2024improved}, but this does not imply that the improvements that we have observed have led to the ultimate limits of free-energy estimation using thermodynamic integration. While we do not expect the dependency of precision $\varepsilon$ and spectral gap $\gamma$ to be improved due to fundamental limits---except when further structure of the problem can be assumed, the dependency on the particle number may be improved via improved representations and quantum algorithmic subroutines. In addition, a large runtime dependency on the particle number can be provided if the ground state preparation can be decoupled somehow from the time evolution, as repeated ground state preparations are the main bottleneck of our algorithm. Further, other improvements over exact computation of entropy differences~\cite{simon2024improved} may be attainable using Metropolis-Hastings ideas for thermalization~\cite{metropolis1949monte,metropolis1953equation, hastings1970monte} to approximately calculate free energies. We hope that future work will continue to make progress in improving the speed, accuracy, and range of applicability of methods for estimating both Gibbs and Helmholtz free energy differences across a wide range of physical systems.  

\section*{Acknowledgments}
    PWH, GB, DMD, HJ, TRB, and BK thank Zhu Sun, Zhenyu Cai, and Simon Benjamin for valuable discussions and their support during the project, as well as Annina Lieberherr for her comments on the manuscript and discussions on MD simulations. PWH would also like to thank Wei-Tse Hsu for discussions on classical MD simulations and rare events, as well as Jon Keating and Sakura Schafer-Nameki for comments on the initial preprint manuscript.  GLRA, MD, NM, RS, MS, and BR thank Clemens Utschig-Utschig for his comments on the manuscript and his support during the project.

    PWH acknowledges support from the Engineering and Physical Sciences Research Council (EPSRC) Doctoral Training Partnership (EP/W524311/1) with a CASE Conversion Studentship in collaboration with Quantum Motion. DMD acknowledges financial support from the EPSRC Hub in Quantum Computing and Simulation (EP/T001062/1). NW and SS's work was supported by support from Boehringer Ingelheim and the U.S. Department of Energy, Office of Science, National Quantum Information Science Research Centers, Co-design Center for Quantum Advantage (C2QA) under contract number DE-SC0012704 (PNNL FWP 76274). NW's research is also supported by PNNL’s Quantum Algorithms and Architecture for Domain Science (QuAADS) Laboratory Directed Research and Development (LDRD) Initiative. The Pacific Northwest National Laboratory is operated by Battelle for the U.S. Department of Energy under Contract DE-AC05-76RL01830. BK thanks UKRI for the Future Leaders Fellowship Theory to Enable Practical Quantum Advantage (MR/Y015843/1). BK also acknowledges funding from the EPSRC projects Robust and Reliable Quantum Computing (RoaRQ, EP/W032635/1) and Software Enabling Early Quantum Advantage (SEEQA, EP/Y004655/1).

\bibliography{main}
\clearpage

\appendix
\counterwithin{theorem}{section}
\counterwithin{proposition}{section}
\counterwithin{lemma}{section}
\counterwithin{assumption}{section}
\counterwithin{equation}{section}
\counterwithin{definition}{section}
\crefalias{section}{appendix}
\setcounter{figure}{0}
\renewcommand{\thefigure}{S\arabic{figure}}
\setcounter{table}{0}
\renewcommand{\thetable}{S\arabic{table}}

\makeatletter
\renewcommand{\theHfigure}{S\arabic{figure}}
\renewcommand{\theHtable}{S\arabic{table}}
\makeatother
\onecolumngrid
\begin{center}
\noindent{\large\bfseries Appendices for ``\textsl{Fullqubit alchemist: Quantum algorithm for alchemical free energy calculations}''}
\end{center}
\appendixtableofcontents
\vspace{2em}
\twocolumngrid

\section{Simulating molecular dynamics on a quantum computer}\label{ap:preliminaries}

In this section, we provide further details on simulating molecular dynamics on quantum computers, including definitions and implementation costs of relevant operators.

\subsection{Nuclear interactions under the NVT Liouvillian}

For operations involving the phase-space density on a quantum computer, we need to consider a finite, discretized phase-space.  Under the Born-Oppenheimer approximation, we separate the nuclear and electronic motion, and model them in separate phase spaces characterized by different sets of quantum registers. For the nuclear phase-space, we discretize the position and momentum and cast the values onto a grid such that
\begin{align}
    g_x     & := \frac{x_{\max}}{h_x} \in \mathbb{N} \\
    g_{p'}  & := \frac{p'_{\max}}{h_{p'}} \in \mathbb{N} \\
    g_s     & := \frac{s_{\max}}{h_s} \in \mathbb{N} \\
    g_{p_s} & := \frac{p_{s,\max}}{h_{p_s}} \in \mathbb{N}
\end{align}
where $x_{\max}/p_{\max}/s_{\max}/p_{s,\max}$ is the maximum value of any $x_{n,j}/p_{n,j}/s_{n,j}/{p_s}_{n,j}$ and $h_x/h_p/h_s/h_{p_s}$ is the size of the grid spacing. Boundary conditions are assumed to be periodic in our system for simplicity.

When implemented on a quantum computer, each grid point of the discretized phase-space corresponds to a computational basis state in the form of
\begin{equation}
    \ket{\vec {\bar x}, \vec {\bar p'}, \bar s, \bar p_s} := \bigotimes_{n, j} \Big( \ket{\bar x_{n,j}} \otimes \ket{ {\bar p'}_{n,j}} \Big) \otimes \ket{\bar s} \otimes \ket{\bar p_s}\, ,
\end{equation}
where $\bar x_{n,j} \in [g_x]$, $\bar p_{n,j}' \in [g_{p'}]$, $\bar s \in [g_s]$, $\bar p_s \in [g_s]$ such that $x_{n,j} = \bar x_{n,j} h_x$, $p_{n,j}' = \bar p_{n,j}' h_p$, $s = \bar s h_s$, and $p_s = \bar p_s h_{p_s}$.

To obtain the NVT Liouvillian, we note that the full formulation is as follows:
\begin{align}
    L^{\rm (NVT)}:=&-i\sum_{n=1}^{N}\sum_{j=1}^{3}\left(\frac{\partial H_{\rm ext}^{\rm (NVT)}}{\partial p'_{n,j}}\partial_{x_{n,j}}-\frac{\partial H_{\rm ext}^{\rm (NVT)} }{\partial x_{n,j}}\partial_{p'_{n,j}}\right)\nonumber \\
    &-i\left(\frac{\partial H_{\rm ext}^{\rm (NVT)}}{\partial p_s}\partial_{s}-\frac{\partial H_{\rm ext}^{\rm (NVT)}}{\partial s}\partial_{p_s}\right)\,.
\end{align}
Dropping the NVT notation, under the discretized phase-space, we obtain
\begin{align}
    L_{\rm disc}:=&-i\sum_{n=1}^{N}\sum_{j=1}^{3}\left(\frac{\partial H_{\rm ext}}{\partial p'_{n,j}} \otimes D_{x_{n, j}}-\frac{\partial H_{\rm ext}}{\partial x_{n,j}}\otimes D_{p'_{n, j}}\right)\nonumber\\
    &-i\left(\frac{\partial H_{\rm ext}}{\partial p_s}\otimes D_{s}-\frac{\partial H_{\rm ext}}{\partial s}\otimes D_{p_s}\right)\,,
\end{align}
where $D_{x_{n, j}}/D_{p'_{n, j}}/D_s/D_{p_s}$ are discrete derivative approximations obtained by central finite differences, and the partial derivatives on $H$ are analytically obtained from \cref{eqNoseHam} and applied to the corresponding registers. The implementation of the NVT Liouvillian, barring the terms related to the ground state energy $E_{\tel}(\vec x)$, is termed the classical Liouvillian and block-encoded in Ref.~\citep{simon2024improved}. Given the classical Hamiltonian $H_{\rm cl} = H_{\rm ext} - E_{\tel}$, we write the discretized classical Liouvillian
\begin{align}
L_{\rm cl, disc} :=& -i\sum_{n=1}^{N} \sum_{j=1}^{3} \left(  D_{x_{n,j}} \otimes \frac{\partial H_{\rm cl}}{\partial {p'_{n,j}}} - \frac{\partial H_{\rm cl}}{\partial {x_{n,j}}} \otimes D_{p'_{n,j}} \right)\nonumber\\
&-i \left(  D_{s} \otimes \frac{\partial H_{\rm cl}}{\partial_{p_{s}}} - \frac{\partial H_{\rm cl}}{\partial_{s}} \otimes D_{p_s} \right),
\end{align}
where
    \begin{align}
        \frac{\partial H_{\rm cl}}{\partial {x_{n,j}}} &= \begin{multlined}[t]
            \sum_{n' \neq n} \sum_{\bar x_{n}} \sum_{\bar x_{n'}} \frac{-Z_n Z_{n'}}{\left( \lVert x_n - x_{n'}\rVert^2 + \Delta^2 \right)^{3/2}}\\ \times(x_{n,j} - x_{n',j}) \ketbra{{\bar x_n}}{{\bar x_n}} \otimes \ketbra{{\bar x_{n'}}}{{\bar x_{n'}}}
        \end{multlined}\\
        \frac{\partial H_{\rm cl}}{\partial {p'_{n,j}}} &= \sum_{\bar p'_{n,j}} \sum_{\bar s} \frac{p'_{n,j}}{m_n \left( s + s_{\min} \right)^2} \ketbra{{\bar p'_{n,j}}}{{\bar p'_{n,j}}} \otimes \ketbra{\bar s}{\bar s}\\
        \frac{\partial H_{\rm cl}}{\partial {s}} &= \begin{multlined}[t]
            \sum_{\bar p_{n,j}} \sum_{\bar s} \frac{- {p'}_{n,j}^2}{m_{n} \left( s + s_{\min} \right)^3} \ketbra{{\bar p_{n,j}}}{{\bar p_{n,j}}} \otimes \ketbra{{\bar s}}{\bar s} \\
            + \sum_{\bar s} \frac{N_f k_B T}{s + s_{\min}} \ketbra{{\bar s}}{\bar s}
        \end{multlined}\\
        \frac{\partial H_{\rm cl}}{\partial p_{s}} &= \sum_{\bar p_{s}} \frac{p_{s}}{Q} \ketbra{{\bar p_{s}}}{\bar p_{s}}
    \end{align}
Thus, the electronic Liouvillian can then be found to be
\begin{equation}
L_{\tel} = L - L_{\rm cl} := i\sum_{n=1}^{N}\sum_{j=1}^{3}\frac{\partial E_{\tel}(\vec x)}{\partial x_{n,j}}\partial_{p_{n,j}'}
\end{equation}
We discuss this further in the next section.

We briefly note that, given that ``virtual'' rescaled momentum is used in this simulation, a distortion of dynamics in the aspect of time occurs. Across subsystems with different $s$ values, the length of the time steps is different, and is dependent on $s$~\citep{nose1984molecular}. However, as long as the entire system is under the ``virtual'' basis, this does not affect the classical thermal state preparation. It is only upon obtaining thermal averages of time-dependent operators, such as the evolution time or average momentum, with the thermal state, that this time distortion requires attention, and a rescaling by $s$ is required on the operator. For our algorithm, we obtain thermal averages on the nuclear Hamiltonian, which is time-independent. Hence, it suffices to describe everything under the ``virtual'' momentum and time picture for our algorithm.

To obtain thermal averages over the thermal state prepared by the simulation of the NVT Liouvillian, we integrate the degrees of freedom $s$ and $p_s$, which in the discretized quantum implementation corresponds to taking the sum of their values, or taking a partial trace over the corresponding registers. However, given that the thermal state is dependent on the true momentum $\vec p$ (or $\vec p'$ and $s$ in the virtual basis), the thermal state must have registers for both $\vec p'$ and $s$ in the virtual basis to constitute a true description. To both be able to trace out the $\ket{s}$ register and retain it for a thermal state description, we duplicate the contents of the $\ket{s}$ register into an ancilla register $\ket{0}_s$ such that we have $\ket{s}\otimes\ket{s}$. We then trace out one of the registers and retain the other one for the thermal state distribution. Note that for our alchemical free energy calculation algorithm, the nuclear Hamiltonian, discretized from the Nos\'e Hamiltonian, does indeed have $s$-related terms, hence would require a corresponding $\ket{s}$ register to be applied to.

\begin{table*}
    \centering
    \caption{Parameters that determine the complexity of our quantum algorithm for simulating NVT Liouvillian dynamics in the Born-Oppenheimer approximation.}
    \begin{tabular}{lll}
        \toprule
        Description                                                                & \\
        \midrule
        Evolution time                                                             & $t$ \\
        Desired precision                                                          & $\varepsilon$ \\
        Failure probability                                                        & $\xi$ \\
        Number of nuclei and electrons                                             & $N$, $\widetilde{N}$ \\
        Mass of the lightest nucleus                                               & $m_{\min}$ \\
        Maximum atomic number over all nuclei                                      & $Z_{\max}$ \\
        Maximum value of a component of the nuclear position vectors               & $x_{\max}$ \\
        Maximum value of a component of the (virtual) momentum vectors             & $p'_{\max}$ \\
        Grid spacing for a component of the discretized variables                  & $h_x$, $h_{p'}$, $h_s$, $h_{p_s}$ \\
        Order of the finite difference scheme used for approximating derivatives   & $d_x$, $d_{p'}$, $d_s$, $d_{p_s}$ \\
        Gap parameter to regularize the Coulomb potential                          & $\Delta$ \\
        Number of plane wave basis functions in the electronic Hamiltonian         & $B$ \\
        Inverse grid spacing for a component of the electronic wave number         & $h_{\tel}$ \\
        Lower bound on the overlap of the initial and true electronic ground state & $\delta$ \\
        Lower bound on the spectral gap of $H_{\tel}$                              & $\gamma$ \\
        Number of phase-space grid points                                          & $\eta$ \\
        Number of degrees of freedom of the physical system                        & $N_f$ \\
        Temperature of the heat bath                                               & $T$ \\
        Mass parameter associated with the heat bath                               & $Q$ \\
        Minimum value of the bath position variable                                & $s_{\min}$ \\
        Maximum value of the bath momentum variable                                & $p_{s, \max}$ \\
        Coupling parameter between systems for thermodynamic integration           & $\Lambda$ \\
        Number of interpolation points in thermodynamic integration                & $N_\Lambda$ \\
        \bottomrule
    \end{tabular}
    \label{tab:NVT_parameters}
\end{table*}

\subsection{Quantum electronic interactions on Born-Oppenheimer potential energy surface}
We now discuss the quantum electronic interactions on the ground state potential energy surface under the Born-Oppenheimer approximation. We provide the formal definition of the electronic Liouvillian for an NVT ensemble under the discretized grid system:

\begin{definition}[Electronic Liouvillian]
    Let $\Del$ be as follows:
    \begin{equation*}
        \Del = \sum_{\vec{\bar x}}  \frac{\partial E_{\tel}(\vec x)}{\partial x_{n,j}} \ketbra{\vec{\bar x}}{\vec{\bar x}}\, .
    \end{equation*}
    and $D_{p'_{n,j}}$ be a discrete derivative approximation to $\partial_{p'_{n,j}}$ obtained by central finite difference as follows:
    \begin{equation*}
        D_{p'_{n,j}} = \frac{1}{h_{p'}} \sum_{k=-d_{p'}}^{d_{p'}}c_{d_{p'}, k} \ketbra{\bar{p}'_{n,j}-k}{\bar{p}'_{n,j}} \, .
    \end{equation*}
    In the NVT ensemble, the discretized electronic Liouvillian acting on the nuclei is given by
    \begin{equation*}
        L_{\tel,\, \mathrm{disc}} := i\sum_{n=1}^{N} \sum_{j=1}^{3} \Del \otimes D_{p'_{n,j}} \, .
    \end{equation*}

    Here $x_{n,j}$ and $p'_{n,j}$ are the $j$-th components of position and associated virtual momentum of the $n$-th nucleus, respectively. Errors from this discretization scheme are discussed in \cref{appDiscErr}.
    \label{defElecLiou}
\end{definition}

In the electronic Liouvillian, we require the operator $\Del$, which encodes the electronic forces on the ground state energy surface in \cref{defElecLiou}. To compute such forces, we require an extra set of electronic registers and describe the electronic system by evaluating the electronic Hamiltonian acting on electronic orbitals at different points of the nuclear grid space, which is then kicked back onto the nuclear registers.

For simulations of the electronic Hilbert space, instead of utilizing a grid discretization, we use a finite set of $B$ plane waves as the basis, which takes the following form:
\begin{equation}
    \phi_{\vec b} (\vec r) := \frac{1}{\sqrt{\Omega}} e^{-i \vec {\kappa_b} \cdot \vec r} \,
\end{equation}
where $\vec r$ is a vector in position space and $\vec {\kappa_b} = \frac{2 \pi \vec b}{\Omega^{1/3}}$ is a wave vector in reciprocal space where $\vec b$ is a vector constrained to the cube $G := \big[ -\frac{B^{1/3}-1}{2}, \frac{B^{1/3}-1}{2} \big]^3$, and $\Omega \in \Theta \left( B \, h_{\tel}^3 \right)$ is the computational cell volume where $\frac{1}{h_{\tel}}$ is the grid spacing in reciprocal space.

Under first quantization, the electronic basis states can then be denoted as $\ket{b_0} \ket{b_1} \cdots \ket{b_{\widetilde{N}-1}}$, where each $\ket{b_k}$ is a quantum register of size $\lceil \log{B} \rceil$ specifying the index $b \in [B]$ of the plane wave basis that is occupied by electron $k$. Under this representation, the electronic Hamiltonian can be efficiently computed from the nuclear position registers and the plane wave momenta through quantum arithmetic circuits~\citep{babbush2019quantum,su2021faulttolerant}.

Given $G_0:= [-B^{1/3}, B^{1/3}]^3 \subset \mathbb{Z}^3 \setminus \{(0,0,0)\}$, the block encoding of the electronic Hamiltonian is provided as follows:

\begin{lemma}[Block encoding of the electronic Hamiltonian -- Lemma 1, \citep{su2021faulttolerant}; Formulation from Lemma 3, revised, \citep{simon2024improved}]
    There exists a Hermitian $(\lambda, a_{\tel}, \varepsilon)$-block-encoding of the discretized electronic Hamiltonian $H_{\tel}$ in a plane wave basis
    \begin{align*}
        H_{\tel} \left( \vec x \right) & := \sum_{p=1}^{\widetilde{N}}  \sum_{\vec b \in G} \frac{\left\lVert \vec{\kappa_b}\right\rVert^2}{2} \ketbra{\vec b}{\vec b}_{p} \\
        & -\frac{4\pi}{\Omega}\sum_{n=1}^{N}\sum_{p=1}^{\widetilde{N}}\sum_{\substack{\vec b,\vec c\in G \\ \vec b\neq \vec c}}\bigg(Z_n \frac{e^{i\vec{\kappa_{c-b}}\cdot \vec{x_n}}}{\lVert \vec{\kappa_{b-c}}\rVert^2}\bigg)\ketbra{\vec b}{\vec c}_p \\
        & \begin{multlined}
            +\frac{2 \pi}{\Omega} \sum_{p\neq q=1}^{\widetilde{N}}\sum_{\vec b, \vec c \in G} \sum_{\substack{\vec \nu\in G_0\\(\vec b+\vec \nu)\in G\\(\vec c-\vec \nu)\in G}}\frac{1}{\left\lVert \vec{\kappa_{\nu}}\right\rVert^2} \\
            \times\ketbra{\vec b + \vec \nu}{\vec b}_p \ketbra{\vec c-\vec \nu}{\vec c}_q \, ,
          \end{multlined}
    \end{align*}
    where
    \begin{equation*}
        \lambda \in \mathcal O \left( \frac{\widetilde N }{h_{\tel}^2} + \frac{N \widetilde N Z_{\max}}{h_{\tel}} + \frac{\widetilde N^2}{h_{\tel}} \right),
    \end{equation*}
    and
    \begin{equation*}
        a_{\tel} \in \mathcal O \left( \log\frac{N \widetilde N B}{\varepsilon}  \right).
    \end{equation*}
    This block encoding can be implemented using
    \begin{equation*}
        \mathcal O \left( N + \log B \left(\widetilde{N} + \log\frac{B g_x}{\varepsilon} \right)\right)
    \end{equation*}
    Toffoli gates.
    \label{lemBEElecHam}
\end{lemma}
Note that $\vec x_n$ is the nuclear position in grid space. By combining the above results with a ground state preparation algorithm~\citep{lin2020nearoptimal}, one can obtain the ground state corresponding to this electronic Hamiltonian and further use this to characterize operations on the ground state surface.

On the other hand, we use the block encoding of the force operator in the Hellmann--Feynman theorem to obtain the derivative of the ground state surface and implement $\Del$. The block encoding of such force operators in first quantization under the plane wave basis is shown in Section IV C of Ref.~\citep{obrien2022efficient}, and repeated as follows:

\begin{lemma}[Block encoding of $\frac{\partial H_\tel}{\partial x_{n, j}}$ \citep{obrien2022efficient}]
    The electronic Hamiltonian derivative/force operator
    \begin{multline*}
        \frac{\partial H_{\tel}}{\partial x_{n, j}} := \sum_{\vec \nu \in G_0} \frac{2\pi Z_n \kappa_{\nu, j}}{\Omega \|\vec{\kappa_\nu}\|^2} \sum_{p=1}^{\widetilde N} \sum_{\ell \in \{0,1\}} \bigg( ie^{-i \vec{\kappa_\nu} \cdot \vec{x_n}} \\
        \times \sum_{\vec b \in G} (-1)^{\ell[(\vec b - \vec \nu) \notin G]} |\vec b - \vec \nu\rangle \langle \vec b|_p \bigg)
    \end{multline*}
    can be implemented as a $(\tau_n, a_{\mathrm{force}}, \varepsilon_{\mathrm{force}})$-block-encoding, where
    \begin{align*}
        \tau_n             & \in \mathcal O \left( \frac{\widetilde NZ_n}{h_{\tel}^2}\right) \\
        a_{\mathrm{force}} & \in \mathcal O \left(\log\frac{\widetilde N B}{\varepsilon_{\mathrm{force}}}\right) \, .
    \end{align*}
    This block encoding can be implemented using $\mathcal O\left(\log B\left(\widetilde N +\log\frac{B g_x}{\varepsilon_{\mathrm{force}}}\right)\right)$ Toffoli gates.
    \label{lemHamDeriv}
\end{lemma}

These results are used as basic building blocks that construct the block encoding of the Liouvillian as shown in \cref{appLS}.

Lastly, we establish an index of variable names that correspond to various physical quantities that are used to implement the Liouvillian operator. The parameters are summarized in \cref{tab:NVT_parameters}.

\section{The quantum linear algebra toolbox}
\label{appQLA}
In this section, we provide definitions and useful theorems from the quantum linear algebra toolbox, composing of details on block encodings~\citep{low2019hamiltonian, chakraborty2019power}, the quantum singular value transformation (QSVT)~\citep{gilyen2019quantum}, its applications in Hamiltonian simulation~\citep{low2019hamiltonian,gilyen2019quantum,martyn2021grand} and ground state preparation~\citep{lin2020nearoptimal}, as well as quantum signal processing (QSP) without angle finding~\citep{alase2025quantum}. For the numbering of definitions and theorems regarding Ref. \citep{gilyen2019quantum}, we refer to the full version on arXiv instead of the shorter published version due to some results we use being omitted in the published version.

\subsection{Block encodings}
We first review the constructions of block encodings and their operations.

\begin{definition}[Block encoding -- Definition 43,~\citep{gilyen2019quantum}, see also~\citep{low2019hamiltonian, chakraborty2019power}]\label{defBlockEncoding}
    Let $A\in \mathcal{M}_{2^n\times2^n}(\mathbb{C})$, $\alpha, \varepsilon \in \mathbb{R}_+$ and $a \in \mathbb{N}$. We say that the $(n+a)$-qubit unitary $U$ is a $(\alpha,a,\varepsilon)$\emph{-block-encoding} of $A$ if
    \begin{equation*}
        \|A-\alpha(\bra{0}_a\otimes \mathbb{I}_n)U(\ket{0}_a\otimes \mathbb{I}_n)\| \leq \varepsilon \, .
    \end{equation*}
\end{definition}

Given block encodings of operators $A_i$, we can construct a block encoding of their linear combination $A=\sum_i y_iA_i$ by using the linear combination of unitaries (LCU)~\citep{childs2012hamiltonian} framework as well as a construction known as a ``state-preparation pair''. Recall that $\|\cdot\|_1$ is the $\ell_1$ norm.

\begin{definition}[State preparation pair -- Definition 51,~\citep{gilyen2019quantum}]
    \label{defStatePrepPair}
    Let $\vec{y} \in \mathbb{C}^m$ and $\lVert\vec{y}\rVert_1 \leq \beta$. The pair of unitaries $(P_L,P_R)$ is called a ($\beta,b,\varepsilon_{\mathrm{SP}}$)\emph{-state-preparation-pair} for $\vec{y}$ if
    \begin{align*}
        P_L\ket{0^b} = \sum_{j=1}^{2^b} c_j\ket{j} \, , \quad
        P_R\ket{0^b} = \sum_{j=1}^{2^b} d_j\ket{j} \, ,
    \end{align*}
    such that $\sum_{j=1}^{m} \lvert y_j-\beta c_j^*d_j\rvert \leq \varepsilon_{\mathrm{SP}}$ and $c_j^*d_j = 0$ for $j=m+1,\dots,2^b$.
\end{definition}
In essence, a state preparation pair encodes a set of weights $\vec{y}$ in the first $m$ elements of a length-$2^b$ column vector whose elements are $c_j^*d_j$, up to an error of $\varepsilon_{\mathrm{SP}}$, where $\beta$ is a normalization factor. To realize the LCU implementation, one would apply $P_R$ on a set of ancilla qubits, then control upon the ancillae, apply controlled versions of the block encoding of $A_i$, and then apply $P_L$ on the ancillae to produce a block encoding of $A=\sum_i y_iA_i$.

For simplification purposes, in the case where $\vec{y} \in \mathbb{R}^m$, we assume that $P = P_L = P_R$ and term $P$ as the \emph{state preparation unitary}. The following lemma improves upon the results of Lemma 5 of Ref.~\citep{simon2024improved}.
\begin{lemma}[Error and cost of the state preparation unitary]
    Let $\vec{y} \in \mathbb{R}_+^m$ and $\|\vec{y}\|_1 \leq \beta$. Given state preparation unitary $P$ such that $\left\lVert P\ket{0} - \sum_j \sqrt{\frac{y_j}{\beta}} \ket{j}\right\rVert \le \frac{\varepsilon_{\mathrm{SP}}}{2\beta}$, then $(P, P)$ is $(\beta, b, \varepsilon_{\mathrm{SP}})$-state-preparation-pair for $\vec y$. $P$ can be constructed with $\mathcal O(m\log\frac{\beta}{\varepsilon_{\mathrm{SP}}})$ Toffoli gates.
    \label{lemStatePrep}
\end{lemma}
\begin{proof}
    From \cref{defStatePrepPair}, we know the error of the state preparation pair is found by upper bounding $\sum_{j=1}^{m} \lvert y_j-\beta c_j^2\rvert$. Note that $\sum_{j=1}^{m} c_j^2 = 1$. Then
    \begin{align}
        \sum_{j=1}^{m} \lvert y_j-\beta c_j^2\rvert & = \beta \sum_{j=1}^{m} \left\lvert \left(\sqrt{\frac{y_j}{\beta}}- c_j\right)\left(\sqrt{\frac{y_j}{\beta}}+ c_j\right)\right\rvert \nonumber \\
        & = \beta \sum_{j=1}^{m} \left\lvert\sqrt{\frac{y_j}{\beta}}- c_j\right\rvert\left\lvert\sqrt{\frac{y_j}{\beta}}+ c_j\right\rvert\nonumber \\
        & \le \beta \sum_{j=1}^{m} \left\lvert\sqrt{\frac{y_j}{\beta}}- c_j\right\rvert \left(\left\lvert\sqrt{\frac{y_j}{\beta}}\right\rvert + \left\lvert c_j\right\rvert\right) \, .
    \end{align}
    Let $\vec v, \vec w, \vec u \in \mathbb{R}_+^m$ such that $v_j = \left\lvert\sqrt{\frac{y_j}{\beta}}- c_j\right\rvert$, $w_j = \left\lvert\sqrt{\frac{y_j}{\beta}}\right\rvert$, and $u_j = \lvert c_j\rvert$. Then we can write the above as
    \begin{align}
        \label{eqBoundSP}
        \sum_{j=1}^{m} \lvert y_j-\beta c_j^2\rvert & \le \beta \vec v \cdot (\vec w + \vec u) \le \beta \lVert \vec v\rVert (\lVert \vec w + \vec u\rVert) \nonumber \\
        & \le \beta \lVert \vec v\rVert (\lVert \vec w \rVert + \lVert \vec u\rVert) \nonumber \\
        & = \beta \lVert \vec v\rVert \left(\sqrt{\frac{\lVert \vec y \rVert_1}{\beta}} + \sqrt{\sum_{j=1}^{m} c_j^2}\right) \le 2\beta \lVert \vec v\rVert \, .
    \end{align}
    Now writing out $\lVert \vec v\rVert$ we get the following equality:
    \begin{align}
        \lVert \vec v\rVert & = \left\lVert \sum_{j=1}^m \left\lvert\sqrt{\frac{y_j}{\beta}}- c_j\right\rvert\ket{j}\right\rVert \nonumber \\
        & = \left\lVert \sum_{j=1}^m \sqrt{\frac{y_j}{\beta}}\ket{j}- \sum_{j=1}^m c_j\ket{j}\right\rVert \nonumber \\
        & = \left\lVert \sum_{j=1}^m \sqrt{\frac{y_j}{\beta}}\ket{j}- P\ket{0}\right\rVert \, .
    \end{align}
    Thus, setting $\left\lVert P\ket{0} - \sum_j \sqrt{\frac{y_j}{\beta}} \ket{j}\right\rVert \le \frac{\varepsilon_{\mathrm{SP}}}{2\beta}$, we find that $\sum_{j=1}^{m} \lvert y_j-\beta c_j^2\rvert \le \varepsilon_{\mathrm{SP}}$, making $(P, P)$ a $(\beta, b, \varepsilon_{\mathrm{SP}})$-state-preparation-pair.

    Lastly, given a general quantum state preparation process, we know that preparing a $\varepsilon$-close state to an $m$-length quantum state requires $\mathcal O(m\log\frac{1}{\varepsilon})$ Toffolis~\citep{nielsen2010quantum}. Plugging in relevant error bounds we see that $\mathcal O(m\log\frac{\beta}{\varepsilon_{\mathrm{SP}}})$ Toffolis are required.
\end{proof}
The consequences of this improved lemma are that all complexity results in Ref.~\citep{simon2024improved} can be slightly reduced by a logarithmic factor. Still, under runtime analyses that hide the poly-logarithmic factors, this has no impact.

The following lemma, originally found in Ref.~\citep{gilyen2019quantum}, contains a minor typo in the error term, which we have corrected.

\begin{lemma}[Linear combination of block encodings -- Lemma 52, revised,~\citep{gilyen2019quantum}]
    \label{lemLCU}
    Let $A = \sum_{j=1}^m y_jA_j$ be an $n$-qubit operator and $\varepsilon \in \mathbb{R}^+$. Suppose that $(P_L, P_R)$ is a $(\beta, b, \varepsilon_1)$-state-preparation-pair for $\vec{y}$ and
    \begin{equation*}
        W = \sum_{j=1}^{m} |j\rangle\langle j| \otimes U_j + \left(\left(\mathbb{I} - \sum_{j=1}^{m} |j\rangle\langle j|\right) \otimes \mathbb{I}_a \otimes \mathbb{I}_n\right)
    \end{equation*}
    is an $n + a + b$ qubit unitary such that for all $j \in {1, \dots, m}$ we have that $U_j$ is a $(\alpha, a, \varepsilon_2)$-block-encoding of $A_j$. Then we can implement a $(\alpha\beta, a + b, \alpha\varepsilon_1 + \beta\varepsilon_2)$-block-encoding of $A$, with a single use of $W$, $P_R$ and $P_L^\dagger$.
\end{lemma}

We can also construct a block encoding of a product of two block-encoded matrices.
\begin{lemma}[Product of block encodings -- Lemma 53,~\citep{gilyen2019quantum}]
    \label{lemProd}
    Let $U_A$ be a $(\alpha,a,\varepsilon_A)$-block-encoding of $A$ and $U_B$ be a $(\beta,b,\varepsilon_B)$-block-encoding of $B$, where $A,B$ are $n$-qubit operators. Then, $(\mathbb{I}_b \otimes U_A)(\mathbb{I}_a \otimes U_B)$ is a $(\alpha\beta,a+b,\alpha\varepsilon_B+\beta\varepsilon_A)$-block-encoding of $AB$.
\end{lemma}

\subsection{Quantum singular value transformation and its applications}
Given a block encoding $U_A$ of a Hermitian matrix $A$, one can use the QSVT~\citep{gilyen2019quantum} to construct a block encoding of a polynomial transformation $f$ of the matrix $A$, by applying the polynomial to its eigenvalues. More generally, in the case where $A$ is not Hermitian, the polynomial transformation is applied to the singular values of the matrix $A$.

\begin{lemma}[Quantum polynomial eigenvalue transformation -- Theorem 56, \citep{gilyen2019quantum}]
    \label{lemArbParity}
    Suppose that $U_A$ is a $(\alpha,a,\varepsilon)$-encoding of a Hermitian matrix $A$. Given $\xi\geq 0$ and a real $\mathcal{D}$-degree polynomial $ f\in\mathbb{R}[x]$ satisfying that
    \begin{equation*}
        \forall x\in[-1,1]:\,|f(x)|\leq \frac{1}{2} \, ,
    \end{equation*}
    one can prepare a $(1,a+2,4\mathcal{D}\sqrt{\varepsilon/\alpha}+\xi)$-block-encoding of $f(A/\alpha)$ by using $\mathcal{D}$ queries to $U_A$ or its inverse, a single query to controlled-$U_A$ and an additional $\mathcal O((a+1)\mathcal{D})$ other one- and two-qubit gates.
    Further, one can compute a description of such a circuit classically in time $\mathcal O(\poly(\mathcal{D},\log(1/\xi)))$.
\end{lemma}
\begin{remark}
    For polynomials of definite parity, a $(1,a+1,4\mathcal{D}\sqrt{\varepsilon/\alpha}+\xi)$-block-encoding of $f(A/\alpha)$ can be constructed. This result can be seen from Corollary 18 of Ref.~\citep{gilyen2019quantum}, or indeed, Theorem 2 of its published version.
\end{remark}
Applications of the quantum eigenvalue transformation by QSVT lead to an optimal implementation of the Hamiltonian simulation task by producing a polynomial approximation to the function $\exp(-iHt)$ and applying it to a Hamiltonian $H$. The robust version of the algorithm concerning block encoding errors is as follows:

\begin{lemma}[Robust Hamiltonian simulation -- Corollary 62, \citep{gilyen2019quantum}]
    \label{lemHamSim}
    Let $t\in\mathbb{R}$, $\varepsilon\in(0,1)$ and let $U_H$ be a $(\alpha,a,\frac{\varepsilon}{|2t|})$-block-encoding of the Hamiltonian $H$. One can then implement a $(1,a+2,\varepsilon)$-block-encoding of $e^{-iHt}$ with
    \begin{equation*}
        6\alpha |t| + 9 \ln \frac{12}{\varepsilon}
    \end{equation*}
    queries of $U_H$ or its inverse, 3 queries of controlled-$U_H$ or its inverse, and an additional $\mathcal O \left(a\left(\alpha|t|+\log\frac{1}{\varepsilon})\right)\right)$ two-qubit gates.
\end{lemma}
\begin{corollary}
    Alternatively, one can use an extra
    $\mathcal O \left(\left(\alpha|t|+\log\frac{1}{\varepsilon}\right)\log\frac{\alpha |t|}{\varepsilon}\right)$ Toffoli gates in place of the arbitrary two-qubit gates.
    \label{corHamSimTof}
\end{corollary}
\begin{proof}
    From the ancilla-controlled phase shift operators, we know that a continuous rotation gate is applied per query of $U_H$, in addition to $\mathcal O(1)$ $m$-ancilla controlled NOT gates.

    We first prepare a unitary using \cref{lemHamSim} to prepare a Hamiltonian simulation up to $\frac{\varepsilon}{2}$-precision. Then, using $\mathcal O(\log \frac{1}{\varepsilon_{\text{gate}}})$ Toffolis to prepare a continuous rotation gate up to $\varepsilon_{\text{gate}}$-precision in the spectral norm difference of the resulting unitary, we note that by triangular inequality, replacing all continuous rotation gates with Toffolis will result in an error in spectral norm difference of $\mathcal O\left( (\alpha |t| + \log \frac{1}{\varepsilon}) \varepsilon_{\text{gate}}\right)$. Capping the upper bound by $\frac{\varepsilon}{2}$, we note that $\frac{1}{\varepsilon_{\text{gate}}}\in \mathcal O\left( \frac{\alpha |t| \log \frac{1}{\varepsilon}}{\varepsilon}\right)$.

    The number of Toffolis required is then found to be $\mathcal O \left(\left(\alpha|t|+\log\frac{1}{\varepsilon}\right)\log\frac{\alpha |t|}{\varepsilon}\right)$.
\end{proof}

Further, QSVT can be used for ground state preparation~\citep{lin2020nearoptimal} by preparing a reflection operator that flips the eigenstates greater than the ground state upper bound $\mu$ and uses amplitude amplification~\citep{brassard2002quantum} to amplify the ground state. The reflector is constructed by applying the following threshold function to the Hamiltonian's eigenvalues:
\begin{equation}
    f_{\text{thresh}}(x) = \begin{cases}
        -1 & x > \mu \\
        1  & x \leq \mu \, .
    \end{cases}
    \label{eqThresh}
\end{equation}
This produces the block encoding
\begin{equation}
    \label{eqReflect}
    R_\mu(H) = \sum_{k:E_k\leq\mu}\ketbra{\psi_k}{\psi_k}- \sum_{k:E_k>\mu}\ketbra{\psi_k}{\psi_k}
\end{equation}
of the phase oracle used for amplitude amplification. When $\mu$ is set to be between the ground state energy and the first excited state energy, $R(H)$ reflects all other states apart from the ground state and can be used to amplify the latter. We prove and show a robust version of Theorem 6 in Ref.~\citep{lin2020nearoptimal} where we do not assume perfect block encoding of the Hamiltonian and do not consider the implementation of perfect continuous rotation gates. If the error of the block encoding is much smaller than that of the Hamiltonian's spectral gap, then \cref{eqThresh} can still be used to successfully filter out the ground state energy despite its perturbations. To start, we require the following lemma:

\begin{lemma}[Polynomial approximations of the sign function -- Lemma 10,~\citep{low2017hamiltonian}; Formulation from Lemma 3,~\citep{lin2020nearoptimal}]
    \label{lemSign}
    For all $0 < \gamma < 1$, $0 < \xi < 1$, there exists an efficiently computable odd polynomial $S(x; \gamma, \xi) \in \mathbb{R}[x]$ of degree $\mathcal{D} = \mathcal O(\frac{1}{\gamma} \log\frac{1}{\xi})$, such that
    \begin{enumerate}
        \item \textit{for all $x \in [-1, 1]$, $\lvert S(x; \gamma, \xi)\rvert \le 1$, and}
        \item \textit{for all $x \in [-1, -\gamma] \cup [\gamma, 1]$, $\lvert S(x; \gamma, \xi) - \operatorname*{sign}(x)\rvert \le \xi$.}
    \end{enumerate}
\end{lemma}

\begin{proposition}[Robust ground state preparation with \textit{a priori} ground state energy bound]
    Given a Hamiltonian $H = \sum_k E_k \ketbra{\psi_k}{\psi_k} \in \mathbb{C}^{N \times N}$, where $E_k \leq E_{k+1}$ and whose ground state energy and the spectral gap obeys the following inequality: $E_0 \leq \mu - \gamma/2 < \mu + \gamma/2 \leq E_1$, where $\mu$ is an upper bound on the ground state energy and $\gamma$ is a lower bound on the spectral gap of $H$. Suppose the Hamiltonian is given through its $(\lambda, m, \varepsilon_H)$-block-encoding $U_H$, where
    \begin{equation*}
        \varepsilon_H\in \widetilde{\mathcal O} \left(\frac{\delta^2\gamma^2\varepsilon_{\mathrm{prep}}^2}{\lambda}\right),
    \end{equation*}
    and $\varepsilon_H \le \frac{\gamma}{4}$.
    Also suppose we have an initial state $\ket{\phi_{\rm init}}$, prepared by a unitary $U_I$, together with a lower bound on the overlap $\lvert\braket{\psi_0|\phi_{\rm init}}\rvert \geq\delta$. Then the ground state $\ket{\psi_0}$ can be prepared with fidelity at least $1-\varepsilon_{\mathrm{prep}}$ using
    \begin{equation*}
        \mathcal O \left( \frac{\lambda}{\delta \gamma} \log\frac{1}{\delta \varepsilon_{\mathrm{prep}}} \right)
    \end{equation*}
    queries to $U_H$,
    \begin{equation*}
        \mathcal O \left( \frac{1}{\delta} \right)
    \end{equation*}
    queries to $U_I$, and
    \begin{equation*}
        \mathcal O \left( \frac{m\lambda}{\delta \gamma} \log \frac{1}{\delta \varepsilon_{\mathrm{prep}}}\right)
    \end{equation*}
    one- or two-qubit gates.
    \label{lemGSPrep}
\end{proposition}
\begin{proof}
    The main algorithm and correctness of results follow from Theorem 6 of Ref.~\citep{lin2020nearoptimal}. We now reiterate parts of the proof and make relevant adjustments so that the error in the block encoding is considered. The first step is the shift the block encoding by a factor of $\mu$ and from \cref{lemLCU}, we note that we can obtain a $(1, m + 1, \frac{\varepsilon_H}{\lambda})$-block-encoding of $\frac{H-\mu I}{\lambda + |\mu|}$. The next step is then to apply a sign function to $\frac{H-\mu I}{\lambda + |\mu|}$, but before doing so, there are a few things to note. First, due to the error in the Hamiltonian block encoding, all eigenenergies of the Hamiltonian $H$, and by extension, $H - \mu I$, may be shifted. Let $\widetilde H = \sum_k \widetilde E_k \ketbra{\widetilde \psi_k}{ \widetilde \psi_k}$ be the actual Hamiltonian that is block-encoded. By Weyl's inequality~\citep{weyl1912das}, we note the difference in eigenenergies can be upper bounded as follows:
    \begin{align}
        \lvert\widetilde E_k - E_k\rvert \le \|\widetilde H - H\| \le \varepsilon_H.
    \end{align}
    Let $\varepsilon_H \le \frac{\gamma}{4}$. Then by reverse triangle inequality, we can find that the spectral gap for $\frac{H-\mu I}{\lambda + |\mu|}$ is lower bounded by $\frac{\gamma}{2(\lambda + |\mu|)}$, which can be further lower bounded by $\frac{\gamma}{4\lambda}$.

    Then, applying the polynomial approximation of the sign function $S(x; \frac{\gamma}{8\lambda}, \xi)$ where for all $x \in [-1, -\frac{\gamma}{8\lambda}] \cup [\frac{\gamma}{8\lambda}, 1]$, $\lvert S(x; \frac{\gamma}{8\lambda}, \xi) - \operatorname*{sign}(x)\rvert \le \xi$, we can still produce a reflection operator that reflects the ground state of $\widetilde H$ up to error $\xi$. We now quantify the error on the actual Hamiltonian $H$ with the assurance that the error in the block encoding is small enough not to affect the results of applying the sign function. By \cref{lemSign}, we note that the polynomial approximation should be taken to degree $\mathcal{D} \in \mathcal O (\frac{\lambda}{\gamma}\log\frac{1}{\xi})$. Then by \cref{lemArbParity}, we obtain a $(1, m+2, \mathcal{D}\sqrt{\frac{\varepsilon_H}{\lambda}}+\xi)$-block-encoding of $R_\mu(H)$ as defined in \cref{eqReflect}.

    Given the reflection operator $R_\mu(H)$ we have now produced, we can use this operator and the approximate ground state oracle $U_I$ to produce the ground state via $\mathcal O (\frac{1}{\delta})$ rounds of amplitude amplification~\citep{brassard2002quantum}. To produce a ground state with at least fidelity $1-\varepsilon_{\mathrm{prep}}$, we require the error of the block encoding of $R_\mu(H)$ to be within $\delta\varepsilon_{\mathrm{prep}}$. Hence, the equation $\mathcal{D}\sqrt{\frac{\varepsilon_H}{\lambda}}+\xi \le \delta\varepsilon_{\mathrm{prep}}$ has to be satisfied. Setting $\xi = \frac{\delta\varepsilon_{\mathrm{prep}}}{2}$, we note that $\mathcal{D} \in \mathcal O (\frac{\lambda}{\gamma}\log\frac{1}{\delta\varepsilon_{\mathrm{prep}}})$ and
    \begin{equation}
        \varepsilon_H\in \mathcal O \left(\frac{\delta^2\gamma^2\varepsilon_{\mathrm{prep}}^2}{\lambda \log^2 (1/\delta\varepsilon_{\mathrm{prep}})}\right)
    \end{equation}
    has to be satisfied, or
    \begin{equation}
        \varepsilon_H\in \widetilde{\mathcal O} \left(\frac{\delta^2\gamma^2\varepsilon_{\mathrm{prep}}^2}{\lambda}\right).
    \end{equation}

    Noting that $\mathcal{D}$ is of the same order as when we have a perfect block encoding. The number of queries to $U_H$, $U_I$, and the number of additional gates is the same as that of Theorem 6 of Ref.~\citep{lin2020nearoptimal}.
\end{proof}
\begin{corollary}
    Alternatively, the ground state can be prepared to at least $1-\varepsilon_{\mathrm{prep}}$ using the same number of queries to $U_H$ and $U_I$, and an extra
    \begin{equation*}
        \mathcal O \left( \frac{\lambda}{\delta \gamma} \log \frac{1}{\delta \varepsilon_{\mathrm{prep}}}  \left( m+\log \left(\frac{\lambda}{\delta \gamma\varepsilon_{\mathrm{prep}}}\right)\right) \right)
    \end{equation*}
    Toffoli gates.
    \label{corGSPrepTof}
\end{corollary}
\begin{proof}
    The extra gates of \cref{lemGSPrep} originate from three sources: the weighted LCU to shift the eigenvalues of $H$ by $\mu$, the ancilla-controlled phase shift operators, and the amplitude amplification step. To calculate the number of Toffoli gates required, we first determine how many continuous rotation gates are used. From the weighted LCU, we first construct a block encoding of $H-\mu I$, which requires two continuous rotation gates per query of $U_H$. From the ancilla-controlled phase shift operators, we know that a continuous rotation gate is applied per query of $U_H$, in addition to $\mathcal O(1)$ $m$-ancilla controlled NOT gates per query. Lastly, a set of reflectors implemented by $\mathcal O(1)$ $m$-ancilla controlled NOT gates are used per query of $U_I$ in the amplitude amplification step. In total, there are $\mathcal O \left( \frac{\lambda}{\delta \gamma} \log\frac{1}{\delta \varepsilon_{\mathrm{prep}}} \right)$ continuous rotation gates and $\mathcal O \left( \frac{\lambda}{\delta \gamma} \log\frac{1}{\delta \varepsilon_{\mathrm{prep}}} \right)$ $m$-ancilla controlled NOT gates.

    We first construct a unitary using \cref{lemGSPrep} to prepare a ground state up to $\frac{\varepsilon_\mathrm{prep}}{2}$-precision. Then, using $\mathcal O(\log \frac{1}{\varepsilon_{\text{gate}}})$ Toffolis to prepare a continuous rotation gate up to $\varepsilon_{\text{gate}}$-precision in the spectral norm difference of the resulting unitary, we note that by triangular inequality, replacing all continuous rotation gates with Toffolis will result in an error in spectral norm difference of $\mathcal O\left( \frac{\lambda\varepsilon_{\text{gate}}}{\delta \gamma} \log \frac{1}{\delta \varepsilon_{\mathrm{prep}}}\right)$. Capping the upper bound by $\frac{\varepsilon_{\mathrm{prep}}}{2}$, we note that $\frac{1}{\varepsilon_{\text{gate}}}\in \mathcal O\left( \frac{\lambda}{\delta \gamma\varepsilon_{\mathrm{prep}}} \log \frac{1}{\delta \varepsilon_{\mathrm{prep}}}\right)$.

    The number of Toffolis required is then found to be $\mathcal O \left( \frac{\lambda}{\delta \gamma} \log\frac{1}{\delta \varepsilon_{\mathrm{prep}}}\log \left(\frac{\lambda}{\delta \gamma\varepsilon_{\mathrm{prep}}} \log \frac{1}{\delta \varepsilon_{\mathrm{prep}}}\right) \right) $
    from the continuous rotations, and $\mathcal O \left( \frac{m\lambda}{\delta \gamma} \log\frac{1}{\delta \varepsilon_{\mathrm{prep}}} \right)$ from the $m$-controlled NOT gates.
\end{proof}

\begin{remark}
    \label{remGSPrepSuperposition}
    For a set of Hamiltonians $\{H_i\}$ prepared in superposition such that $H = \sum_i \ketbra{i}{i} \otimes H_i$, and given the approximate ground state oracle $U_I$ such that $U_I \ket{i} \ket{0} \to \ket{i} \ket{\phi_{\rm init, i}}$, one can use \cref{lemGSPrep} to prepare the ground state $\ket{\psi_{0, i}}$ with fidelity at least $1-\varepsilon_{\mathrm{prep}}$.
\end{remark}

Note that the controlled Hamiltonian $H = \sum_i \ketbra{i}{i} \otimes H_i$ can be reexpressed as a direct sum of the terms $H_i$, or $\bigoplus_{i} H_i.$ The singular value decomposition of a direct sum of matrices is equivalent to taking the direct sum of the decomposed components of the original matrices. Hence, the diagonalized controlled Hamiltonian $H$ is equal to the direct sum of the diagonalized individual Hamiltonians $H_i$. Hence, we can prepare the reflector in superposition, or $\sum_{i}\ketbra{i}{i} \otimes R(H_i)$.

To construct the other reflection, given that the approximate ground state oracle can produce $U_I \ket{i} \ket{0} \to \ket{i} \ket{\phi_{\rm init, i}}$, we apply $U_I\left(\mathbb{I}\otimes(2\ketbra{0}{0}-\mathbb{I})\right)U_I^\dagger$, which is equal to $\sum_{i}\ketbra{i}{i}\otimes (2\ketbra{\phi_{\rm init, i}}{\phi_{\rm init, i}}-\mathbb{I})$.

Thus, the Grover diffusion operator takes the form of
\begin{equation}
    -\sum_{i}\ketbra{i}{i} \otimes \left(R(H_i) \times(2\ketbra{\phi_{\rm init,i}}{\phi_{\rm init,i}}-\mathbb{I})\right) \, ,
\end{equation}
which allows us to prepare the ground state of different Hamiltonians in superposition via amplitude amplification.

\subsection{Quantum signal processing without angle finding}
Lastly, we review an alternate version of QSP proposed by \citet{alase2025quantum} which doesn't require the use of phase angle computation on classical computers. Here, we only discuss the specific application of quantum eigenvalue transformation. These methods enable the computation of functional representations in superposition on quantum computers, allowing for the execution of QSP with respect to different functions in superposition.

We provide the following definition of a function modified for the QET setting.
\begin{definition}[Diagonal block encoding of functions -- Definition 1, modified, \citep{alase2025quantum}]
    Let $f: [-1, 1] \to \mathbb{C}$, with $\lvert f(x)\rvert \le 1$. For $\mathcal{D} = 2^m$, where $m \in \mathbb{N}$, we say a $(1, a, \varepsilon)$-block-encoding encodes function $f$ if
    \begin{equation*}
        \lVert D_{f, 4\mathcal{D}} - (\mathbb I \otimes \bra{0^a}) U_{f, 4\mathcal{D}} (\mathbb I \otimes \ket{0^a})\rVert \le \varepsilon
    \end{equation*}
    and
    \begin{equation*}
        \bra{j'}D_{f, 4\mathcal{D}} \ket{j} = \delta_{j, j'} f\left(\cos \frac{2\pi j}{4\mathcal{D}}\right)
    \end{equation*}
    for $j, j' \in [4\mathcal{D}]$ \, .
\end{definition}

Given a $(n+a)$ qubit unitary $U_H$ that block-encodes a $n$-qubit Hermitian $H$, we define the following operators:
\begin{align}
    W                      & = (2\ketbra{0^a}{0^a}-\mathbb{I}_a) U_H \\
    V           & = \sum_{k=0}^{4\mathcal{D}-1} W^k \otimes \ketbra{k}{k} \\
    \ket{\varphi_\ell}        & = \sum_{k=0}^{4\mathcal{D}-1} \frac{1}{\sqrt{4\mathcal{D}}} e^{\frac{2i\pi k\ell}{4\mathcal{D}}}\ket{k} \nonumber \\
                           & = \mathrm{QFT}_{4\mathcal{D}}\ket{\ell} \\
    F_{4\mathcal{D}}       & = \sum_{\ell=0}^{4\mathcal{D}-1}f\left(\cos\frac{2\pi k}{4\mathcal{D}}\right) \ketbra{\varphi_\ell}{\varphi_\ell}\nonumber \\
                           & = \mathrm{QFT}_{4\mathcal{D}} D_{f, 4\mathcal{D}} \mathrm{QFT}_{4\mathcal{D}}^\dagger \\
    \ket{+_{2\mathcal{D}}} & = \frac{1}{\sqrt{2\mathcal{D}}}\sum_{k=\mathcal{D}}^{3\mathcal{D}-1}\ket{k} \\
    \ket{+_{4\mathcal{D}}} & = \frac{1}{\sqrt{4\mathcal{D}}}\sum_{k=0}^{4\mathcal{D}-1}\ket{k} \, .
\end{align}
With these operations, one can construct a degree $\mathcal{D}$ Laurent polynomial approximation $f_\mathcal{D}(H)$ such that
\begin{equation}
    f_\mathcal{D}(H) = \sqrt{2}\bra{+_{2\mathcal{D}}}V^\dagger \mathrm{QFT}_{4\mathcal{D}} D_{f, 4\mathcal{D}} \mathrm{QFT}_{4\mathcal{D}}^\dagger V \ket{+_{4\mathcal{D}}} \, .
\end{equation}
One would then obtain the following proposition,
\begin{proposition}[Quantum polynomial eigenvalue transform by angleless QSP -- Theorem 5, modified, ~\citep{alase2025quantum}]
    Let $f: [-1, 1] \to \mathbb{C}$, with $\lvert f(x)\rvert \le 1$. Given a $(n+a)$ qubit unitary $U_H$ that block-encodes a $n$-qubit Hermitian $H$, and a unitary $U_{f, 4\mathcal{D}}$ for $\mathcal{D} = 2^m$ that is a $(1, b, \delta)$-block-encoding, then we can construct a $(\sqrt{2}, m+a+b+2, (1+\sqrt{2})E_\mathcal{D}(f\circ \cos \circ \arg) + \sqrt2 \delta)$-block-encoding of $f(H)$ with $4\mathcal{D}-1$ queries of $C-U_H$ and $C-U_H^\dagger$ each, one query to $U_{f, 4\mathcal{D}}$, and $\mathcal O (\mathcal{D}a)$ Toffoli gates, where $E_\mathcal{D}(g)$ is the approximation error of $g(x)$ and its best degree $\mathcal{D}$ Laurent polynomial approximation $g_{\mathcal{D},*}(x)$.
    \label{propAnglelessQSP}
\end{proposition}
We note in the conventional QSP, the eigenvalues $\lambda$ of the Hamiltonian are transformed into $e^{i \arccos \lambda}$ from the qubitization approach~\citep{low2019hamiltonian} and the error stems from the Chebyshev polynomial approximation error $\varepsilon_\mathcal{D}(f)$. Hence, the corresponding polynomial error in this approach is that of $f\circ \cos \circ \arg$. Given that $E_\mathcal{D}(f\circ \cos \circ \arg)$ is the best Laurent approximation of $f\circ \cos \circ \arg$, we note that $E_\mathcal{D}(f\circ \cos \circ \arg) \le \varepsilon_\mathcal{D}(f)$.

We now apply these results to robust Hamiltonian simulation.
\begin{proposition}[Robust Hamiltonian simulation by angleless QSP]
    Let $t\in\mathbb{R}$, $\varepsilon\in(0,1)$ and let $U_H$ be a $(\alpha,a,\frac{\varepsilon}{2|t|})$-block-encoding of the Hamiltonian $H$. Further let $U_{f, 4\mathcal{D}}$ be a $(1, b, \frac{\varepsilon}{24\sqrt{2}})$-block-encoding of the function $f(x)= e^{-i\alpha t x}$. One can then implement a $(1, a+b+3, \varepsilon)$-block-encoding of $e^{-iHt}$ with
    \begin{equation*}
        48\alpha |t| + 72 \ln \frac{48(1+\sqrt 2)}{\varepsilon} -6
    \end{equation*}
    queries of $U_H$, $3$ queries to $U_{f, 4\mathcal{D}}$, and an additional $\mathcal O \left(a\left(\alpha|t|+\log\frac{1}{\varepsilon})\right)\right)$ Toffoli gates.
    \label{propAnglelessHamSim}
\end{proposition}
\begin{proof}
    Let $H' = \alpha(\bra{0}_a\otimes \mathbb{I}_n)U_H(\ket{0}_a\otimes \mathbb{I}_n)$. Then by Lemma 50 of Ref. \citep{chakraborty2019power}, we have $\lVert e^{-iHt} - e^{-iH't} \rVert \le \lVert H - H' \rVert t$ which we use to cap the error from Hamiltonian block encoding to $\frac{\varepsilon}{2|t|}$, and producing a total error of $\frac{\varepsilon}{2}$.

    We aim to create a $(\sqrt{2}, a+b+2, \frac{\varepsilon}{6})$-block-encoding of $e^{-iH't}$. Then by \cref{propAnglelessQSP}, we see that we require a $(1, b, \frac{\varepsilon}{24\sqrt{2}})$-block-encoding of the function $f(x)= e^{-i\alpha t x}$, or $U_{f, 4\mathcal{D}}$, where the error stems for first halving $\frac{\varepsilon}{12}$, then factoring the $\sqrt{2}$ from angleless QSP. From polynomial approximations of $e^{-i\alpha t x}$ as shown in Lemma 57 and 59 of Ref.~\citep{gilyen2019quantum}, we require
    \begin{equation}
        \mathcal{D} = r\left(\frac{e\alpha t}{2}, \frac{\varepsilon}{48(1+\sqrt 2)}\right) \le 2\alpha |t| + 3 \ln \frac{48(1+\sqrt 2)}{\varepsilon}
    \end{equation}
    by \cref{propAnglelessQSP} to produce a $\frac{\varepsilon}{48(1+\sqrt 2)}$ approximation for $\cos$ and $\sin$ separately, resulting in a total $\frac{\varepsilon}{24}$ error when factoring in the $1+\sqrt{2}$ from angleless QSP.

    Given a $(\sqrt{2}, a+b+2, \frac{\varepsilon}{12})$-block-encoding of $e^{-iH't}$, we then tensor product another Hadamard gate, which would result in a $(2, a+b+3, \frac{\varepsilon}{12})$-block-encoding of $e^{-iH't}$. Then by one round of robust oblivious amplitude amplification~\citep{berry2015simulating}, we obtain a $(1, a+b+3, \frac{\varepsilon}{2})$-block-encoding of $e^{-iH't}$, or a $(1, a+b+3, \varepsilon)$-block-encoding of $e^{-iHt}$. In total, we then need $12\mathcal{D}-3$ queries of $c-U_H$ and $12\mathcal{D}-3$ queries of $c-U_H^\dagger$ each from the results of \cref{propAnglelessQSP} and the amplitude amplification, or
    \begin{equation}
        48\alpha |t| + 72 \ln \frac{48(1+\sqrt 2)}{\varepsilon} -6
    \end{equation}
    queries in total.

    The Toffoli gates stem from the projection operators in the qubitized Hamiltonian, which are multi-controlled-NOT gates, which can be decomposed into $\mathcal O \left(a\left(\alpha|t|+\log\frac{1}{\varepsilon}\right)\right)$ Toffoli gates accordingly. The other source is the decomposition of QFT into Toffolis, which requires a trivial $\mathcal O (\log^2 \mathcal{D})$ Toffolis.
\end{proof}

\section{Construction and proof of the Liouvillian simulation algorithm}
\label{appLS}

We detail in this appendix the construction of the Liouvillian simulation algorithm. As mentioned in the main text, the simulation algorithm is realized by a block encoding of the Liouvillian $L$. With the implementation of the classical Liouvillian shown in Ref.~\citep{simon2024improved}, we show the implementation of the electronic Liouvillian $L_{\tel}$ as defined in \cref{defElecLiou}.

In the first subsection, we detail the implementation of $\Del$ by the Hellmann--Feynman theorem and ground state preparation subroutines. We then use this to build $L_{\tel}$ and $L$ in the subsequent subsection, and combine everything to construct the NVT Liouvillian evolution operator $e^{-iLt}$ in the third subsection. Lastly, we provide results on the NVE Liouvillian in the fourth subsection.

\subsection{Implementation of the force diagonal block encoding}

We first note that while the nuclear positions required for the implementation of the block encoding of the electronic Hamiltonian are loaded via QROM in Ref.~\citep{babbush2019quantum,su2021faulttolerant}, the nuclear positions are input into the electronic Hamiltonian via SWAP gates in Ref.~\citep{simon2024improved}. Thus, this construction provides the following operator:
\begin{equation}
    {H}_{\tel}^{\rm ctrl} = \sum_{\vec{\bar x}}\ketbra{\vec{\bar x}}{\vec{\bar x}} \otimes H_{\tel}(\vec x) \, .
    \label{eqCtrlHam}
\end{equation}
From \cref{remGSPrepSuperposition}, we know that \cref{lemGSPrep} can be used to prepare all ground states $\ket{\psi_0(\vec x)}$ of ${H}_{\tel}^{\rm ctrl}$ in superposition given the following assumption:
\begin{assumption}
    We are given a $(\lambda, a_{\tel}, \varepsilon_{\tel})$-block-encoding $U_H$ of the nuclear-position-controlled electronic Hamiltonian
    \begin{equation*}
        H_{\tel}^{\rm ctrl} = \sum_{\vec{\bar x}}\ketbra{\vec{\bar x}}{\vec{\bar x}} \otimes H_{\tel}(\vec x) \, .
    \end{equation*}
    We are given $\mu$ and $\gamma$ such that $E_0(\vec x) \leq \mu - \gamma/2 < \mu + \gamma/2 \leq E_1(\vec x)$ for all $\vec x$ and a state preparation oracle $U_I$ that prepares an approximate ground state such that $U_I \ket{\vec{\bar x}}\ket{0} \to \ket{\vec{\bar x}}\ket{\phi_{\rm init}(\vec x)}$ where for all $\vec x$, $\lvert\braket{ \psi_0(\vec x)| \phi_{\rm init}(\vec x)}\rvert \geq\delta$ is satisfied.
    \label{assumptionGS}
\end{assumption}

Recall that
\begin{equation}
    \Del = \sum_{\vec{\bar x}}  \frac{\partial E_{\tel}(\vec x)}{\partial x_{n,j}} \ketbra{\vec{\bar x}}{\vec{\bar x}}\, .
\end{equation}
For each $(n, j)$ pair, we aim to prepare the energy derivative operator $\Del$, which one can do by leveraging the Hellmann--Feynman theorem~\citep{hellmann1937einfuhrung, feynman1939forces}. From the first quantization formulation of the electronic Hamiltonian, one can analytically derive and then implement the operators $\frac{\partial H_{\tel}}{\partial x_{n,j}}$ for each $(n,j)$, the results of which have been shown in \cref{lemHamDeriv}. Given the smaller scaling factor compared to the electronic Hamiltonian, the total runtime can be lowered in the Hamiltonian simulation step when compared to using central finite difference methods to implement the energy derivative.

\begin{figure*}
    \includegraphics[width=\linewidth]{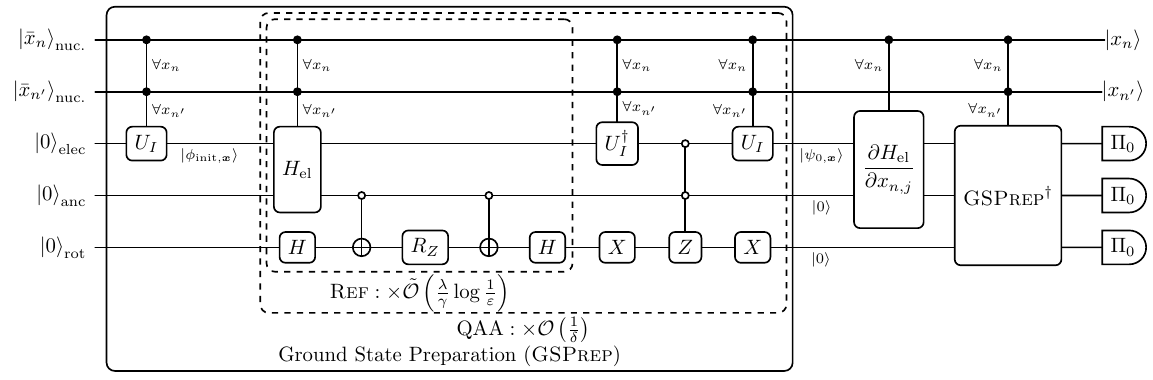}
    \caption[Block encoding implementation of the $\Del$ operator.]{\emph{Block encoding implementation of the $\Del$ operator.} This diagonal block encoding is implemented using the Hellmann--Feynman theorem~\citep{hellmann1937einfuhrung, feynman1939forces}, where we first prepare the ground state from the electronic Hamiltonian~\citep{lin2020nearoptimal}, apply the Hamiltonian derivative operator, and then uncompute the ground state. In this figure, we draw the approximate ground state preparation oracle $U_I$, electronic Hamiltonian $H_{\tel}$, and force operator $\frac{\partial H_{\tel}}{\partial x_{n,j}}$ with controls on the nuclear registers to symbolize that they are nuclear-register controlled. However, these operators are implemented as one single operator and need not be reimplemented for each different set of $\vec x$. In this figure and the figures after this, we symbolize the ancilla qubits of the block encoding by appending a
        \raisebox{-2pt}{\tikz{\draw[line width=0.7pt]
                (0,0) -- ++(0,1em) -- ++(1em,0)
                arc[start angle=90, end angle=0, radius=0.5em]
                -- ++(0, 0)
                arc[start angle=0, end angle=-90, radius=0.5em]
                -- ++(-0.5em, 0) -- cycle;
                \node at (0.75em,0.5em) {\scriptsize $\Pi_0$};}}
        symbol at the end of the qubit wire.}
    \label{figDel}
\end{figure*}

Similar to the electronic Hamiltonian, when the values of $\vec{x}$ are input via SWAP operations rather than calling from QROM, the operator that we obtain is a nuclear-position-controlled force operator
\begin{equation}
    \sum_{\vec{\bar x}}\ketbra{\vec{\bar x}}{\vec{\bar x}} \otimes \frac{\partial H_\tel(\vec x)}{\partial x_{n, j}} \, .
    \label{eqForce}
\end{equation}
We now use this operator to encode the derivative of the ground state energy for different nuclear positions, or the operator $\Del$. The circuit detailing this implementation is shown in \cref{figDel}.

\begin{lemma}[Block encoding of $\Del$]
    Suppose we have a $(\lambda, a_{\tel}, \varepsilon_{\tel})$-block-encoding $U_H$ of the nuclear-position-controlled electronic Hamiltonian satisfying \cref{assumptionGS} in regards to spectral gap $\gamma$, ground state upper bound $\mu$, and approximate ground state preparation unitary $U_I$ with up to an overlap $1-\delta$. Further, we are given a $(\tau_n, a_{\mathrm{force}}, \varepsilon_{\mathrm{force}})$-block-encoding of the nuclear-position-controlled electronic force operator as shown in \cref{eqForce} by \cref{lemHamDeriv}. Let $\tau = \frac{\tau_n}{Z_n}$. If
    \begin{equation*}
        \varepsilon_{\tel}\in\widetilde{\mathcal O} \left(\frac{\delta^2\gamma^2\varepsilon_{D_E}^4}{\lambda\tau^4Z_{\max}^4}\right),\,\varepsilon_{\mathrm{force}} = \frac{\varepsilon_{D_E}}{2}
    \end{equation*}
    then then exists a $(\tau_n, a_{D_E}, \varepsilon_{D_E})$-block-encoding of $\Del$ where
    \begin{equation*}
        a_{D_E} \in \max(a_{\tel}, a_{\mathrm{force}})+\mathcal{O}(\widetilde{N}\log B).
    \end{equation*}
    This block encoding can be prepared with $\mathcal{O}(\frac{\lambda}{\delta\gamma}\log\frac{\tau Z_{\max}}{\delta\varepsilon_{D_E}})$ queries to $U_H$, $\mathcal{O}(\frac{1}{\delta})$ queries to $U_I$, one query to $U_{\mathrm{force}}$, and an additional $\mathcal O \left( \frac{\lambda}{\delta \gamma} \log \frac{\tau Z_{\max}}{\delta \varepsilon_{D_E}}  \left( \max(a_{\tel}, a_{\rm force})+\log \frac{\lambda\tau Z_{\max}}{\delta \gamma\varepsilon_{D_E}}\right) \right)$ Toffoli gates.
    \label{lemDel}
\end{lemma}

\begin{proof}
    We first show the correctness of the block encoding construction. Let $W$ be the ground state preparation unitary such that $W(\vec x)\ket{0} = \ket{\psi_0(\vec x)}$. We denote the unitary as $W_{\vec{x}}$ for simplicity and readability. Assuming perfect block encoding and ground state preparation, we note that preparing the ground state in superposition gives us
    \begin{equation}
        \sum_{\vec{\bar x}}\ket{\vec{\bar x}}\ket{0}_{\mathrm{anc}} \to \sum_{\vec{\bar x}}\ket{\vec{\bar x}}\left(W_{\vec{x}}\ket{0}\right)\ket{0}_{\mathrm{anc}} \, .
    \end{equation}
    Then, applying the controlled Hamiltonian derivative operator, we obtain
    \begin{multline}
        \sum_{\vec{\bar x}}\ket{\vec{\bar x}}\bigg(\frac{1}{\tau_n} \left(\frac{\partial H_{\tel}}{\partial x_{n,j}} W_{\vec{x}}\ket{0}\right)\ket{0}_{\mathrm{anc}}\\
        +\sum_{j}\lambda_j'\ket{\mathrm{bad}_j}\ket{j}_{\mathrm{anc}}\bigg) \, .
    \end{multline}
    where $\lambda'_j$ is a normalization factor for some bad state $\ket{\mathrm{bad}_j}$. Then, uncomputing the ground state, we obtain
    \begin{multline}
        \sum_{\vec{\bar x}}\ket{\vec{\bar x}}\bigg(\frac{1}{\tau_n}\left(W_{\vec{x}}^\dagger\frac{\partial H_{\tel}}{\partial x_{n,j}} W_{\vec{x}}\ket{0}\right)\ket{0}_{\mathrm{anc}}\\
        +\sum_{j}\lambda_j'W_{\vec{x}}^\dagger\ket{\mathrm{bad}_j}\ket{j}_{\mathrm{anc}}\bigg) \, .
        \label{eqForceState}
    \end{multline}
    By the Hellmann--Feynman theorem, we note the following:
    \begin{equation}
        W_{\vec{x}}^\dagger\frac{\partial H_{\tel}}{\partial x_{n,j}} W_{\vec{x}}\ket{0} = \frac{\partial E_\tel}{\partial x_{n,j}}\ket{0} + \sum_k \lambda_k' \ket{k} \, ,
    \end{equation}
    where $E_\tel$ is the ground state energy of $H_{\tel}$. Thus, \cref{eqForceState} can be written as
    \begin{multline}
        \sum_{\vec{\bar x}}\ket{\vec{\bar x}}\bigg(\frac{1}{\tau_n}\frac{\partial E_\tel}{\partial x_{n,j}}\ket{0}_{\mathrm{anc}} + \sum_k \frac{\lambda_k'}{\tau_n} \ket{k}\ket{0}_{\mathrm{anc}}\\
        +\sum_{j}\lambda_j'W_{\vec{x}}^\dagger\ket{\mathrm{bad}_j}\ket{j}_{\mathrm{anc}}\bigg) \, .
    \end{multline}
    Therefore, when post-selecting on $\ket{0}_\mathrm{anc}$, we find that the encoded amplitude on $\ket{\vec{\bar x}}$ is then $\frac{1}{\tau_n}\bra{0}W_{\vec{x}}^\dagger\frac{\partial H_{\tel}}{x_{n,j}}W_{\vec{x}}\ket{0} = \frac{1}{\tau_n}\frac{\partial E_{\tel}}{\partial x_{n,j}}$. Hence, this process block-encodes the electronic forces on the nuclear registers by phase kickback.

    Next, we proceed to the factors of the block encoding. The scaling factor follows from the Hamiltonian derivative operator in \cref{lemHamDeriv}, which is $\tau_n \in \mathcal O \left( \frac{\widetilde NZ_n}{h_{\tel}^2}\right)$. Given that the electronic register is computed and then uncomputed, when encoding the ground state energy onto the nuclear register, the ancilla qubits are the max of the ancilla qubits between $H_{\tel}$ and $H_{\mathrm{force}}$, or $a_{\tel}$ and $a_{\mathrm{force}}$, plus the size of the electronic register, or $\mathcal{O}(\widetilde N \log B)$.

    Regarding the error, we first prepare the ground state of $H^{\rm ctrl}_{\tel}$ up to fidelity $1-\varepsilon_{\mathrm{prep}}$ with \cref{lemGSPrep}. This produces a state $\ket{\phi_0}$ such that $\lvert(\bra{\psi_0}\otimes\bra{0})\ket{\phi_0}\rvert \geq 1-\varepsilon_{\mathrm{prep}}$. For cleanliness, we denote $\ket{\widetilde \psi_0} = \ket{\psi_0}\otimes\ket{0}$. Then the $\ell_2$ distance between the two states can be found to be
    \begin{align}
        \left\lVert\ket{\phi_0}-\ket{\widetilde \psi_0}\right\rVert & = \sqrt{(\bra{\phi_0}-\bra{\widetilde \psi_0})(\ket{\phi_0}-\ket{\widetilde \psi_0})} \nonumber \\
        & = \sqrt{2-2\lvert\braket{\phi_0|\widetilde\psi_0}\rvert} \leq \sqrt{2\varepsilon_{\mathrm{prep}}} \, .
    \end{align}
    Given that the final block encoding is that of a \emph{diagonal} matrix on different nuclear positions, the error term of the block encoding (spectral norm difference) is less than that of the maximum difference on a single entry, that is, the maximum difference on the eigenvalues themselves. We derive the difference as follows via a series of triangular inequalities and drop the notation on nuclear positions, as the derivation is for a single position:
    \begin{align}
         & \left\lvert\tau_n\bra{0}W^\dagger U_{\mathrm{force}} W\ket{0} - \frac{\partial E_\tel}{\partial x_{n,j}}\right\rvert \nonumber \\
         & \begin{aligned}[b]
               \le\; & \tau_n\left\lvert\bra{0} W^\dagger U_{\mathrm{force}}W\ket{0} - \bra{0} W^\dagger U_{\mathrm{force}}\ket{\widetilde \psi_0}\right\rvert \\
               & + \tau_n\left\lvert\bra{0}W^\dagger U_{\mathrm{force}} \ket{\widetilde \psi_0} - \bra{\widetilde \psi_0}  U_{\mathrm{force}} \ket{\widetilde \psi_0}\right\rvert \\
               & + \left\lvert\tau_n \bra{\widetilde \psi_0} U_{\mathrm{force}} \ket{\widetilde \psi_0} - \bra{\widetilde \psi_0}  \frac{\partial H_{\tel}}{\partial x_{n,j}} \ket{\widetilde \psi_0}\right\rvert \, .
           \end{aligned}
    \end{align}
    We now derive the upper bounds of the three absolute differences. The first term can be bounded by the Cauchy–Schwarz inequality such that
    \begin{align}
         & \tau_n\left\lvert\bra{0}W^\dagger U_{\mathrm{force}}W\ket{0} - \bra{0}W^\dagger U_{\mathrm{force}}\ket{\widetilde \psi_0}\right\rvert\nonumber \\
         & = \tau_n\left\lvert\bra{0}W^\dagger U_{\mathrm{force}}\ket{\phi_0} - \bra{0} W^\dagger U_{\mathrm{force}}\ket{\widetilde \psi_0}\right\rvert\nonumber \\
         & = \tau_n \left\lvert\bra{0}W^\dagger U_{\mathrm{force}}\left(\ket{\phi_0} - \ket{\widetilde \psi_0}\right)\right\rvert \nonumber \\
         & \leq \tau_n\left\lVert\bra{0}W^\dagger U_{\mathrm{force}}\right\rVert\left\lVert\ket{\phi_0} - \ket{\widetilde \psi_0}\right\rVert \nonumber \\
         & = \tau_n\left\lVert\ket{\phi_0} - \ket{\widetilde \psi_0}\right\rVert \, ,
    \end{align}
    where the last equality holds because $W$ is unitary.

    The second term can be similarly bounded such that
    \begin{align}
         & \tau_n\left\lvert\bra{0}W^\dagger U_{\mathrm{force}} \ket{\widetilde \psi_0} - \bra{\widetilde \psi_0}  U_{\mathrm{force}} \ket{\widetilde \psi_0}\right\rvert\nonumber \\
         & = \tau_n\left\lVert\bra{\phi_0} - \bra{\widetilde \psi_0}\right\rVert \left\lVert U_{\mathrm{force}}\ket{\widetilde\phi_0}\right\rVert\nonumber \\
         & = \tau_n\left\lVert\ket{\phi_0} - \ket{\widetilde \psi_0}\right\rVert \, ,
    \end{align}

    The third term can be bounded as follows:
    \begin{align}
         & \left\lvert\bra{\widetilde \psi_0}\tau_n W^\dagger U_{\mathrm{force}} \ket{\widetilde \psi_0} - \bra{\widetilde \psi_0} \frac{\partial H_\tel}{\partial x_{n,j}} \ket{\widetilde \psi_0}\right\rvert \nonumber \\
         & = \left\lvert\bra{\psi_0}\left(\tau_n(\mathbb{I}\otimes\bra{0}) U_{\mathrm{force}} (\mathbb{I}\otimes\ket{0}) - \frac{\partial H_\tel}{\partial x_{n,j}}\right) \ket{\psi_0}\right\rvert\nonumber \\
         & \leq \left\lVert\bra{\psi_0}\right\rVert\left\lVert\tau_n(\mathbb{I}\otimes\bra{0}) U_{\mathrm{force}} (\mathbb{I}\otimes\ket{0}) -  \frac{\partial H_\tel}{\partial x_{n,j}}\right\rVert\left\lVert\ket{\psi_0}\right\rVert \nonumber \\
         & \leq \left\lVert\left(\tau_n(\mathbb{I}\otimes\bra{0}) U_{\mathrm{force}} (\mathbb{I}\otimes\ket{0}) - \frac{\partial H_\tel}{\partial x_{n,j}} \right) \right\rVert\le \varepsilon_{\mathrm{force}}.
    \end{align}
    where the second-to-last inequality is by Cauchy–Schwarz inequality and the last inequality is by the definition of the spectral norm.

    Putting the results together, we get:
    \begin{align}
         & \left\lvert\tau_n\bra{0}W^\dagger U_{\mathrm{force}}W\ket{0} - \frac{\partial E_\tel}{\partial x_{n,j}}\right\rvert \nonumber \\
         & \le 2\tau_n \lVert\ket{\phi_0} - \ket{\widetilde \psi_0}\rVert + \varepsilon_{\mathrm{force}}\nonumber \\
         & \le 2\tau_n \sqrt{2\varepsilon_{\mathrm{prep}}} + \varepsilon_{\mathrm{force}} \, .
    \end{align}
    Set $\varepsilon_{\mathrm{force}} = \frac{\varepsilon_{D_E}}{2}$ and $2\tau_n \sqrt{2\varepsilon_{\mathrm{prep}}} = \frac{\varepsilon_{D_E}}{2}$.
    Replacing the bounds of $\varepsilon_{\mathrm{prep}}$ of that obtained in \cref{lemGSPrep}, where we see that
    \begin{equation}
        \varepsilon_{\tel}\in \widetilde{\mathcal O} \left(\frac{\delta^2\gamma^2\varepsilon_{\mathrm{prep}}^2}{\lambda}\right)\subseteq  \widetilde{\mathcal O} \left(\frac{\delta^2\gamma^2\varepsilon_{D_E}^4}{\lambda\tau_n^4}\right) \, .
    \end{equation}
    For consistency across different $n$, we set $\varepsilon_{\tel}$ to a tighter bound. Let $\tau = \frac{\tau_n}{Z_n}$. We then let
    \begin{align}
        \varepsilon_{\mathrm{prep}} & \in \mathcal O \left(\frac{\varepsilon_{D_E}^2}{\tau^2 Z_{\max}^2}\right) \\
        \varepsilon_{\tel}          & \in \widetilde{\mathcal O} \left(\frac{\delta^2\gamma^2\varepsilon_{\mathrm{prep}}^2}{\lambda}\right)\subseteq  \widetilde{\mathcal O} \left(\frac{\delta^2\gamma^2\varepsilon_{D_E}^4}{\lambda\tau^4Z_{\max}^4}\right) \, .
    \end{align}
    which still satisfies the former bound.
    The runtime is largely dominated by ground state preparation plus one query to the force operator. Replacing $\varepsilon_{\mathrm{prep}}$ by its bounds of $\varepsilon_{D_E}$, we note that we require $\mathcal{O}(\frac{\lambda}{\delta\gamma}\log\frac{\tau Z_{\max}}{\delta\varepsilon_{D_E}})$ queries to $U_H$. The number of Toffoli gates needed can be found by \cref{corGSPrepTof} to be
    \begin{equation}
        \mathcal O \left( \frac{\lambda}{\delta \gamma} \log \frac{\tau Z_{\max}}{\delta \varepsilon_{D_E}}  \left( \max(a_{\tel} \, , a_{\mathrm{force}})+\log \left(\frac{\lambda\tau Z_{\max}}{\delta \gamma\varepsilon_{D_E}}\right)\right) \right) \, .
    \end{equation}
\end{proof}

\begin{figure*}
    \includegraphics[width=\linewidth]{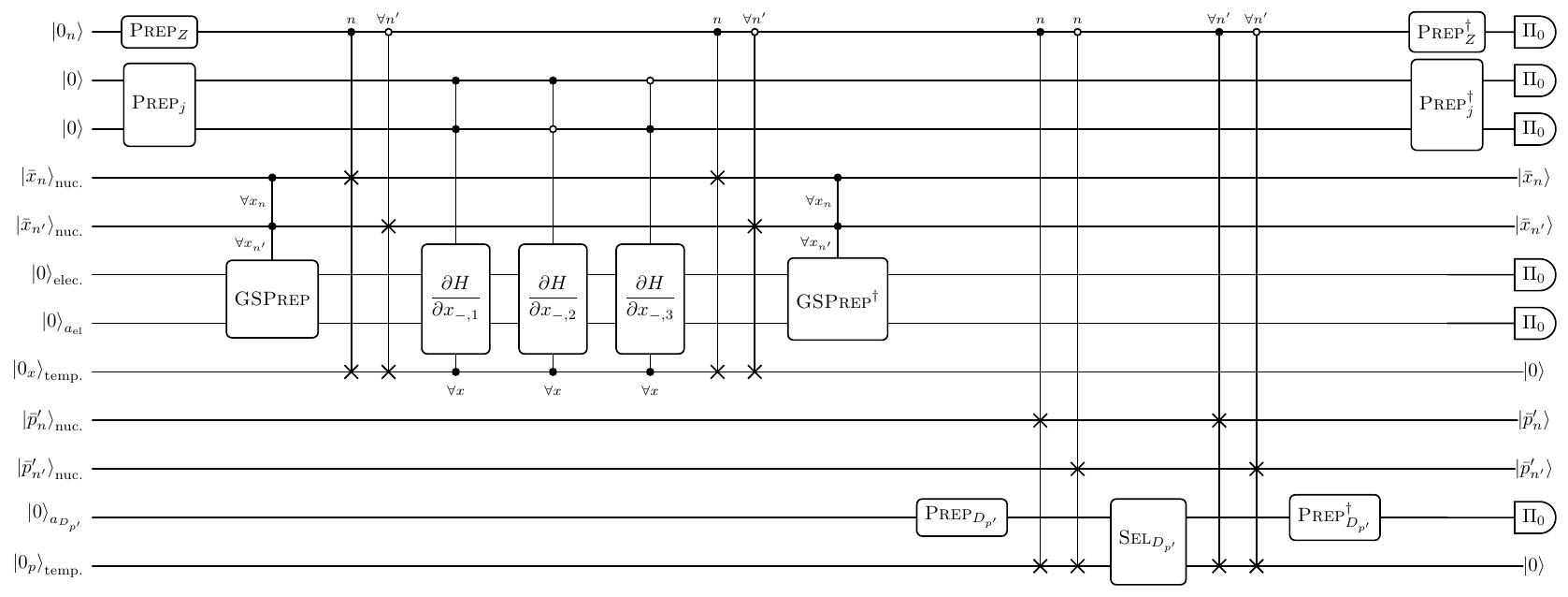}
    \caption[Block encoding implementation of the electronic Liouvillian.]{\emph{Block encoding implementation of the electronic Liouvillian.} We combine the implementation of $\Del$ from \cref{figDel} and $D_{p'}$ as shown in Ref.~\citep{simon2024improved} for all $n$ and $j$ via LCU. Note that the nuclear-register-controlled ground state preparation operator is implemented as a single operator and need not be repeated for different sets of $\vec x$.}
    \label{figLel}
\end{figure*}

\subsection{Implementation of the Liouvillian}

Now that we have the implementation of $\Del$ at hand, we can implement the electronic Liouvillian, which we state again as follows:
\begin{equation}
    L_{\tel}^{\mathrm{(NVT)}} := i\sum_{n=1}^{N} \sum_{j=1}^{3} \Del \otimes D_{p'_{n,j}} \, ,
\end{equation}
where the central finite difference operator reads,
\begin{equation}
    D_{p'_{n,j}} = \frac{1}{h_{p'}} \sum_{k=-d_{p'}}^{d_{p'}}c_{d_{p'}, k} \ketbra{\bar{p}'_{n,j}-k}{\bar{p}'_{n,j}} \, .
\end{equation}

From the results of Equations (A79) to (A82) in Ref.~\citep{simon2024improved}, we know that a $(\alpha_{D_{p'}}, a_{D_{p'}}, \varepsilon_{D_{p'}})$-block-encoding of $D_{p'_{n,j}}$, where
\begin{align}
    \alpha_{D_{p'}} & \in \mathcal O \left(\frac{\log d_{p'}}{h_{p'}}\right), \\
    a_{D_{p'}}      & \in \mathcal O (\log d_{p'}) \, ,
\end{align}
can be prepared with $\mathcal O \left(d_{p'}\log \frac{g_{p'}\log d_{p'}}{h_{p'}\varepsilon_{D_{p'}}}\right)$ Toffoli gates using a weighted combination of unitary adders~\citep{draper2000addition} with the adjustments made in \cref{lemStatePrep}. Then, to implement $L_{\tel}$, we simply take a weighted sum over $\Del \otimes D_{p'_{n,j}}$ via LCU, or \cref{lemLCU}. We again implement this operator by \textsc{Prep}$^\dagger$--\textsc{Sel}--\textsc{Prep} formalism. Given that the $\textsc{Sel}$ operators would have different weightings dependent on $n$, we separate $n$ and $j$ into two registers. Note that given for each $\Del$, the scaling factor $\tau_n$ differs by a factor of $Z_n$, we need to prepare the state
\begin{equation}
    \sum_{n=1}^N \sqrt{\frac{Z_n}{\sum_{n'=1}^N Z_{n'}}}\ket n \otimes \frac{1}{\sqrt{3}}(\ket{00} + \ket{01} + \ket{10})
\end{equation}
via the $\textsc{Prep}_Z$ operator. On the other hand, the \textsc{Sel} operator is formulated as follows:
\begin{equation}
    \sum_{n=1}^N \sum_{j=1}^3 \ketbra{n}{n}\otimes\ketbra{j}{j}\otimes \Del \otimes D_{p'_{n,j}} \, .
\end{equation}

Note that, since the $\Del$ operator utilizes ground state preparation on the electronic registers, we can leverage the ground state preparation algorithm across all LCU components by querying it twice, thereby reducing the runtime. Further, the $D_{p'}$ only needs to be queried once if we use an ancilla register that swaps in the momentum states and applies the operator in superposition. A complete circuit diagram can be shown in \cref{figLel}. We formalize the results as follows:

\begin{proposition}[Block encoding of the electronic Liouvillian]
    Suppose we have a nuclear-position-controlled electronic Hamiltonian that can be block-encoded in the first quantization scheme with scaling factor $\lambda$ satisfying \cref{assumptionGS} in regards to spectral gap $\gamma$, ground state upper bound $\mu$. We are given an approximate ground state preparation oracle $U_I$ that can prepare ground states up to an overlap of $1-\delta$. Then there exists a $(\alpha_{L_{\tel}}, a_{L_{\tel}}, \varepsilon_{L_{\tel}})$-block-encoding of $L_{\tel}$ where
    \begin{align*}
        \alpha_{L_{\tel}} & \in \mathcal O \left( \frac{N\widetilde NZ_{\max}\log d_{p'}}{h_{\tel}^2h_{p'}}\right) \\
        a_{L_\tel}        & \in \mathcal{O}\left(\widetilde{N}\log B + \log\frac{d_{p'}}{\delta\gamma}+\log \frac{\alpha_{L_\tel}}{\varepsilon_{L_{\tel}}}\right) \, .
    \end{align*}
    This block encoding can be prepared with $\mathcal{O}(\frac{1}{\delta})$ queries to $U_I$, and
    \begin{align*}
        \widetilde{\mathcal O} & \Bigg( \frac{\lambda}{\delta \gamma}\log \frac{\alpha_{L_{\tel}}}{\varepsilon_{L_{\tel}}}\left( N + \widetilde N \log B + \log B\log\frac{B g_x \alpha_{L_{\tel}}}{\varepsilon_{L_{\tel}}}\right)\nonumber \\
        & + N \log g_{p'} +d_{p'}\log \frac{\alpha_{L_{\tel}}}{\varepsilon_{L_{\tel}}}\Bigg)
    \end{align*}
    Toffoli gates.
    \label{propElecLiou}
\end{proposition}
\begin{proof}
    We first calculate the scaling factors for the block encoding. We are given a $(\tau_n, a_{D_E}, \varepsilon_{D_E})$-block-encoding of $\Del$ and a $(\alpha_{D_{p'}}, a_{D_{p'}}, \varepsilon_{D_{p'}})$-block-encoding of $D_{p'_{n,j}}$. By \cref{lemProd}, we can construct a $(\tau_n\alpha_{D_{p'}}, a_{D_E}+a_{D_{p'}}, \tau_n\varepsilon_{D_{p'}} + \alpha_{D_{p'}}\varepsilon_{D_E})$-block-encoding of $\Del\otimes D_{p'_{n,j}}$.

    Noting that by \cref{lemHamDeriv}, $\forall n, n', \frac{\tau_n}{Z_n} = \frac{\tau_{n'}}{Z_{n'}}$, we let $\tau = \frac{\tau_n}{Z_n}$. We can then notice that for each $(n,j)$ pair, we have a
    \begin{equation*}
        \left(\tau\alpha_{D_{p'}}, a_{D_E}+a_{D_{p'}}, \tau\varepsilon_{D_{p'}} + \frac{\alpha_{D_{p'}}\varepsilon_{D_E}}{Z_n}\right)\text{-block-encoding}
    \end{equation*}
    of $\frac{1}{Z_n}\Del\otimes D_{p'_{n,j}}$. Given that the minimum value of $Z_n$ is 1, we can also set the error term to be $\tau\varepsilon_{D_{p'}} + \alpha_{D_{p'}}\varepsilon_{D_E}$ for consistency over all $(n,j)$ pairs. We can then construct the weighted sum of
    \begin{equation}
        i\sum_{n=1}^N\sum_{j=1}^3 Z_n \frac{\Del\otimes D_{p'_{n,j}}}{Z_n} = L_{\tel}
    \end{equation}
    to construct the electronic Liouvillian. The additional $i$ factor can be added by $XSXS$ gates, or in the controlled version, $(X\otimes H)CX(X\otimes H)CX(\mathbb{I}\otimes S)CX(\mathbb{I}\otimes S)$ gates. Given a $(\alpha_Z, a_Z, \varepsilon_Z)$-state-preparation-pair that prepares the state $\sum_{n=1}^N \sqrt{\frac{Z_n}{\sum_{n'=1}^N Z_{n'}}}\ket n \otimes \frac{1}{\sqrt{3}}(\ket{00} + \ket{01} + \ket{10})$, where $\alpha_Z = 3\sum_{n=1}^NZ_n\in \mathcal O (NZ_{\max})$ and $a_Z \in \mathcal O (\log N)$, by \cref{lemLCU}, we construct a $(\alpha_{L_\tel}, a_{L_\tel}, \varepsilon_{L_\tel})$-block-encoding of the electronic Liouvillian, where
    \begin{align}
        \alpha_{L_\tel}      & \in \mathcal O (NZ_{\max}\tau\alpha_{D_{p'}}) \\
        a_{L_\tel}           & \in a_{D_E}+a_{D_{p'}} + \mathcal O (\log N) \\
        \varepsilon_{L_\tel} & \in \mathcal O (NZ_{\max}\tau\varepsilon_{D_{p'}} + NZ_{\max}\alpha_{D_{p'}}\varepsilon_{D_E} +  \tau\alpha_{D_{p'}}\varepsilon_Z) \, .
    \end{align}
    Plugging in the values of $\tau$ and $\alpha_{D_{p'}}$, we note that we get for the scaling factor
    \begin{equation}
        \alpha_{L_\tel} \in \mathcal O \left( \frac{N\widetilde NZ_{\max}\log d_{p'}}{h_{\tel}^2h_{p'}}\right)\, .
    \end{equation}
    Next, as the number of ancillae is dependent on error, we check the error terms first. Starting with the largest bound on $\varepsilon_{L_\tel}$ and working our way down, we obtain the following:
    \begin{align}
        \varepsilon_{D_{p'}}         & \in \mathcal O \left(\frac{\varepsilon_{L_\tel}}{NZ_{\max}\tau}\right)\subseteq \mathcal O \left(\frac{h_{\tel}^2\varepsilon_{L_\tel}}{N\widetilde NZ_{\max}}\right) \\
        \varepsilon_{D_E}            & \in \mathcal O \left(\frac{\varepsilon_{L_\tel}}{NZ_{\max}\alpha_{D_{p'}}}\right)\subseteq \mathcal O \left(\frac{h_{p'}\varepsilon_{L_\tel}}{NZ_{\max}\log d_{p'}}\right) \\
        \varepsilon_{Z}              & \in \mathcal O \left(\frac{\varepsilon_{L_\tel}}{\tau\alpha_{D_{p'}}}\right)\subseteq \mathcal O \left(\frac{h_{\tel}^2h_{p'}\varepsilon_{L_\tel}}{\widetilde N\log d_{p'}}\right) \\
        \varepsilon_{\tel}           & \in \widetilde{\mathcal O} \left(\frac{\delta^2\gamma^2\varepsilon_{D_E}^4}{\lambda\tau^4Z_{\max}^4}\right)\subseteq  \widetilde{\mathcal O} \left(\frac{\delta^2\gamma^2h_{p'}^4h_{\tel}^8\varepsilon_{L_\tel}^4}{\lambda N^4\widetilde N^4Z_{\max}^8\log^4 d_{p'}}\right) \\
        \varepsilon_{\mathrm{force}} & \in \mathcal O (\varepsilon_{D_E}) \subseteq \mathcal O \left(\frac{h_{p'}\varepsilon_{L_\tel}}{NZ_{\max}\log d_{p'}}\right) \, .
    \end{align}
    Moving on to the ancillae, we note that
    \begin{align}
        a_{L_\tel} & \in a_{D_E}+a_{D_{p'}} + \mathcal O \left(\log N\right) \nonumber \\
        & \subseteq \mathcal{O}\left( \widetilde{N}\log B + \log\frac{N d_{p'}}{\varepsilon_\tel}\right)\nonumber \\
        & \subseteq \mathcal{O}\left( \widetilde{N}\log B + \log\frac{N Z_{\max} d_{p'}}{\delta\gamma h_{p'}h_{\tel}\varepsilon_{L_{\tel}}}\right) \, .
    \end{align}
    Note that we can rewrite this expression as
    \begin{equation}
        a_{L_\tel} \in \mathcal{O}\left(\widetilde{N}\log B + \log\frac{d_{p'}}{\delta\gamma}+\log \frac{\alpha_{L_\tel}}{\varepsilon_{L_{\tel}}}\right) \, .
    \end{equation}
    We also make use of $\mathcal O (N\log g_x)$ ancillae for swapping in nuclear positions and another $\mathcal O (\log g_{p'})$ pure ancillae for momentum positions.

    Lastly, to calculate the number of Toffoli gates, we first go through the \textsc{Prep} operator. Per \cref{lemStatePrep}, using the  \textsc{Prep} operator to prepare a state preparation pair up to $\varepsilon_Z$, we note that this would cost $\mathcal O \left(N\log\frac{NZ_{\max}}{\varepsilon_Z}\right)$ Toffoli gates. Next, moving on to the \textsc{Sel} operator, we note that for all $(n,j)$ pairs, we prepare the same ground state; we only need to query the ground state preparation algorithm twice---once for computing and once for uncomputing the ground state. Further, given that we control-SWAP the nuclear positions into an ancilla register and compute the force operator, we query the force operator $\frac{\partial H_{\tel}}{x_{n,j}}$ thrice---once for each value of $j$. If we further decompose the force operator, we note that only the single qubit rotation controlled by $\ket{\nu}$ is different---the other operations may be shared. Nevertheless, the asymptotics are the same. Then, for the $D_{p'_{n,j}}$ operator, we again only need to query once, as we can also control-SWAP the momentum positions into an ancilla register. Thus, we obtain the total number of gates as follows:
    \begin{align}
        \#_{\text{Toffolis}} & =
        \begin{multlined}[t]
            \mathcal O \left( \frac{\lambda}{\delta \gamma} \log \frac{\tau Z_{\max}}{\delta\varepsilon_{D_E}}\right)\\
            \times \mathcal O \left( N + \log B \left(\widetilde{N} + \log\frac{Bg_x}{\varepsilon_{\tel}} \right)\right)
        \end{multlined}\nonumber \\
        & +\mathcal O \left( \frac{\lambda}{\delta \gamma} \log \frac{\tau Z_{\max}}{\delta \varepsilon_{D_E}}  \left( a_{\tel}+\log \left(\frac{\lambda\tau Z_{\max}}{\delta \gamma\varepsilon_{D_E}}\right)\right) \right)\nonumber \\
        & +\mathcal O\left(N\log g_x+ \log B\left(\widetilde N +\log\frac{Bg_x}{\varepsilon_{\mathrm{force}}}\right)\right)\nonumber \\
        & +\mathcal O \left(N\log g_{p'}+ d_{p'}\log \frac{g_{p'}\log d_{p'}}{h_{p'}\varepsilon_{D_{p'}}}\right)\nonumber \\
        & +\mathcal O \left(N\log\frac{NZ_{\max}}{\varepsilon_Z}\right) \, ,
    \end{align}
    where the first two terms are from ground state preparation, the third term is from the force operator, the fourth term is from the momentum derivative operator, and the last term is from the state preparation pair.
    Cleaning up the result, we note that a total of
    \begin{align}
        \widetilde{\mathcal O} & \Bigg(\frac{\lambda}{\delta \gamma}\log \frac{\alpha_{L_{\tel}}}{\varepsilon_{L_{\tel}}}\left( N + \widetilde N \log B + \log B\log\frac{B g_x \alpha_{L_{\tel}}}{\varepsilon_{L_{\tel}}}\right)\nonumber \\
        & + N \log g_{p'} +d_{p'}\log \frac{g_{p'}\alpha_{L_{\tel}}}{\varepsilon_{L_{\tel}}}\Bigg)
    \end{align}
    Toffoli gates are used, where, further cleaning up constant values, we obtain
    \begin{align}
         & \widetilde{\mathcal O} \left( \frac{\widetilde N N_{\mathrm{tot}}}{\delta \gamma}\log \frac{1}{\varepsilon_{L_{\tel}}}\left( N_{\mathrm{tot}} + \log\frac{1}{\varepsilon_{L_{\tel}}}\right)\right)\nonumber \\
         & \subseteq \widetilde{\mathcal O}\left( \frac{\widetilde N N_{\mathrm{tot}}^2}{\delta \gamma}\log^2 \frac{1}{\varepsilon_{L_{\tel}}}\right) \, .
    \end{align}
\end{proof}
To block-encode the final Liouvillian, we require the block encoding of the classical Liouvillian under the NVT ensemble as shown in Ref.~\citep{simon2024improved}.
\begin{lemma}[Block encoding of the discretized classical NVT Liouvillian -- Lemma 2, revised, \citep{simon2024improved}]
    There exists a $(\alpha_{\mathrm{NVT}}, a_{\mathrm{NVT}},\varepsilon_\mathrm{NVT})$-block-encoding of the discretized classical Liouvillian $L_{\mathrm{NVT}}$ with scaling factor
    \begin{multline*}
        \alpha_{\mathrm{NVT}} \in \mathcal O\Bigg(N \frac{p'_{\max}}{m_{\min} s_{\min}^2} \frac{\log d_x}{h_x} + N^2 \frac{Z_{\max}^2 x_{\max}}{\Delta^3} \frac{\log d_{p'}}{h_{p'}} \\
        + \frac{p_{s,\max}}{Q} \frac{\log d_s}{h_s} + \left(N \frac{{p'}_{\max}^2}{m_{\min} s_{\min}^3} + \frac{N_f k_B T}{s_{\min}} \right) \frac{\log d_{p_s}}{h_{p_s}} \Bigg)
    \end{multline*}
    and the number of ancilla qubits
    \begin{equation*}
        a_{\mathrm{NVT}} \in \mathcal O\left(\log d+ \log \frac{\alpha_{NVT}}{\varepsilon_{\mathrm{NVT}}} \right),
    \end{equation*}
    where $ d := \max( d_x, d_{p'}, d_s, d_{p_s} )$. This block encoding can be implemented using
    \begin{equation*}
        \widetilde{\mathcal O} \left((N+d) \log \frac{g \alpha_{NVT}}{\varepsilon_{\mathrm{NVT}}} + \log^{\log 3} \frac{\alpha_{NVT}}{\varepsilon_{\mathrm{NVT}}} \right)
    \end{equation*}
    Toffoli gates, where $g := \max( g_x, g_{p'}, g_s, g_{p_s} )$.
    \label{lemClassLiou}
\end{lemma}

\begin{turnpage}
    \begin{figure*}
        \centering
        \includegraphics[width=\textheight]{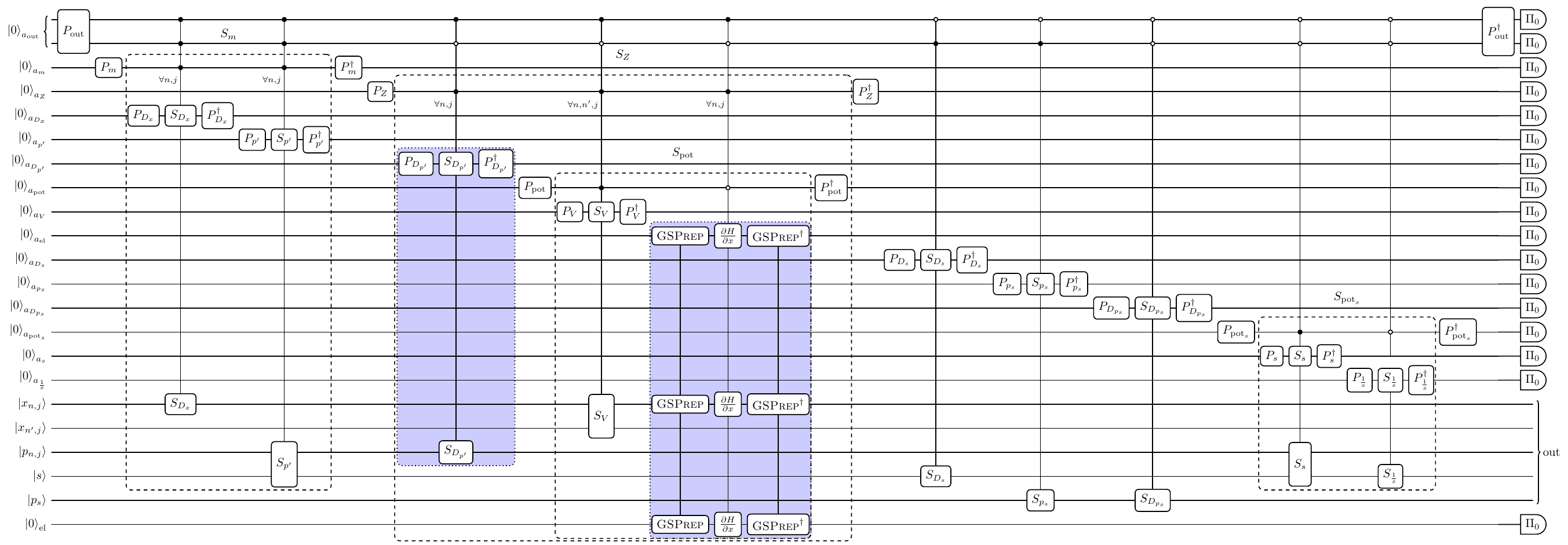}
        \caption[Block encoding implementation of the NVT Liouvillian $L$.]{\emph{Block encoding implementation of the NVT Liouvillian $L$.} As in Ref.~\citep{simon2024improved}, the Liouvillian is broken into separate operators that can be individually implemented by \textsc{Prep}--\textsc{Sel}--\textsc{Prep}$^\dagger$ operators. \textsc{Prep} are denoted as $P$ while \textsc{Sel} are denoted as $S$. These block-encoded operators include the discrete derivative operators for position $U_{D_x}$, momentum $U_{D_{p'}}$ and bath variables $U_{D_s}$, $U_{D_{p_s}}$, the momentum operator $U_{p'}$, the nuclear force operator $U_V$, the momentum operator for the heat bath $U_{p_s}$, the kinetic energy force operator of the heat bath $U_s$, and the potential energy force operator of the heat bath $U_{\frac{1}{s}}$. In addition to the nine operators previously block-encoded in Ref.~\citep{simon2024improved}, we provide the block encoding of $U_{D^{\tel}}$ as shown in \cref{lemDel} and \cref{figDel}, which corresponds to the electronic force operator. When applying the tensor product with $D_{p'}$, this forms the electronic Liouvillian shown in \cref{propElecLiou} and \cref{figLel}, which is highlighted in blue. In addition, we also require the use of $P_m$ to encode the nuclear masses and $P_Z$ to encode the nuclear charges to create the correct weighted sum of block encodings, as well as $P_{f}$ and $P_{f_s}$ to combine the force terms correctly and $P_{\mathrm{out}}$ to combine the Liouvillian terms correctly. Note that for simplicity, the figure only displays the implementation used for a single $(n,j)$ pair. Implementations for all $(n,j)$ would make use of SWAP gates to reduce the number of queries to individual block encodings, as illustrated in \cref{figLel} for the electronic Liouvillian.}
        \label{figLiou}
    \end{figure*}
\end{turnpage}
Lastly, we arrive at the block encoding of the full Liouvillian. Note that in the practice of implementing the block encoding of the hardware, there is no need to separate the classical and electronic Liouvillians. One can simply LCU the operator $\Del$ to the nuclear force operator $V$ and then tensor-product $D_{p'}$. This implementation is shown in \cref{figLiou}. However, in the proofs, the choice to separate classical and electronic Liouvillian is simply due to the ease of theoretical proofs, without having to perform repetitive calculations regarding the operators in the classical Liouvillian. Nevertheless, we note that due to the scaling factor being the $\ell_1$-norm of components of an LCU, the scaling factor of the Liouvillian should not be affected. The runtime and ancilla count \textit{may} be affected slightly due to the changes in how error propagates, but as error has only a logarithmic dependency in both the runtime and ancilla count, the dependencies of the multiplicative factors acting on the error would be hidden within polylogarithmic factors not shown in the asymptotics. We now present the following theorem, where the implementation is based on separate classical and electronic Liouvillians.

\begin{proposition}[Block encoding of the NVT canonical Liouvillian]
    Suppose we have a nuclear-position-controlled electronic Hamiltonian that can be block-encoded in the first quantization scheme with scaling factor $\lambda$ satisfying \cref{assumptionGS} in regards to spectral gap $\gamma$, ground state upper bound $\mu$. We are given an approximate ground state preparation oracle $U_I$ that can prepare ground states up to an overlap of $1-\delta$. Then there exists a $(\alpha_L, a_L, \varepsilon_L)$-block-encoding of the NVT canonical Liouvillian $L$, where
    \begin{align*}
        \alpha_L \in \mathcal O & \Bigg(N\widetilde N\frac{Z_{\max}}{h_{\tel}^2}\frac{\log d_{p'}}{h_{p'}} + N^2 \frac{Z_{\max}^2 x_{\max}}{\Delta^3} \frac{\log d_{p'}}{h_{p'}} \\
        & + N \frac{p'_{\max}}{m_{\min} s_{\min}^2} \frac{\log d_x}{h_x}+ \frac{p_{s,\max}}{Q} \frac{\log d_s}{h_s} \\
        & + \left( N \frac{{p'}_{\max}^2}{m_{\min} s_{\min}^3} + \frac{N_f k_B T}{s_{\min}} \right) \frac{\log d_{p_s}}{h_{p_s}} \Bigg) \\
        a_L \in \mathcal{O}     & \left(\widetilde{N}\log B + \log\frac{d}{\delta\gamma} + \log\frac{\alpha_L}{\varepsilon_L}\right)
    \end{align*}
    and $ d := \max( d_x, d_{p'}, d_s, d_{p_s} )$, can be prepared with $\mathcal{O}(\frac{1}{\delta})$ queries to $U_I$, and
    \begin{equation*}
        \widetilde{\mathcal O}  \Bigg(\left(\frac{\lambda}{\delta \gamma}\left( N + \log B \left(\widetilde N+ \log\frac{B g \alpha_L}{\varepsilon_L}\right)\right) + d\right)\log \frac{g\alpha_L}{\varepsilon_L}\Bigg)
    \end{equation*}
    Toffoli gates, where $g := \max( g_x, g_{p'}, g_s, g_{p_s} )$.
    \label{propLiou}
\end{proposition}
\begin{proof}
    To take the combination of the electronic and classical Liouvillian, we first note that by \cref{propElecLiou} and \cref{lemClassLiou}, we have a $(1, a_{L_\tel}, \frac{\varepsilon_{L_{\tel}}}{\alpha_{L_{\tel}}})$-block-encoding of $\frac{L_\tel}{\alpha_{L_{\tel}}}$ and a $(1, a_{\mathrm{NVT}}, \frac{\varepsilon_{\mathrm{NVT}}}{\alpha_{\mathrm{NVT}}})$-block-encoding of $\frac{L_\mathrm{NVT}}{\alpha_{\mathrm{NVT}}}$. Let $\alpha_L = \alpha_{L_{\tel}} + \alpha_{\mathrm{NVT}}$. Then, given a state preparation unitary that prepares the state
    \begin{equation}
        \sqrt{\frac{\alpha_{L_{\tel}}}{\alpha_{L_{\tel}} + \alpha_{\mathrm{NVT}}}} \ket{0} + \sqrt{\frac{\alpha_{\mathrm{NVT}}}{\alpha_{L_{\tel}} + \alpha_{\mathrm{NVT}}}} \ket{1}
    \end{equation}
    that functions as a $\left( \alpha_L, 1,\varepsilon_{\mathrm{rot}} \right)$-state-preparation-pair, we can construct via \cref{lemLCU} a $(\alpha_L, a_L, \varepsilon_L)$-block-encoding of the Liouvillian $L$ such that
    \begin{align}
        \alpha_L      & = \alpha_{L_{\tel}} + \alpha_{\mathrm{NVT}}, \\
        a_L           & = \max(a_{L_{\tel}},a_{\mathrm{NVT}}) + 1, \\
        \varepsilon_L & = \alpha_L \cdot \max\left(\frac{\varepsilon_{L_{\tel}}}{\alpha_{L_{\tel}}}, \frac{\varepsilon_{\mathrm{NVT}}}{\alpha_{\mathrm{NVT}}}\right) + \varepsilon_{\mathrm{rot}} \, .
    \end{align}
    Dealing with the errors, we set the following to satisfy the above bounds
    \begin{align}
        \varepsilon_{\mathrm{rot}} & \in \mathcal{O}(\varepsilon_L), \\
        \varepsilon_{L_{\tel}}     & \in \mathcal{O}\left(\frac{\alpha_{L_{\tel}}\varepsilon_L}{\alpha_L}\right), \\
        \varepsilon_{\mathrm{NVT}} & \in \mathcal{O}\left(\frac{\alpha_{\mathrm{NVT}}\varepsilon_L}{\alpha_L}\right) \, .
    \end{align}
    We can then use this to capture the bounds on the ancillae, where
    \begin{align}
        a_L & = \max(a_{L_{\tel}},a_{\mathrm{NVT}}) + 1\nonumber \\
            & \in \begin{multlined}[t]\mathcal{O}\bigg(\max\bigg(\left(\widetilde{N}\log B + \log\frac{d_{p'}}{\delta\gamma}+\log \frac{\alpha_{L_\tel}}{\varepsilon_{L_{\tel}}}\right),\\
            \left(\log d+ \log \frac{\alpha_{NVT}}{\varepsilon_{\mathrm{NVT}}} \right)\bigg)\bigg)\end{multlined}\nonumber \\
            & \subseteq \mathcal{O}\left(\widetilde{N}\log B + \log\frac{d\alpha_L}{\delta\gamma\varepsilon_L}\right).
    \end{align}
    Now, for the cost of implementation, we note that we simply have to add the implementation cost of both the classical and electronic Liouvillian, plus the cost for the single rotation gate (which is negligible), and adjust for errors, which we can find to be as follows:
    \begin{equation}
        \widetilde{\mathcal O}  \Bigg(\left(\frac{\lambda}{\delta \gamma}\left( N + \log B \left(\widetilde N+ \log\frac{B g \alpha_L}{\varepsilon_L}\right)\right) + d\right)\log \frac{g\alpha_L}{\varepsilon_L}\Bigg).
    \end{equation}
\end{proof}

\subsection{Implementing the Liouvillian simulation algorithm}

The final step in the Liouvillian simulation algorithm involves Hamiltonian simulation of $L$. We refer the reader to \cref{algoLS} for the full algorithmic description that includes the construction of the electronic Liouvillian in prior sections.
\begin{theorem}[Liouvillian simulation]
    Suppose we have a nuclear-position-controlled electronic Hamiltonian that can be block-encoded in the first quantization scheme with scaling factor $\lambda$ satisfying \cref{assumptionGS} in regards to spectral gap $\gamma$, ground state upper bound $\mu$. We are given an approximate ground state preparation oracle $U_I$ that can prepare ground states up to an overlap of $1-\delta$. Given an initial state $\ket{\rho_0}$ that encodes the initial discretized phase-space density and a discretized Liouvillian under the NVT canonical ensemble, there exists a quantum algorithm that outputs a quantum state that is $\varepsilon$-close in $\ell_2$ distance to $\ket{\rho_t} = e^{-iLt}\ket{\rho_0}$ using
    \begin{multline*}
        \widetilde{\mathcal{O}} \Bigg(\bigg(\alpha_L t + \log\frac{1}{\varepsilon} \bigg)\bigg(\frac{\lambda}{\delta \gamma} \log \frac{g}{\varepsilon}\\
            \times\left( N + \widetilde{N} \log B + \log B \log\frac{B g}{\varepsilon} \right) + d \log \frac{g}{\varepsilon} \bigg) \Bigg)
    \end{multline*}
    Toffoli gates,
    \begin{align*}
        \widetilde{\mathcal{O}} \left( \left(\alpha_L t + \log\frac{1}{\varepsilon} \right) \frac{1}{\delta}\right)
    \end{align*}
    queries to $U_I$, and
    \begin{align*}
        \widetilde{\mathcal{O}}\left(N \log g + \widetilde{N}\log B + \log\frac{d}{\delta\gamma} + \log\frac{\alpha_L t }{\varepsilon}\right)
    \end{align*}
    qubits, where $\alpha_L$ is as defined in \cref{propLiou}, $d := \max(d_x, d_{p'}, d_s, d_{p_s} ) $ and $ g := \max(g_x, g_{p'}, g_s, g_{p_s})$.
    \label{thmLiouSim}
\end{theorem}
\begin{proof}
    Set $\varepsilon_L = \frac{\varepsilon}{2t}$, and apply the block encoding constructed in \cref{propLiou} to \cref{lemHamSim,corHamSimTof}. This provides the number of Toffoli gates directly.
\end{proof}
Note that if we consider only the particle numbers, the spectral gap, approximation ground state preparation overlap, and error precision, the runtime can be written as
\begin{align}
     & \widetilde{\mathcal{O}} \left(\frac{\widetilde N N_{\mathrm{tot}}}{\delta \gamma} \log \frac{1}{\varepsilon}\left(N_{\mathrm{tot}}+\log\frac{1}{\varepsilon}\right)\left( N N_{\mathrm{tot}} t + \log\frac{1}{\varepsilon}\right)\right)\nonumber \\
     & \subseteq \widetilde{\mathcal{O}}\left(\frac{N \widetilde N N_{\mathrm{tot}}^3 t}{\delta\gamma}\log^3\frac{1}{\varepsilon}\right) \, ,
\end{align}
as shown in the informal version of the theorem (\cref{thmLiouSimInfml}) in the main text.

\subsection{Block encoding results for the NVE microcanonical ensemble}
\label{appNVE}

Here we present the results for the NVE microcanonical ensemble.

\begin{proposition}[Block encoding of the microcanonical NVE Liouvillian]
    Suppose we have a nuclear-position-controlled electronic Hamiltonian that can be block-encoded in the first quantization scheme with scaling factor $\lambda$ satisfying \cref{assumptionGS} in regards to spectral gap $\gamma$, ground state upper bound $\mu$. We are given an approximate ground state preparation oracle $U_I$ that can prepare ground states up to an overlap of $1-\delta$. Then there exists a $(\alpha_L, a_L, \varepsilon_L)$-block-encoding of the microcanonical NVE Liouvillian $L$ where
    \begin{align*}
        \alpha_L \in \mathcal O & \Bigg(N\widetilde N\frac{Z_{\max}}{h_{\tel}^2}\frac{\log d_{p}}{h_{p}} + N^2 \frac{Z_{\max}^2 x_{\max}}{\Delta^3} \frac{\log d_{p}}{h_{p}} \\
                                & + N \frac{p_{\max}}{m_{\min}} \frac{\log d_x}{h_x} \Bigg) \\
        a_L \in \mathcal O      & \left(\widetilde{N}\log B + \log\frac{d}{\delta\gamma} + \log\frac{\alpha_L}{\varepsilon_L}\right)
    \end{align*}
    and $ d := \max( d_x, d_{p})$, can be prepared with $\mathcal{O}(\frac{1}{\delta})$ queries to $U_I$, and
    \begin{equation*}
        \widetilde{\mathcal O}  \Bigg(\left(\frac{\lambda}{\delta \gamma}\left( N + \log B \left(\widetilde N+ \log\frac{B g \alpha_L}{\varepsilon_L}\right)\right) + d\right)\log \frac{g\alpha_L}{\varepsilon_L}\Bigg)
    \end{equation*}
    Toffoli gates, where $g := \max(g_x, g_{p})$.
\end{proposition}
\begin{proof}
Proof follows trivially by replacing the NVT Liouvillian with the NVE Liouvillian in \cref{thmLiouSim}.
\end{proof}

\section{Construction and proof of the thermodynamic integration algorithm}

\label{appThermoInt}
In the thermodynamic integration algorithm, we integrate the nuclear Hamiltonian difference over an interpolated set of Liouvillian-equilibrated states, which can be achieved by a Hadamard test upon the difference of nuclear Hamiltonians. We first discuss the bounds on the discretization error from taking the Riemann sum as an estimation instead of integration. Next, we implement the interpolated Liouvillian evolution states, followed by the implementation of the nuclear Hamiltonian difference block encoding. Lastly, we calculate the cost of the full thermodynamic integration algorithm.

Although thermal averages of systems are not time-dependent, the algorithm that we implement produces the free energy estimation results in the computed value being dependent on the evolution time. To minimize the amount of dependencies and better error bounds from taking the Riemann sum across interpolated systems in the thermodynamic integration algorithm, we fix the equilibration time $t_{\rm eq}$ across different $\Lambda$-dependent Liouvillian simulations in our proof and construction of the algorithm. Given the potentially different scaling factors of the Liouvillian block encodings, we need a method to perform different Hamiltonian simulations in superposition.

However, we note that in practice, as long as the equilibration time is long enough and assuming that the evolved state settles into the thermal state such that the computed values do not fluctuate much, given the time-independence of thermal averages, the equilibration time $t_{\rm eq}$ need not be consistent across different $\Lambda$-s, and one can simply use the same Hamiltonian simulation circuit across all $\Lambda$-s, provided that all states are sufficiently equilibrated. Regardless, the techniques developed in this section are still valuable outside of the convenience of theoretical guarantees, as they can be slightly modified to allow for coherent superposed evolutions over different evolution times for better time-independent estimates of free energy, as mentioned in the main text.

\subsection{Upper bounds on the discretization error}
\label{appThmIntDiscErr}
In this section, we take a deeper look at the discretization error from approximating the thermodynamic integration with a Riemann sum approximation.

\begin{lemma}[Discretization error of thermodynamic integration by left Riemann sum]
    Given two systems $A$ and $B$, and suppose for both systems, we are given Liouvillians $L_A$ and $L_B$ under the NVT canonical ensemble encoding the dynamics, as well as nuclear Hamiltonians $H_A$ and $H_B$. Further, we are given an initial state $\ket{\rho_0}$ that encodes the initial discretized phase-space density as well as an equilibration time $t_{\mathrm{eq}}$ for Liouvillian simulation by \cref{thmLiouSim}. Then the discretization error from approximating the thermodynamic integration algorithm by computing the Riemann sum of $N_\Lambda$ discretized points with the left rectangle rule is upper bounded by
    \begin{equation*}
        \frac{\lVert L_B-L_A\rVert\cdot\lVert H_B-H_A\rVert t_{\mathrm{eq}}}{N_\Lambda} \, .
    \end{equation*}
    \label{lemDiscLeft}
\end{lemma}
\begin{proof}
    We begin by observing that the thermodynamic integration algorithm obtains the free energy difference of two systems $A$ and $B$ via the following equation:
    \begin{equation}
    \Delta F_{A\to B} = \int_0^1 \left\langle H_B - H_A \right\rangle_{\Lambda} \mathrm d\Lambda.
    \end{equation}
    By the left rectangle rule, we can approximate the above integral as follows:
    \begin{equation}
        \Delta \widetilde F = \frac{1}{N_\Lambda} \sum_{\bar \Lambda=0}^{N_\Lambda-1} \braket{H_B-H_A}_\Lambda \, ,
    \end{equation}
    where $\bar \Lambda = N_\Lambda \cdot \Lambda$. The discretization error arising from this approximation can then be upper-bounded as follows:
    \begin{equation}
        \left\lvert\Delta F_{A\to B} - \Delta \widetilde F\right\rvert \le \left\lvert\frac{\partial \braket{H_B-H_A}_\Lambda}{\partial \Lambda}\right\rvert\frac{1}{2N_\Lambda} \, .
    \end{equation}
    Evaluating the first fraction, we obtain by the triangle inequality,
    \begin{align}
         & \left\lvert\frac{\partial \braket{H_B-H_A}_\Lambda}{\partial \Lambda}\right\rvert\nonumber \\
         & = \left\lvert\frac{\partial}{\partial \Lambda} \bra{\rho_0}e^{iL_\Lambda t_{\mathrm{eq}}} (H_B-H_A) e^{-iL_\Lambda t_{\mathrm{eq}}}\ket{\rho_0}\right\rvert\nonumber \\
         & \begin{multlined}
               = \bigg\lvert\bra{\rho_0}\left(\frac{\partial}{\partial \Lambda}e^{iL_\Lambda t_{\mathrm{eq}}}\right) (H_B-H_A) e^{-iL_\Lambda t_{\mathrm{eq}}}\ket{\rho_0}\\
               +\bra{\rho_0}e^{iL_\Lambda t_{\mathrm{eq}}} (H_B-H_A) \left(\frac{\partial}{\partial \Lambda}e^{-iL_\Lambda t_{\mathrm{eq}}}\right)\ket{\rho_0}\bigg\rvert
           \end{multlined}\nonumber \\
         & \begin{multlined}[b]
               \le \left\lvert\bra{\rho_0}\left(\frac{\partial}{\partial \Lambda}e^{iL_\Lambda t_{\mathrm{eq}}}\right) (H_B-H_A) e^{-iL_\Lambda t_{\mathrm{eq}}}\ket{\rho_0}\right\rvert\\
               + \left\lvert\bra{\rho_0}e^{iL_\Lambda t_{\mathrm{eq}}} (H_B-H_A) \left(\frac{\partial}{\partial \Lambda}e^{-iL_\Lambda t_{\mathrm{eq}}}\right)\ket{\rho_0}\right\rvert \, .
           \end{multlined}
    \end{align}
    We evaluate the first term such that by the tracial matrix H\"older's inequality~\citep{holder1889ueber,horn1991topics,baumgartner2011inequality} and unitarily invariant properties of spectral/Schatten norms~\citep{schatten1960norm}, we can obtain
    \begin{align}
         & \left\lvert\bra{\rho_0}\left(\frac{\partial}{\partial \Lambda}e^{iL_\Lambda t_{\mathrm{eq}}}\right) (H_B-H_A) e^{-iL_\Lambda t_{\mathrm{eq}}}\ket{\rho_0}\right\rvert\nonumber \\
         & =\left\lvert\tr\left(\left(\frac{\partial}{\partial \Lambda}e^{iL_\Lambda t_{\mathrm{eq}}}\right) (H_B-H_A) e^{-iL_\Lambda t_{\mathrm{eq}}}\ketbra{\rho_0}{\rho_0}\right)\right\rvert\nonumber \\
         & \le \left\lVert\left(\frac{\partial}{\partial \Lambda}e^{iL_\Lambda t_{\mathrm{eq}}}\right) (H_B-H_A) e^{-iL_\Lambda t_{\mathrm{eq}}}\right\rVert \tr\left(\left\lvert\ketbra{\rho_0}{\rho_0}\right\rvert\right)\nonumber \\
         & \le \left\lVert\frac{\partial}{\partial \Lambda}e^{iL_\Lambda t_{\mathrm{eq}}}\right\rVert \left\lVert H_B-H_A\right\rVert \left\lVert e^{-iL_\Lambda t_{\mathrm{eq}}}\right\rVert \tr\left(\left\lvert\ketbra{\rho_0}{\rho_0}\right\rvert\right) \nonumber \\
         & = \left\lVert\frac{\partial}{\partial \Lambda}e^{iL_\Lambda t_{\mathrm{eq}}}\right\rVert \left\lVert H_B-H_A\right\rVert \, .
    \end{align}
    Then by Lemma 50 of Ref.~\citep{chakraborty2019power}, which states
    \begin{equation}
        \lVert e^{-iHt} - e^{-iH't} \rVert \le \lVert H - H' \rVert t \, .
    \end{equation}
    We find that
    \begin{align}
         & \left\lVert\frac{\partial}{\partial \Lambda}e^{-iL_\Lambda t_{\mathrm{eq}}}\right\rVert \nonumber \\
         & = \lim_{\Delta \Lambda\to0}\frac{\left\lVert e^{-i(L_A +(\Lambda +\Delta \Lambda)(L_B-L_A) t_{\mathrm{eq}}} - e^{-i(L_A +\Lambda (L_B-L_A) t_{\mathrm{eq}}}\right\rVert}{\Delta \Lambda} \nonumber \\
         & \le \lim_{\Delta \Lambda\to0} \frac{(\Lambda +\Delta \Lambda)\lVert L_B-L_A\rVert t_{\mathrm{eq}} - \Lambda \lVert L_B-L_A\rVert t_{\mathrm{eq}}}{\Delta \Lambda}\nonumber \\
         & = \lVert L_B-L_A\rVert t_{\mathrm{eq}} \, .
    \end{align}
    We can then see that
    \begin{align}
         & \left\lvert\bra{\rho_0}\left(\frac{\partial}{\partial \Lambda}e^{iL_\Lambda t_{\mathrm{eq}}}\right) (H_B-H_A) e^{-iL_\Lambda t_{\mathrm{eq}}}\ket{\rho_0}\right\rvert\nonumber \\
         & \le \|L_B-L_A\|\cdot\|H_B-H_A\|t_{\mathrm{eq}} \, .
    \end{align}
    Likewise, we can show that
    \begin{align}
         & \left\lvert\bra{\rho_0}e^{iL_\Lambda t_{\mathrm{eq}}} (H_B-H_A) \left(\frac{\partial}{\partial \Lambda}e^{-iL_\Lambda t_{\mathrm{eq}}}\right)\ket{\rho_0}\right\rvert\nonumber \\
         & \le \|L_B-L_A\|\cdot\|H_B-H_A\|t_{\mathrm{eq}} \, .
    \end{align}
    Thus, one can show that the discretization error is upper-bounded as shown
    \begin{align}
        \left\lvert\Delta F_{A\to B} - \Delta \widetilde F\right\rvert & \le \left\lvert\frac{\partial \braket{H_B-H_A}_\Lambda}{\partial \Lambda}\right\rvert\frac{1}{2N_\Lambda}\nonumber \\
        & \le\frac{\|L_B-L_A\|\cdot\|H_B-H_A\|t_{\mathrm{eq}}}{N_\Lambda} \, .
    \end{align}
\end{proof}

While one can obtain tighter bounds by obtaining high-order numerical integration approximation methods and showing the proof via iteratively applying the Lie-Trotter formula~\citep{trotter1959product}, this is not required unless QSVT phase angles are loaded into the quantum circuit via lookup tables, which would incur an additional $\mathcal{O} (N_\Lambda)$ cost per query to the Liouvillian. If the angles are computed on separate quantum registers, or avoided via QSP without angle finding, then the dependency of the runtime would be at most $\mathcal{O}(\poly\log N_\Lambda)$ per query to the Liouvillian, to which higher-order approximations would not help much as the dependency is already logarithmic.

\subsection{Derivations of \texorpdfstring{$\Lambda$}{Lambda}-dependent Liouvillian simulation}
\label{appLambdaLiouvillian}

In this section, we further describe methods for $\Lambda$-dependent Liouvillian simulation discussed in \cref{secTI}. First, as mentioned in the main text, we can implement a $\Lambda$-dependent Liouvillian simulation by implementing the entire QSVT algorithm in superposition, controlling both the Liouvillian and the phase angles by $\Lambda$. Phase angles are computed in a separate register that implements phase-finding algorithms, the details of which are outside the scope of this paper.

\begin{figure*}
    \includegraphics[width=\textwidth]{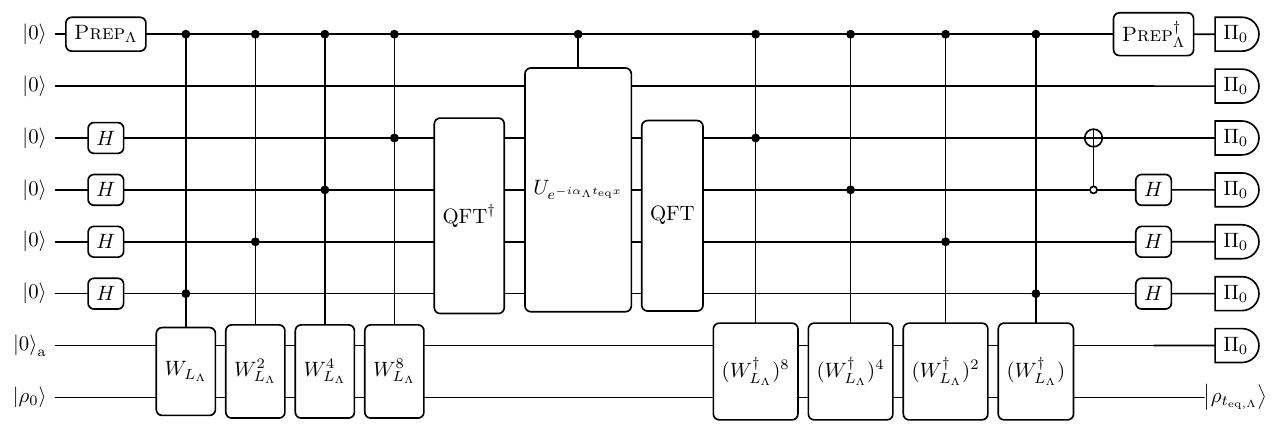}
    \caption[$\Lambda$-dependent Hamiltonian simulation of $L_\Lambda$ using QSP without angle finding.]{\emph{$\Lambda$-dependent Hamiltonian simulation of $L_\Lambda$ using QSP without angle finding}~\citep{alase2025quantum}. The central $\Lambda$-controlled unitary is diagonal with the entries $\exp(-i\alpha_\Lambda t \cos\frac{2\pi k}{4\mathcal D})$ and can be performed using $\log^2(\frac{N_\Lambda}{\varepsilon})$ gates. In this example figure, we assume polynomial degree $\mathcal{D}=4$. The gate $W_{L_\Lambda}$ is the qubitization of the block-encoded Liouvillian $U_{L_\Lambda}$ such that $W_{L_\Lambda} = (2\Pi_a-I_a)U_{L_\Lambda}$.}
    \label{fig:lambda_simulation_angleless}
\end{figure*}

Alternatively, one can apply the robust Hamiltonian simulation in superposition by utilizing quantum signal processing without angle finding as shown in \cref{propAnglelessQSP,propAnglelessHamSim}. Corresponding to the methodology of constructing Hamiltonian simulation in superposition for standard QSVT, which is to control the phase angles, in angleless QSP, to construct Hamiltonian simulation in superposition, we produce the functional diagonal block encoding $U_{f, 4\mathcal{D}}$ in superposition, such that we have
\begin{equation}
    \sum_{\Lambda} \ketbra{\Lambda}{\Lambda} \otimes U_{f_\Lambda, 4\mathcal{D}} \, ,
    \label{eqFuncUni}
\end{equation}
where $\mathcal{D}$ is chosen to be the maximum degree for polynomial expansions for all $\Lambda$, or
\begin{equation}
    \mathcal{D} \in \mathcal {O} \left(\max(\alpha_{L_A}, \alpha_{L_B})t_{\mathrm{eq}} + \log \frac{1}{\varepsilon}\right),
\end{equation}
and
\begin{equation}
    f_\Lambda(x) = \exp(-i\alpha_\Lambda t_{\mathrm{eq}} \cos x) \, .
\end{equation}

For simplicity, in the angleless QSP setting, we approximate the thermodynamic integration algorithm via the left Riemann sum per \cref{lemDiscLeft} such that
\begin{equation}
    N_\Lambda = \frac{\lVert L_B-L_A\rVert\lVert H_B-H_A\rVert t_{\mathrm{eq}}}{\varepsilon} \in \mathcal O \left(\frac{\alpha_L\alpha_Ht_{\mathrm{eq}}}{\varepsilon}\right).
\end{equation}
To detail the implementation of \cref{eqFuncUni}, we first prepare the state
\begin{equation}
    \sum_{\bar \Lambda=0}^{N_\Lambda - 1} \sum_{k=0}^{4\mathcal{D}-1} \ket{\Lambda}\ket{k} \, .
\end{equation}
Then by implementing the CORDIC algorithm~\citep{volder1959cordic} via quantum arithmetic circuits~\citep{vedral1996quantum} and pre-stored classical values on lookup tables, whose values can just be directly provided on ancilla registers, or via a modification of the quantum CORDIC algorithm~\citep{burge2024cordic}, one can implement the $\cos$ function up to $\varepsilon$ additive precision with $\mathcal O(\log^2 \frac{\mathcal{D}}{\varepsilon})$ Toffoli gates. We then produce the state
\begin{equation}
    \sum_{\bar \Lambda=0}^{N_\Lambda - 1} \sum_{k=0}^{4\mathcal{D}-1} \ket{\Lambda}\ket{k}\Ket{\cos \frac{2\pi k}{4\mathcal{D}}} \, .
\end{equation}
By further use of quantum arithmetic circuits, we can also obtain
\begin{equation}
    \sum_{\bar \Lambda=0}^{N_\Lambda - 1} \sum_{k=0}^{4\mathcal{D}-1} \ket{\Lambda}\ket{k}\ket{\alpha_\Lambda}\Ket{\cos \frac{2\pi k}{4\mathcal{D}}}\Ket{\alpha_\Lambda t_{\mathrm{eq}}\cos \frac{2\pi k}{4\mathcal{D}}}
\end{equation}
using an additional $\mathcal O(\log \frac{\mathcal{D}N_\Lambda}{\varepsilon})$ Toffoli gates. Then controlled on the final fixed-point register of $\alpha_\Lambda t_{\mathrm{eq}}\cos \frac{2\pi k}{4\mathcal{D}}$, we can perform a rotational Z transform to an ancilla qubit or via the alternating sign trick~\citep{berry2014exponential} such that we have (ignoring intermediate states)
\begin{equation}
    \sum_{\bar \Lambda=0}^{N_\Lambda - 1} \sum_{k=0}^{4\mathcal{D}-1} \ket{\Lambda}\ket{k}\Ket{\alpha_\Lambda t_{\mathrm{eq}}\cos \frac{2\pi k}{4\mathcal{D}}}e^{-i\alpha_\Lambda t_{\mathrm{eq}}\cos \frac{2\pi k}{4\mathcal{D}}}\ket{0}
\end{equation}
with $\mathcal O(\log^2 \frac{\mathcal{D}N_\Lambda}{\varepsilon})$ Toffoli gates.
Then, uncomputing the arithmetic circuits, we obtain
\begin{equation}
    \sum_{\bar \Lambda=0}^{N_\Lambda - 1} \sum_{k=0}^{4\mathcal{D}-1} \ket{\Lambda}\ket{k}e^{-i\alpha_\Lambda t_{\mathrm{eq}}\cos \frac{2\pi k}{4\mathcal{D}}}\ket{0}\,.
\end{equation}
Notice that this can then be used to provide a
$(1, 1, \frac{\varepsilon}{24\sqrt{2}})$-block-encoding of $D_{f_\Lambda, 4\mathcal{D}}$ in superposition as shown in \cref{eqFuncUni} for which the implementation costs
\begin{equation}
    \widetilde{\mathcal O}\left(\log^2 \frac{\alpha_L\alpha_Ht_{\mathrm{eq}}}{\varepsilon}\right)
\end{equation}
Toffoli gates.

Further accounting for the gates to prepare the weighted sum between $L_A$ and $L_B$, we note that an additional $\mathcal{O} (\log^2 \frac{N_\Lambda}{\varepsilon})$ gates are needed for state preparation per access to $L_\Lambda$.  Hence, the runtime for preparing the $L_\Lambda$ simulation in superposition of $\Lambda$ in total would still be asymptotically the same as preparing a single Liouvillian simulation when polylogarithmic factors are disregarded. We provide a circuit for this approach in \cref{fig:lambda_simulation_angleless} (based on a circuit from Ref.~\citep{alase2025quantum}). Lastly, we note that although controlled versions of the qubitized Liouvillian are needed, as the ground state oracle is only used for state preparation that is then uncomputed, controlled versions of the approximate ground state oracle are not required.

\subsection{Implementation of nuclear Hamiltonian}
\label{appInHam}
Recall from \cref{eqNucHam} that the nuclear Hamiltonian $H_{\rm nuc}$ can be decomposed into $H_{\mathrm{kin}}$, $H_{\mathrm{pot}}$, $H_{\mathrm{gse}}$, which correspond to the kinetic energy, potential energy, and ground state energy values at relevant points in phase-space:
\begin{align}
    H_{\mathrm{kin}} & = \sum_{n, j} \sum_{\bar{p}'_{n,j}} \sum_{\bar{s}} \frac{{p'}_{n,j}^2}{2m_n (s + s_{\min})^2} \ketbra{\bar{p}'_{n,j}}{\bar{p}'_{n,j}} \otimes \ketbra{\bar{s}}{\bar{s}} \, , \\
    H_{\mathrm{pot}} & = \sum_{n' \ne n} \sum_{\bar{x}_n} \sum_{\bar{x}_{n'}} \frac{Z_n Z_{n'}}{\left( \| x_n - x_{n'} \|^2 + \Delta^2 \right)^{1/2}} \ketbra{\bar{x}_n}{\bar{x}_n} \nonumber \\
                     & \otimes \ketbra{\bar{x}_{n'}}{\bar{x}_{n'}} \, , \\
    H_{\mathrm{gse}} & = \sum_{\vec{\bar{x}}} E_{\tel}(\vec x) \ketbra{\vec{\bar x}}{\vec{\bar x}} \, .
\end{align}
In this section, we review results on block-encoding the kinetic and potential energy operators from Ref.~\citep{simon2024improved}, and we present an improved implementation of the ground state energy operator $H_{\mathrm{gse}}$.

\begin{lemma}[Block encoding of $H_{\mathrm{kin}}$ -- Lemma 22, \citep{simon2024improved}]
    There exists a $(\alpha_{\mathrm{kin}}, a_{\mathrm{kin}}, \varepsilon_{\mathrm{kin}})$-block-encoding of $H_{\mathrm{kin}}$ with normalization constant
    \begin{equation}
        \alpha_{\mathrm{kin}} \in \mathcal{O} \left( N \frac{p'^2_{\max}}{m_{\min} s_{\min}^2} \right)
    \end{equation}
    and the number of ancilla qubits
    \begin{equation}
        a_{\mathrm{kin}} \in \mathcal{O} \left( \log \frac{\alpha_{\mathrm{kin}}}{\varepsilon_{\mathrm{kin}}} \right)
    \end{equation}
    that can be implemented using
    \begin{equation}
        \mathcal{O} \left( N \log \frac{g_{p'} \alpha_{\mathrm{kin}}}{\varepsilon_{\mathrm{kin}}} + \log^{\log 3} \frac{\alpha_{\mathrm{kin}}}{\varepsilon_{\mathrm{kin}}} \right)
    \end{equation}
    Toffoli gates, where $g_{p'}$ is the number of discrete momenta considered in the classical part of the Liouvillian.
    \label{lemKin}
\end{lemma}

\begin{lemma}[Block encoding of $H_{\mathrm{pot}}$ -- Lemma 23, \citep{simon2024improved}]
    There exists a $(\alpha_{\mathrm{pot}}, a_{\mathrm{pot}}, \varepsilon_{\mathrm{pot}})$-block-encoding of $H_{\mathrm{pot}}$ with normalization constant
    \begin{equation}
        \alpha_{\mathrm{pot}} \in \mathcal{O} \left( N^2 \frac{Z_{\max}^2}{\Delta} \right)
    \end{equation}
    and the number of ancilla qubits
    \begin{equation}
        a_{\mathrm{pot}} \in \mathcal{O} \left( \log \frac{\alpha_{\mathrm{pot}}}{\varepsilon_{\mathrm{pot}}} \right) \, .
    \end{equation}
    This block encoding can be implemented using
    \begin{equation}
        \mathcal{O} \left( N \log \frac{g_x \alpha_{\mathrm{pot}}}{\varepsilon_{\mathrm{pot}}}  + \log^{\log 3} \frac{\alpha_{\mathrm{pot}}}{\varepsilon_{\mathrm{pot}}}  \right)
    \end{equation}
    Toffoli gates, where $g_x$ is the number of discrete positions considered in the classical part of the Liouvillian.
    \label{lemPot}
\end{lemma}

Given that the electronic Hamiltonian can be accessed in superposition, to construct $H_{\mathrm{gse}}$, one can obtain the ground state energy of the electronic Hamiltonian $H_{\tel}(\vec x)$ for different values of $\vec{\bar x}$ in superposition. Ref.~\citep{simon2024improved} achieves this by preparing a separate set of electronic registers and preparing the ground states of the electronic Hamiltonian via a ground state preparation algorithm by \citet{lin2020nearoptimal} for different nuclear positions in superposition. Then, using phase estimation, fixed point representations of the ground state energies are extracted, after which one can then approximately encode the fixed point representations into amplitudes of the quantum state via the alternating sign trick~\citep{berry2014exponential}. However, this construction is largely bottlenecked by phase estimation, which requires $\mathcal O(\frac{1}{\varepsilon})$ applications of the block encoding of the electronic Hamiltonian $U_{H_{\tel}}$. In contrast, our implementation requires only $\mathcal O(\poly\log\frac{1}{\varepsilon})$ applications of $U_{H_{\tel}}$.

In our work, similar to Ref.~\citep{simon2024improved}, we implement $H_{\mathrm{gse}}$ by first preparing the ground state with the preparation algorithm by \citet{lin2020nearoptimal}, but instead of using phase estimation to extract the eigenvalues, we then directly apply the block encoding of the Hamiltonian such that the scaled eigenvalues are directly encoded in the amplitude of the ground states in superposition, followed lastly by uncomputing the ground state, similar to our strategy of implementing the force operator in \cref{lemHamDeriv}. This produces a block encoding implementation of $H_{\mathrm{gse}}$ that is only dependent on the precision logarithmically in terms of runtime. The following is therefore an improved version of Lemma 24 of Ref.~\citep{simon2024improved}.
\begin{lemma}[Block encoding of $H_{\mathrm{gse}}$]
    Suppose we have a $(\lambda, a_{\tel}, \varepsilon_{\tel})$-block-encoding $U_H$ of the nuclear-position-controlled electronic Hamiltonian satisfying \cref{assumptionGS} in regards to spectral gap $\gamma$, ground state upper bound $\mu$, and we are given an approximate ground state preparation oracle $U_I$ that can prepare ground states up to an overlap of $1-\delta$. If
    \begin{equation*}
        \varepsilon_{\tel}\in\widetilde{\mathcal O} \left(\frac{\delta^2\gamma^2\varepsilon_{\mathrm{gse}}^4}{\lambda^5}\right),
    \end{equation*}
    then there exists a $(\lambda, a_{\mathrm{gse}},\varepsilon_{\mathrm{gse}})$-block-encoding of $H_{\mathrm{gse}}$ where
    \begin{equation*}
        a_{\mathrm{gse}} \in a_{\tel}+\mathcal{O}(\widetilde{N}\log B).
    \end{equation*}
    This block encoding can be prepared with $\mathcal{O}(\frac{\lambda}{\delta\gamma}\log\frac{\lambda}{\delta\varepsilon_{\mathrm{gse}}})$ queries to $U_H$, $\mathcal{O}(\frac{1}{\delta})$ queries to $U_I$ and an additional $\mathcal O \left( \frac{\lambda}{\delta \gamma} \log \frac{\lambda}{\delta \varepsilon_{\mathrm{gse}}}  \left( a_{\tel}+\log \frac{\lambda}{\delta \gamma\varepsilon_{\mathrm{gse}}}\right) \right)$ Toffoli gates.
    \label{lemGSE}
\end{lemma}

\begin{proof}
    First, we demonstrate the correctness of the block encoding construction. Assuming perfect block encoding and ground state preparation, we note that preparing the ground state in superposition gives us
    \begin{equation}
        \sum_{\vec{\bar x}}\ket{\vec{\bar x}}\ket{0}_{\mathrm{anc}} \to \sum_{\vec{\bar x}}\ket{\vec{\bar x}}\ket{\psi_0(\vec x)}\ket{0}_{\mathrm{anc}} \, .
    \end{equation}
    Then, applying the controlled Hamiltonian, we obtain
    \begin{equation}
        \sum_{\vec{\bar x}}\ket{\vec{\bar x}}\bigg(\frac{1}{\lambda} H_{\tel}(\vec x) \ket{\psi_0(\vec x)}\ket{0}_{\mathrm{anc}}
        +\sum_{j}\lambda_j'\ket{\mathrm{bad}_j}\ket{j}_{\mathrm{anc}}\bigg) \, .
    \end{equation}
    where $\lambda'_j$ is a normalization factor for some bad state $\ket{\mathrm{bad}_j}$. Noting that $H_{\tel}(\vec x) \ket{\psi_0(\vec x)} = E_0(\vec x) \ket{\psi_0(\vec x)}$, if we then uncompute the ground state, we obtain
    \begin{equation}
        \sum_{\vec{\bar x}}\ket{\vec{\bar x}}\bigg(\frac{1}{\lambda} E_0(\vec x) \ket{0}_{\mathrm{anc}}+\sum_{j}\lambda_j'\ket{\mathrm{bad}_j'}\ket{j}_{\mathrm{anc}}\bigg) \, .
    \end{equation}
    Then, post-selecting on $\ket{0}$ for the ancilla registers and the electronic register, we see that we would obtain
    \begin{equation}
        \sum_{\vec{\bar x}}\frac{1}{\lambda} E_0(\vec x)\ket{\vec{\bar x}} \, .
    \end{equation}

    Now, to the factors of the block encoding, the scaling factor $\lambda$ of the block encoding follows from the scaling factor $H_{\tel}$. Given that the electronic register is computed and then uncomputed, in regards to encoding the ground state energy onto the nuclear register, the ancilla qubits are the sum of the ancilla qubits of the $H_{\tel}$, or $a_{\tel}$, and the size of the electronic register, or $\mathcal{O}(\widetilde N \log B)$.

    Regarding the error, we first prepare the ground state of $ H^{\rm ctrl}_{\tel}$ up to fidelity $1-\varepsilon_{\mathrm{prep}}$ with \cref{lemGSPrep}. This produces a state $\ket{\phi_0}$ such that $\lvert(\bra{\psi_0}\otimes\bra{0})\ket{\phi_0}\rvert \geq 1-\varepsilon_{\mathrm{prep}}$. Similar to \cref{lemHamDeriv}, for cleanliness, we denote $\ket{\widetilde \psi_0} = \ket{\psi_0}\otimes\ket{0}$. Then the $\ell_2$ distance between the two states can be found to be
    \begin{equation}
        \left\lVert\ket{\phi_0}-\ket{\widetilde \psi_0}\right\rVert \leq \sqrt{2\varepsilon_{\mathrm{prep}}} \, .
    \end{equation}
    Let $W$ be the unitary operator that prepares the approximate ground state $\ket{ \phi_0}$ from $\ket{0}$, composed of $U_I$ and the ground state preparation procedure in \cref{lemGSPrep}. Given that the final block encoding is that of a diagonal matrix on different nuclear positions, the error term of the block encoding (spectral norm difference) is less than that of the maximum difference on a single entry. We derive the difference as follows via a series of triangular inequalities and drop the notation on nuclear positions, as the derivation is for a single position:
    \begin{align}
         & \left\lvert\lambda\bra{0}W^\dagger U_{H_{\tel}} W\ket{0} - E_0\right\rvert \nonumber \\
         & \begin{aligned}[b]
               \le\; & \lambda\left\lvert\bra{0} W^\dagger U_{H_{\tel}}W\ket{0} - \bra{0} W^\dagger U_{H_{\tel}}\ket{\widetilde \psi_0}\right\rvert \\
               & + \lambda\left\lvert\bra{0}W^\dagger U_{H_{\tel}} \ket{\widetilde \psi_0} - \bra{\widetilde\psi_0} U_{H_{\tel}} \ket{\widetilde \psi_0}\right\rvert \\
               & +\left\lvert \bra{\widetilde\psi_0}   U_{H_{\tel}} \ket{\widetilde \psi_0} - \bra{0}  E_0\ket{0}\right\rvert \, .
           \end{aligned}
    \end{align}
    We now derive the upper bounds of the four absolute differences. The first term can be bounded by the Cauchy–Schwarz inequality such that
    \begin{align}
         & \lambda\left\lvert\bra{0} W^\dagger U_{H_{\tel}}W\ket{0} - \bra{0} W^\dagger U_{H_{\tel}}\ket{\widetilde \psi_0}\right\rvert\nonumber \\
         & = \lambda\left\lvert\bra{0}W^\dagger U_{H_{\tel}}\ket{\phi_0} - \bra{0} W^\dagger U_{H_{\tel}}\ket{\widetilde \psi_0}\right\rvert\nonumber \\
         & = \lambda\left\lVert\ket{\phi_0} - \ket{\psi_0}\right\rVert \, .
    \end{align}

    The second term can be similarly bounded as follows:
    \begin{equation}
        \lambda\left\lvert\bra{0}W^\dagger U_{H_{\tel}} \ket{\widetilde \psi_0} - \bra{\widetilde\psi_0} U_{H_{\tel}} \ket{\widetilde \psi_0}\right\rvert\le  \lambda \lVert\ket{\phi_0} - \ket{\psi_0}\rVert \,.
    \end{equation}

    The third term can be upper-bounded as follows:
    \begin{align}
         & \left\lvert\bra{\widetilde\psi_0} U_{H_{\tel}} \ket{\widetilde \psi_0} - \bra{\widetilde\psi_0}  E_0 \ket{\widetilde \psi_0}\right\rvert \nonumber \\
         & = \left\lvert\bra{ \psi_0}\left(\lambda(\mathbb{I}\otimes\bra{0}) U_{H_{\tel}} (\mathbb{I}\otimes\ket{0}) -  H\right) \ket{\psi_0}\right\rvert\nonumber \\
         & \leq \left\lVert\lambda(\mathbb{I}\otimes\bra{0}) U_{H_{\tel}} (\mathbb{I}\otimes\ket{0}) -  H\right\rVert \le \varepsilon_{\tel} \, ,
    \end{align}

    When we combine the results, we get:
    \begin{align}
        \left\lvert\lambda\bra{0}W^\dagger U_{H_{\tel}}W\ket{0} - E_0\right\rvert & \le 2\lambda \lVert\ket{\phi_0} - \ket{\psi_0}\rVert + \varepsilon_{\tel}\nonumber \\
        & \le2\lambda \sqrt{2\varepsilon_{\mathrm{prep}}}  + \varepsilon_{\tel} \, .
    \end{align}
    From Ref.~\citep{su2021faulttolerant}, we know that $\lambda$ is given from the $\ell_1$ norm of the coefficients $\alpha_\ell$ from the LCU decomposition of the electronic Hamiltonian $ H_{\tel} = \sum_{\ell=1}\alpha_\ell  H_{\ell}$, where $ H_{\ell}$ are unitary. Thus, $\lVert  H_{\tel}\rVert$ can be upper bounded by $\lambda$ via the triangular inequality. The error of the block encoding can then be further upper-bounded by
    \begin{equation}
        2\lambda \sqrt{2\varepsilon_{\mathrm{prep}}} + \varepsilon_{\tel} \, .
    \end{equation}
    Set $\varepsilon_{\tel} \le \frac{\varepsilon_{\mathrm{gse}}}{2}$ and $2\lambda \sqrt{2\varepsilon_{\mathrm{prep}}} = \frac{\varepsilon_{\mathrm{gse}}}{2}$. Replacing the bounds of $\varepsilon_{\mathrm{prep}}$ of that obtained in \cref{lemGSPrep}, where we see that
    \begin{equation}
        \varepsilon_{\tel}\in \widetilde{\mathcal O} \left(\frac{\delta^2\gamma^2\varepsilon_{\mathrm{prep}}^2}{\lambda}\right)\subseteq  \widetilde{\mathcal O} \left(\frac{\delta^2\gamma^2\varepsilon_{\mathrm{gse}}^4}{\lambda^5}\right) \, .
    \end{equation}
    For consistency across different $n$, we set $\varepsilon_{\tel}$ to a tighter bound. We then let
    \begin{align}
        \varepsilon_{\mathrm{prep}} & \in \mathcal O \left(\frac{\varepsilon_{\mathrm{gse}}^2}{\lambda^2}\right) \\
        \varepsilon_{\tel}          & \in \widetilde{\mathcal O} \left(\frac{\delta^2\gamma^2\varepsilon_{\mathrm{prep}}^2}{\lambda}\right)\subseteq  \widetilde{\mathcal O} \left(\frac{\delta^2\gamma^2\varepsilon_{\mathrm{gse}}^4}{\lambda^5}\right) \, ,
    \end{align}
    which still satisfies the former bound. The runtime is largely dominated by ground state preparation, hence it has the same asymptotic runtime as the latter. Replacing $\varepsilon_{\mathrm{prep}}$ by its bounds of $\varepsilon_{\mathrm{gse}}$, we note that we require $\mathcal{O}(\frac{\lambda}{\delta\gamma}\log\frac{\lambda}{\delta\varepsilon_{\mathrm{gse}}})$ queries to $U_H$. The number of Toffoli gates needed can be found by \cref{corGSPrepTof} to be
    \begin{equation}
        \mathcal O \left( \frac{\lambda}{\delta \gamma} \log \frac{\lambda}{\delta \varepsilon_{\mathrm{gse}}}  \left( a_{\tel}+\log \left(\frac{\lambda}{\delta \gamma\varepsilon_{\mathrm{gse}}}\right)\right) \right) \, .
    \end{equation}
\end{proof}
Note that the requirement on the error bound of the electronic Hamiltonian is tighter than that of the requirements for implementing the force operator in \cref{lemHamDeriv}. However, due to the error being a logarithmic dependency, this is hidden in the asymptotics. Alternatively, we can remove this requirement if we take the approach in Ref.~\citep{simon2024improved} and assume that the electronic Hamiltonian is a $(\lambda, a_{\tel}, 0)$-block-encoding of an alternative Hamiltonian $\widetilde H$, which we can prepare the ground state to since we technically use this Hamiltonian as for ground state preparation. This requires dependency on an alternate spectral gap $\widetilde \gamma$, which can still be bounded by $\gamma$ by Weyl's inequality, and an alternative overlap bound $\widetilde \delta$ on the approximate ground state oracle $U_I$, which in practice would be roughly $\delta$, but further assumptions are required for theoretical completeness. Since the dependency on error is logarithmic in our scheme, we opt to consider the higher requirements on the electronic Hamiltonian block encoding, in contrast to adding another assumption on the oracle.

We now combine the Hamiltonian terms to produce a block encoding of the nuclear Hamiltonian $H_{\rm nuc}$.
\begin{proposition}[Block encoding of nuclear Hamiltonian]
    Suppose we have a nuclear-position-controlled electronic Hamiltonian that can be block-encoded in the first quantization scheme with scaling factor $\lambda$ satisfying \cref{assumptionGS} in regards to spectral gap $\gamma$, ground state upper bound $\mu$. We are given an approximate ground state preparation oracle $U_I$ that can prepare ground states up to an overlap of $1-\delta$. Then there exists a $(\alpha_H, a_H, \varepsilon_H)$-block-encoding of $H_{\rm nuc}$ where
    \begin{align*}
        \alpha_H & \in \mathcal O \left( \frac{Np'^2_{\max}}{m_{\min} s_{\min}^2} +  \frac{N^2Z_{\max}^2}{\Delta} + \lambda\right), \\
        a_H      & \in \mathcal{O}\left(\widetilde{N}\log B + \log\frac{1}{\delta\gamma}+\log \frac{\alpha_H}{\varepsilon_H}\right) \, .
    \end{align*}
    This block encoding can be prepared with $\mathcal{O}(\frac{1}{\delta})$ queries to $U_I$, and
    \begin{align*}
        \widetilde{\mathcal O} & \Bigg(\frac{\lambda}{\delta \gamma}\log \frac{g\alpha_H}{\varepsilon_H}\left( N + \widetilde N \log B + \log B\log\frac{B g \alpha_H}{\varepsilon_H}\right) \Bigg)
    \end{align*}
    Toffoli gates, where $g = \max(g_x, g_{p'})$.
    \label{propNuclearHam}
\end{proposition}
\begin{proof}
    To take the combination of the electronic and classical Liouvillian, we first note that by \cref{lemKin,lemPot,lemGSE}, we have a $(1, a_{\mathrm{kin}}, \frac{\varepsilon_{\mathrm{kin}}}{\alpha_{\mathrm{kin}}})$-block-encoding of $\frac{H_{\mathrm{kin}}}{\alpha_{\mathrm{kin}}}$, a $(1, a_{\mathrm{pot}}, \frac{\varepsilon_{\mathrm{pot}}}{\alpha_{\mathrm{pot}}})$-block-encoding of $\frac{H_{\mathrm{pot}}}{\alpha_{\mathrm{pot}}}$, and a $(1, a_{\mathrm{gse}}, \frac{\varepsilon_{\mathrm{gse}}}{\lambda})$-block-encoding of $\frac{H_{\mathrm{gse}}}{\lambda}$. Let $\alpha_H = \alpha_{\mathrm{kin}} + \alpha_{\mathrm{pot}} + \lambda$. Then, given a state preparation unitary that prepares the state
    \begin{equation}
        \sqrt{\frac{\alpha_{\mathrm{kin}}}{\alpha_H}} \ket{00} + \sqrt{\frac{\alpha_{\mathrm{pot}}}{\alpha_H}} \ket{01} + \sqrt{\frac{\lambda}{\alpha_H}} \ket{10}
    \end{equation}
    that functions as a $\left( \alpha_H, 2,\varepsilon_{\mathrm{rot}} \right)$-state-preparation-pair, we can construct via \cref{lemLCU} a $(\alpha_H, a_H, \varepsilon_H)$-block-encoding of the Liouvillian $L$ such that
    \begin{align}
        \alpha_H      & = \alpha_{\mathrm{kin}} + \alpha_{\mathrm{pot}} + \lambda \, , \\
        a_H           & = \max(a_{\mathrm{kin}}, a_{\mathrm{pot}}, a_{\mathrm{gse}}) + 2 \, , \\
        \varepsilon_H & = \alpha_H \cdot \max\left(\frac{\varepsilon_{\mathrm{kin}}}{\alpha_{\mathrm{kin}}}, \frac{\varepsilon_{\mathrm{pot}}}{\alpha_{\mathrm{pot}}}, \frac{\varepsilon_{\mathrm{gse}}}{\lambda}\right) + \varepsilon_{\mathrm{rot}} \, .
    \end{align}
    Dealing with the errors, we set the following to satisfy the above bounds
    \begin{align}
        \varepsilon_{\mathrm{rot}} & \in \mathcal{O}(\varepsilon_H) \, , \\
        \varepsilon_{\mathrm{kin}} & \in \mathcal{O}\left(\frac{\alpha_{\mathrm{kin}}\varepsilon_H}{\alpha_H}\right) \, , \\
        \varepsilon_{\mathrm{pot}} & \in \mathcal{O}\left(\frac{\alpha_{\mathrm{pot}}\varepsilon_H}{\alpha_H}\right) \, , \\
        \varepsilon_{\mathrm{gse}} & \in \mathcal{O}\left(\frac{\lambda\varepsilon_H}{\alpha_H}\right) \, , \\
        \varepsilon_{\tel}         & \in \widetilde{\mathcal O} \left(\frac{\delta^2\gamma^2\varepsilon_H^4}{\lambda\alpha_H^4}\right) \, .
    \end{align}
    We can then use this to capture the bounds on the ancillae, where
    \begin{align}
        a_H & = \max(a_{\mathrm{kin}}, a_{\mathrm{pot}}, a_{\mathrm{gse}}) + 2\nonumber \\
        & \in \mathcal{O}\left(\max\left(\log \frac{\alpha_{\mathrm{kin}}}{\varepsilon_{\mathrm{kin}}}, \log \frac{\alpha_{\mathrm{pot}}}{\varepsilon_{\mathrm{pot}}}, \log\frac{\lambda}{\varepsilon_{\tel}} + \widetilde N \log B\right)\right)\nonumber \\
        & \subseteq \mathcal{O}\left(\widetilde{N}\log B + \log\frac{1}{\delta\gamma} + \log\frac{\alpha_H}{\varepsilon_H}\right) \, .
    \end{align}
    Now, for the cost of implementation, we note that we simply have to add the implementation cost of the kinetic, potential, and ground state energy Hamiltonians, plus the cost for the single rotation gate (which is negligible), and adjust for errors, which we can find to be as follows:

    \begin{equation}
        \widetilde{\mathcal O} \Bigg(\frac{\lambda}{\delta \gamma}\log \frac{g\alpha_H}{\varepsilon_H}\left( N + \widetilde N \log B + \log B\log\frac{B g \alpha_H}{\varepsilon_H}\right) \Bigg).
    \end{equation}

    Note that the number of ancillae and the number of gates is largely dominated by the ground state energy Hamiltonian, with the other two only affecting a logarithmic term regarding $g$.
\end{proof}
If we consider only the particle numbers, the spectral gap, approximation ground state preparation overlap, and error precision, the scaling factor can be written as
\begin{equation}
    \widetilde{\mathcal O}  \left(N_{\mathrm{tot}}^2\right).
\end{equation}
while the gate cost can be written as
\begin{align}
     & \widetilde{\mathcal O}  \left(\frac{\widetilde N N_{\mathrm{tot}}}{\delta\gamma} \log \frac{1}{\varepsilon_H}\left( N_{\mathrm{tot}} + \log\frac{1}{\varepsilon_H}\right) \right),\nonumber \\
     & \subseteq\widetilde{\mathcal O}  \left(\frac{\widetilde N N_{\mathrm{tot}}^2}{\delta\gamma} \log^2 \frac{1}{\varepsilon_H}\right).
\end{align}

Lastly, given the construction of the nuclear Hamiltonian $H_{\rm nuc}$, we can construct the internal energy difference between two systems, $A$ and $B$, as utilized in the thermodynamic integration algorithm by a simple LCU. We note here that, in the context of thermodynamic integration, a large portion of the kinetic and potential Hamiltonian terms cancel each other out. A more efficient block encoding with a smaller scaling factor can be constructed by examining the exact LCU terms and canceling identical terms. However, given that the ground state energy Hamiltonian is constructed by preparing the ground state on an electronic register and applying the electronic Hamiltonian, it may be hard to cancel terms out in this block encoding. Given that the scaling factor of the nuclear Hamiltonian $H_{\rm nuc}$ is primarily dominated by the ground state energy Hamiltonian, we note that even if we can cancel out most of the terms in the kinetic and potential operators, the scaling factor would still be in $\mathcal O(\widetilde N N_{\text{tot}})$ with additional terms dependent on the differing terms of the kinetic and potential Hamiltonians. Hence, for simplicity of theoretical constructions, we do not delve into the cancellation of Hamiltonian terms, but instead construct the difference using a simple LCU, noting that in practice, the block encoding implementation can be made more efficient with a lower scaling factor.

It is then easy to see that when we are provided two systems $A$ and $B$, given $H_{\Delta \rm int} = H_{\rm int_{A}}- H_{\rm int_{B}}$, by \cref{lemLCU}, there exists a $(\alpha_{\Delta}, a_{\Delta}, \varepsilon_{\Delta})$-block-encoding of $H_{\Delta\rm int}$ where
\begin{align}
    \alpha_{\Delta} & \in \mathcal O \left(N_{\mathrm{tot},A}^2 + N_{\mathrm{tot},B}^2\right) \\
    a_{\Delta}      & \in
    \begin{multlined}[t]
        \mathcal{O}\bigg(\max(\widetilde{N}_{A}\log B_{A},\widetilde{N}_{B}\log B_{B}) \\
        + \log\frac{1}{\min(\delta_{A}\gamma_{A},\delta_{B}\gamma_{B})} + \log\frac{\alpha_{\Delta}}{\varepsilon_{\Delta}}\bigg)
    \end{multlined}
\end{align}
with cost
\begin{equation}
    \widetilde{\mathcal O}  \left(\left(\frac{\widetilde N_{A} N_{\mathrm{tot},A}^2}{\delta_{A}\gamma_{A}} + \frac{\widetilde N_{B} N_{\mathrm{tot},B}^2}{\delta_{B}\gamma_{B}}\right)\log^2 \frac{1}{\varepsilon_{\Delta}}\right) \, .
\end{equation}

\subsection{The quantum thermodynamic integration algorithm}

Our quantum implementation of the thermodynamic integration algorithm first prepares $\Lambda$-dependent Liouvillian simulations in superposition, such that we have 
\begin{equation}
    \frac{1}{\sqrt{N_\Lambda}}\sum_\Lambda \ket{\Lambda}\ket{\rho_\Lambda} = \frac{1}{\sqrt{N_\Lambda}}\sum_\Lambda \ket{\Lambda}\otimes e^{-iL_\Lambda t_{\rm eq}}\ket{\rho_0}
\end{equation}
by the methods outlined in \cref{appLambdaLiouvillian}. We then apply the controlled version of the block encoding of $\Delta H$ as shown in \cref{appInHam} in the form of a Hadamard test circuit, where we then obtain probability
\begin{equation}
    \mathbb{P}_{\Delta \widetilde F} =\frac{1}{2}\left(\frac{\Delta \widetilde F}{\alpha_{\Delta}} + 1\right) \, ,
\end{equation}
of measuring the zero state on the measurement qubit. One can then use amplitude estimation to obtain an $\varepsilon/6\alpha_{\Delta}$-close estimate of $p_{\Delta \widetilde F}$ with constant success probability. This in turn provides a $\varepsilon/3$-close estimate of $\Delta \widetilde F$, and consequently a $\varepsilon$-close estimate of $\Delta F$ when Riemann sum errors and block encoding errors are sufficiently bounded. We refer the reader to \cref{algoTI} for the complete algorithm description.

We now obtain the formal statement of the thermodynamic integration algorithm. 

\begin{theorem}[Thermodynamic integration]
    Given two systems $A$ and $B$, suppose for both systems, we have a nuclear-position-controlled electronic Hamiltonian that can be block-encoded in the first quantization scheme with scaling factor $\lambda$ satisfying \cref{assumptionGS} in regards to spectral gap lower bound $\gamma$ for both systems and that we are each given an approximate ground state preparation oracle $U_{I,A/B}$ that can prepare ground states up to an overlap of $1-\delta$ for either system. We are given block encodings of discretized Liouvillians $L_A$ and $L_B$ under the NVT canonical ensemble encoding the dynamics for systems $A$ and $B$ per \cref{propLiou}, as well as block encodings of discretized nuclear Hamiltonians $H_A$ and $H_B$ for systems $A$ and $B$ per \cref{propNuclearHam}. Further, we are given an initial state $\ket{\rho_0}$ that encodes the initial discretized phase-space density as well as an equilibration time $t_{\mathrm{eq}}$ for Liouvillian simulation by \cref{thmLiouSim}. There exists a quantum algorithm that computes the Helmholtz free energy difference between the two systems up to additive $\varepsilon$-precision with success probability $1-\xi$ using
    \begin{multline*}
        \widetilde{\mathcal{O}} \bigg( \frac{\alpha_H\alpha_L t_{\mathrm{eq}}\log g}{\varepsilon}\log\frac{1}{\xi}\\
        \times\bigg(\frac{\lambda}{\delta \gamma}(N + \widetilde{N} \log B + \log B \log Bg)+ d\bigg)\bigg)
    \end{multline*}
    Toffoli gates,
    \begin{align*}
        \widetilde{\mathcal{O}} \left(\frac{\alpha_H\alpha_L t_{\mathrm{eq}}}{\delta\varepsilon}\log\frac{1}{\xi}\right)
    \end{align*}
    queries to $U_{I,A}$ and $U_{I,B}$, and
    \begin{align*}
        \widetilde{\mathcal{O}}\left(N \log g + \widetilde{N}\log B + \log\frac{d}{\delta\gamma} + \log\frac{\alpha_L\alpha_H t_{\mathrm{eq}}}{\varepsilon}\right)
    \end{align*}
    qubits, where $N = \max(N_A, N_B)$, $\widetilde N = \max(\widetilde N_A, \widetilde N_B)$, $N_{\mathrm{tot}}=N+\widetilde N$, $\alpha_L:= \max(\alpha_{L_A}, \alpha_{L_B})$ and $\alpha_{L_A}$ and $\alpha_{L_B}$ are as defined for systems $A$ and $B$ in \cref{propLiou}, $\alpha_H:= \max(\alpha_{H_A}, \alpha_{H_B})$ and $\alpha_{H_A}$ and $\alpha_{H_B}$ are as defined for systems $A$ and $B$ in \cref{propNuclearHam},
    $ d := \max(d_x, d_{p'}, d_s, d_{p_s} ) $ and $ g := \max( g_x, g_{p'}, g_s, g_{p_s} )$.
    \label{thmThermInt}
\end{theorem}
\begin{proof}
    We now consider the error bounds of the algorithm. From \cref{thmThermInt}, we can prepare a state that $\varepsilon_L$-close in $\ell_2$ norm to the Liouvillian evolved state $\ket{\rho}$, upon successful post-selection of the ancillae, which requires a constant factor times of reruns. Here, because we need to use the state as input to the Hadamard test, we do not want to perform post-selection. Here we use slightly looser bounds and note that the prepared state $\ket{\sigma}$ (inclusive of the ancilla registers) can be prepared to be close enough to the target state $\ket{\widetilde \rho} = \ket{\rho}\ket{0}_a$ with fidelity up to $1-\varepsilon_L$. Formally, we write
    \begin{equation}
        \lvert \braket{\sigma|\widetilde \rho}\rvert \ge 1-\varepsilon_L\,.
    \end{equation}
    The $\ell_2$-norm difference can then be bounded as follows:
    \begin{equation}
        \lVert \ket{\sigma} - \ket{\widetilde \rho}\rVert \le \sqrt{2 \varepsilon_L}
    \end{equation}
    We note that the above bounds apply to all $\Lambda$-dependent states $\ket{\rho_{\Lambda}}$.
    Then applying the state to the $(\alpha_\Delta, a_\Delta, \varepsilon_\Delta)$-block-encoding $U_\Delta$ of the Hamiltonian difference $\Delta H = H_B - H_A$, we note the total block encoding errors can be bounded as follows:
    \begin{align}
         & \lvert\alpha_\Delta\bra{\sigma}U_\Delta\ket{\rho} - \bra{\sigma}\Delta H\ket{\rho}\rvert\nonumber \\
         & \begin{aligned}
               \le\; & \alpha_\Delta\lvert\bra{\sigma}U_\Delta\ket{\sigma} -  \bra{\sigma}U_\Delta\ket{\widetilde \rho}\rvert \\
               & +\alpha_\Delta\lvert\bra{\sigma}U_\Delta\ket{\widetilde \rho} -  \bra{\widetilde \rho}U_\Delta\ket{\widetilde \rho}\rvert \\
               & +\lvert\alpha_\Delta\bra{\widetilde \rho}U_\Delta\ket{\widetilde \rho} - \bra{\rho}\Delta H\ket{\rho}\rvert
           \end{aligned}\nonumber \\
         & \le 2\alpha_\Delta\sqrt{2\varepsilon_L} + \varepsilon_\Delta
    \end{align}
    To obtain the total error over the state $\frac{1}{\sqrt{N_\Lambda}}\sum_{\Lambda}\ket{\Lambda}\ket{\rho_\Lambda}$, we simply compute the following
    \begin{align}
        &\begin{multlined}
            \Bigg\lvert\frac{\alpha_{\Delta}}{N_{\Lambda}} \left(\sum_{\Lambda} \bra{\Lambda}\bra{\sigma_\Lambda}\right)\mathbb{I} \otimes U_\Delta\left(\sum_{\Lambda} \ket{\Lambda}\ket{\sigma_\Lambda}\right)\\
            -\frac{1}{N_{\Lambda}} \left(\sum_{\Lambda} \bra{\Lambda}\bra{\rho_\Lambda}\right)\mathbb{I}\otimes \Delta H\left(\sum_{\Lambda} \ket{\Lambda}\ket{\rho_\Lambda}\right)\Bigg\rvert
        \end{multlined}\nonumber\\
        & = \frac{1}{N_\Lambda} \left\lvert\sum_\Lambda\alpha_{\Delta}\bra{\sigma_\Lambda}U_\Delta\ket{\sigma_\Lambda} - \bra{\rho_\Lambda}\Delta H\ket{\rho_\Lambda}\right\rvert\nonumber\\
        &\le\frac{1}{N_\Lambda} \sum_\Lambda\left\lvert\alpha_{\Delta}\bra{\sigma_\Lambda}U_\Delta\ket{\sigma_\Lambda} - \bra{\rho_\Lambda}\Delta H\ket{\rho_\Lambda}\right\rvert\nonumber\\
        &\le\frac{1}{N_\Lambda} \sum_\Lambda(2\alpha_\Delta\sqrt{2\varepsilon_L} + \varepsilon_\Delta)\nonumber\\
        &=2\alpha_\Delta\sqrt{2\varepsilon_L} + \varepsilon_\Delta
    \end{align}
    
    From the discretization errors of Riemann sum bounds, we add an additional term of $\varepsilon_{\rm disc}$, which is shown in \cref{appThmIntDiscErr}. From the Hadamard test circuit, we obtain the probability
    \begin{equation}
        \mathbb{P}_{\Delta \widetilde F} =\frac{1}{2}\left(\frac{\Delta \widetilde F}{\alpha_{\Delta}} + 1\right) \, ,
    \end{equation}
    on measurement of the zero state. Viewing the entire circuit (the $\Lambda$-Liouvillian evolution and the Hadamard test) as a quantum state $\ket{\psi}$ we note that 
    \begin{equation}
        \tr\left(\ketbra{\psi}{\psi}(\ketbra{0}{0}\otimes \mathbb I)\right) = \mathbb{P}_{\Delta \widetilde F}
    \end{equation}
    By amplitude estimation~\citep{brassard2002quantum,montanaro2015quantum}, we can output an $\frac{\varepsilon_{\rm qae}}{2\alpha_\Delta}$-close estimation of $p_{\Delta \widetilde F}$ with $\mathcal{O}(\frac{\alpha_\Delta}{\varepsilon_{\rm qae}})$ accesses to $2\ketbra{\psi}{\psi} - \mathbb{I}$ and $2\ketbra{0}{0}\otimes \mathbb I - \mathbb{I}$. Thus, the total error in the estimation of $\Delta F$ is then found to be
    \begin{equation}
    2\alpha_\Delta\sqrt{2\varepsilon_L} + \varepsilon_\Delta + \varepsilon_{\rm disc} + \varepsilon_{\rm qae}
    \end{equation}

\begin{table*}
    \centering
    \begin{tabular}{lccc}
    \toprule
         & Qubit Count & Toffoli Count & Scaling Factor\\
    \midrule
        Alchemical Free Energy [\cref{thmThermInt}] & $\widetilde{\mathcal{O}}(N_{\rm tot})$ & $\displaystyle\widetilde{\mathcal{O}}\left(\frac{N\widetilde N N_{\rm tot}^5t}{\delta\gamma\varepsilon}\log \frac{1}{\xi}\right)$ & --  \\
        \dent{1}Liouvillian Simulation [\cref{thmLiouSim}] & $\widetilde{\mathcal{O}}(N_{\rm tot})$ & $\displaystyle\widetilde{\mathcal{O}}\left(\frac{N\widetilde N N_{\rm tot}^3t}{\delta\gamma}\log^3\frac{1}{\varepsilon}\right)$ & 1\\
        \dent{2}Liouvillian [\cref{propLiou}] & $\widetilde{\mathcal{O}}(N_{\rm tot})$ & $\displaystyle\widetilde{\mathcal{O}}\left(\frac{\widetilde N N_{\rm tot}^2}{\delta\gamma}\log^2\frac{1}{\varepsilon}\right)$ & $\mathcal{O}(NN_{\rm tot})$\\
        \dent{3}Classical Liouvillian~\citep{simon2024improved}  [\cref{lemClassLiou}]& $\widetilde{\mathcal{O}}(N)$ & $\displaystyle\widetilde{\mathcal{O}}\left( N_{\rm tot}\log\frac{1}{\varepsilon}+\log^{\log 3}\frac{1}{\varepsilon}\right)$ & $\mathcal{O}(N^2)$\\
        \dent{3}Electronic Liouvillian [\cref{propElecLiou}] & $\widetilde{\mathcal{O}}(N_{\rm tot})$ & $\displaystyle\widetilde{\mathcal{O}}\left(\frac{\widetilde N N_{\rm tot}^2}{\delta\gamma}\log^2\frac{1}{\varepsilon}\right)$ & $\mathcal{O}(N\widetilde N)$\\
        \dent{4}Ground State Energy Derivative [\cref{lemDel}] & $\widetilde{\mathcal{O}}(\widetilde N)$ & $\displaystyle\widetilde{\mathcal{O}}\left(\frac{\widetilde N N_{\rm tot}^2}{\delta\gamma}\log^2\frac{1}{\varepsilon}\right)$ & $\mathcal{O}(\widetilde N)$\\
        \dent{5}Ground State~\citep{lin2020nearoptimal} [\cref{lemGSPrep}] & $\widetilde{\mathcal{O}}(\widetilde N)$ & $\displaystyle\widetilde{\mathcal{O}}\left(\frac{\widetilde N N_{\rm tot}^2}{\delta\gamma}\log^2\frac{1}{\varepsilon}\right)$ & 1\\
        \dent{5}Force Operator~\citep{obrien2022efficient} [\cref{lemHamDeriv}] & $\widetilde{\mathcal{O}}(\widetilde N)$ & $\displaystyle\widetilde{\mathcal{O}}\left(\widetilde N+\log\frac{1}{\varepsilon}\right)$ & $\mathcal{O}(\widetilde N)$\\
        \dent{1}Nuclear Hamiltonian [\cref{propNuclearHam}] & $\widetilde{\mathcal{O}}(N_{\rm tot})$ & $\displaystyle\widetilde{\mathcal{O}}\left(\frac{\widetilde N N_{\rm tot}^2}{\delta\gamma}\log^2\frac{1}{\varepsilon}\right)$ & $\mathcal{O}(N_{\rm tot}^2)$\\
        \dent{2}Ground State Energy Hamiltonian [\cref{lemGSE}] & $\widetilde{\mathcal{O}}(N_{\rm tot})$ & $\displaystyle\widetilde{\mathcal{O}}\left(\frac{\widetilde N N_{\rm tot}^2}{\delta\gamma}\log^2\frac{1}{\varepsilon}\right)$ & $\mathcal{O}(\widetilde NN_{\rm tot})$\\
    \bottomrule
    \end{tabular}
    \caption{We summarize the simplified qubit count and Toffoli gate counts corresponding to the block encoding components that make up the Liouvillian simulation and alchemical free energy algorithms in our paper.}
    \label{tabOurRuntime}
\end{table*}

\begin{table*}
    \centering
    \begin{tabular}{lcc}
    \toprule
     & Qubit Count & Toffoli Count \\
    \midrule
        Absolute Free Energy & $\widetilde{\mathcal{O}}(N_{\rm tot})$ & $\displaystyle\mathcal{O}^*\left( \left( \frac{\eta^{o(1)} N^2 \widetilde N^2 N_{\mathrm{tot}}^3 t}{\delta\gamma \, \varepsilon} \left( N_{\mathrm{tot}}^2 + \frac{\eta}{\sqrt{\varepsilon}} \right) + \frac{\widetilde N N_{\text{tot}}^4 }{\varepsilon^2}\right) \log \frac{1}{\xi}\right)$ \\
        \dent{1}Internal Energy & $\widetilde{\mathcal{O}}(N_{\rm tot})$ & $\displaystyle\mathcal{O}^*\left( \left( \frac{\eta^{o(1)} N^2 \widetilde N^2 N_{\mathrm{tot}}^5 t}{\delta\gamma \, \varepsilon} + \frac{\widetilde N N_{\text{tot}}^4 }{\varepsilon^2}\right) \log \frac{1}{\xi}\right)$ \\
        \dent{2}Trotterized Liouvillian Simulation & $\widetilde{\mathcal{O}}(N_{\rm tot})$ & $\displaystyle\mathcal{O}^*\left(\frac{N^2 \widetilde N^2 N_{\mathrm{tot}}^3  t}{\delta\gamma\varepsilon^{o(1)}}\right)$ \\
        \dent{3}Classical Liouvillian & $\widetilde{\mathcal{O}}(N)$ & $\displaystyle\widetilde{\mathcal{O}}\left( N_{\rm tot}\log\frac{1}{\varepsilon}+\log^{\log 3}\frac{1}{\varepsilon}\right)$ \\
        \dent{3}Electronic Liouvillian Evolution & $\widetilde{\mathcal{O}}(N_{\rm tot})$ & $\displaystyle\widetilde{\mathcal{O}}\left( N\widetilde NN_{\rm tot}^2t\log^2\frac{1}{\varepsilon} + \frac{N\widetilde N N_{\rm tot}^2}{\delta\gamma}\log\frac{1}{\varepsilon}\right)$ \\
        \dent{4}\makecell[l]{Controlled Electronic\\ Hamiltonian Simulation} & $\widetilde{\mathcal{O}}(N_{\rm tot})$ & $\displaystyle\widetilde{\mathcal{O}}\left( \widetilde NN_{\rm tot}^2t\log^2\frac{1}{\varepsilon}\right)$ \\    
        \dent{2}Nuclear Hamiltonian & $\widetilde{\mathcal{O}}(N_{\rm tot})$ & $\displaystyle\widetilde{\mathcal{O}}\left(\widetilde N N_{\rm tot}^2\left(\frac{1}{\delta\gamma}+\frac{1}{\varepsilon}\right)\right)$ \\
        \dent{3}Ground State Energy Hamiltonian & $\widetilde{\mathcal{O}}(N_{\rm tot})$ & $\displaystyle\widetilde{\mathcal{O}}\left(\widetilde N N_{\rm tot}^2\left(\frac{1}{\delta\gamma}+\frac{1}{\varepsilon}\right)\right)$ \\
        \dent{1}Gibbs Entropy & $\widetilde{\mathcal{O}}(N_{\rm tot})$ & $\displaystyle\mathcal{O}^*\left( \left( \frac{\eta N^2 \widetilde N^2 N_{\mathrm{tot}}^3 t}{\delta\gamma \, \varepsilon^{1.5}}\right) \log \frac{1}{\xi}\right)$\\

    \bottomrule
    \end{tabular}
    \caption{We summarize the simplified qubit count and Toffoli gate counts corresponding to the block encoding components that make up the Liouvillian simulation and free energy algorithms by \citet{simon2024improved} to provide a comparison to our paper. We note that the electronic Liouvillian evolution algorithm provided by \citet{simon2024improved} is in fact more efficient than block encoding the electronic Liouvillian and applying QSVT directly, due to the decoupling of the ground state preparation and the time evolution. This advantage would be lost once the Trotter iterates between classical and electronic Liouvillian evolution is brought in, as the number of ground state preparations is then multiplied by the Liouvillian spectral norm and total evolution time.}
    \label{tabPRXRuntime}
\end{table*}
    
    Upper-bounding the total error by $\varepsilon$, we note that the error terms are bounded as follows:
    \begin{align}
        \varepsilon_L         & \in \mathcal{O} \left(\frac{\varepsilon^2}{\alpha_\Delta^2}\right) \\
        \varepsilon_\Delta,\varepsilon_{\rm qae}, \varepsilon_{\rm disc} & \in \mathcal {O} (\varepsilon)
    \end{align}
    Plugging these results into the runtime, we note that the cost of block-encoding the Hamiltonian in \cref{propNuclearHam} is overshadowed by the cost of block-encoding the $\Lambda$-dependent Liouvillian evolution in \cref{thmLiouSim}. Then, factoring the cost of amplitude estimation to precision $\varepsilon /\alpha_\Delta$, we multiply the cost of Liouvillian simulation by the inverse of the precision, while noting that $\alpha_\Delta \in \mathcal{O} (\alpha_H)$. This, in turn, overshadows the polylogarithmic error terms in the Liouvillian simulation. Lastly, we include the additional cost of using median boosting to raise the success probability to $1-\xi$, which yields the result in the theorem. The number of qubits can be found by plugging in the error bounds of $\varepsilon_L$ as provided above to \cref{thmLiouSim}.
\end{proof}
Note that if we consider only the particle numbers, the spectral gap, approximation ground state preparation overlap, and error precision, the runtime can be written as
\begin{align}
    \widetilde{\mathcal{O}}\left(\frac{N \widetilde N N_{\mathrm{tot}}^5 t_{\mathrm{eq}}}{\delta\gamma\varepsilon}\log\frac{1}{\xi}\right)
\end{align}
as shown in the informal version of the theorem (\cref{thmThermIntInfml}) in the main text.

Lastly, we provide a summarized list of runtime costs for the components that make up our Liouvillian simulation algorithm and free energy calculation algorithm in \cref{tabOurRuntime}. For comparisons purpose, we also provide a similar summary of \citet{simon2024improved}'s approach in \cref{tabPRXRuntime}.

\section{Discretization errors of the Liouvillian}
\label{appDiscErr}

We provide a qualitative discussion of the results of errors in the discretization of the Liouvillian. From \cref{thmLiouSim,thmThermInt}, we note that the number of grid points in the nuclear registers $g$ and the number of plane wave bases in the electronic register $B$ are both only logarithmically dependent within the runtime. Hence, one can potentially decrease the error caused by discretization by increasing $g$ and $B$ to a large number without making polynomial increases to the runtime.

Regarding approximating the derivative of position and momentum via central finite difference to order $d$, we note that the discretization error of a function $f$ can be upper bounded~\citep{li2005general} by
\begin{equation}
    \varepsilon \in \mathcal O \left(\max_{\xi \in [a, b]} \left|f^{(2d+1)}(\xi)\right| \left(\frac{eh}{2}\right)^{2d}\right) \, ,
\end{equation}
for which $d \in \mathcal O (\log \frac{1}{\varepsilon})$ if $h$ is constant. Thus, given that the dependency of the central finite difference orders $d$ in the runtime of \cref{thmLiouSim,thmThermInt} is only linear, the error stemming from this discretization error can potentially be suppressed, provided that there exists a constant $\mathcal{C}$ such that
\begin{equation}
    \mathcal{C} \ge \sup_{d\ge 0} \max_{\xi \in [a,b]}\left|f^{(2d+1)}(\xi) \right|^{\frac{1}{2d}} \, .
\end{equation}
However, this assumption is, in fact, difficult to validate in general, as the maximum scale of the higher-order derivatives is difficult to bound in full generality as the continuum limit is approached. This discretization issue remains an important consideration within the field~\citep{leveque2007finite}, and the development of alternative differential operators that do not rely on central finite difference would serve as valuable material for future work.
\end{document}